\documentclass[graybox]{svmult}

\usepackage[square,numbers,sort&compress]{natbib}

\usepackage{type1cm}        
\usepackage{makeidx}         
\usepackage{graphicx}        
\usepackage{multicol}        
\usepackage[bottom]{footmisc}

\usepackage{newtxtext}       %
\usepackage{newtxmath}       

\usepackage{amsmath}
\usepackage{amsfonts}
\usepackage{bbm,float,graphics}

\usepackage{hyperref}

\usepackage{etex, xy}
\xyoption{all}

\usepackage{alphabeta}

\usepackage[colorinlistoftodos]{todonotes}

\newcommand{\lr}[1]{\langle#1\rangle}

\let\set\mathbbm

\newcommand{\ssumB}[2]{\textstyle\sum\limits_{\scriptscriptstyle#1}^{\scriptscriptstyle#2}}

\newcommand{\Sum}{\text{Sum}}
\newcommand{\Prod}{\text{Prod}}
\newcommand{\RPow}{\text{RPow}}
\newcommand{\cRR}{{\cal R}}

\makeatletter 
\def\newdef#1{\newtheorem{@#1}[Theorem]{{\it #1}}%
	\newenvironment{#1}{\begin{@#1}\em}{\end{@#1}}}
\makeatother

\newcommand{\shiftS}{{S}}

\newcommand\ToDo[1][x]{\fbox{\textbf{TODO} \ifx#1x\fi}}

\newcommand{\depth}{\texttt{d}}

\newcommand{\vect}[1]{\boldsymbol{#1}}

\newcommand{\AR}{\set A}

\newcommand{\QQ}{\set Q}
\newcommand{\ZZ}{\set Z}
\newcommand{\NN}{{\set Z}_{\geq0}}
\newcommand{\NNP}{{\set Z}_{\geq1}}
\newcommand{\KK}{\set K}
\newcommand{\HH}{\set H}
\newcommand{\GG}{\set G}
\newcommand{\FF}{\set F}

\newcommand{\EE}{\set E}
\newcommand{\iiota}{\set i}

\newcommand{\calA}{{\cal A}}

\newcommand{\mGG}{\GG_{m}}
\newcommand{\bGG}{\GG_{b}}
\newcommand{\rGG}{\GG_r}
\newcommand{\expr}{\text{expr}}

\newcommand{\fct}[3]{#1\colon #2 \to #3}
\newcommand{\dfield}[2]{({#1},{#2})}
\newcommand{\const}[2]{{\rm const}_{#2}{#1}}

\newcommand{\sigmaE}{$\Sigma$}
\newcommand{\piE}{$\Pi$}

\newcommand{\rE}{$R$}

\newcommand{\pisiE}{$\Pi\Sigma$}
\newcommand{\rpisiE}{$R\Pi\Sigma$}
\newcommand{\rpiE}{$R\Pi$}

\newcommand{\ord}{\text{ord}}

\newcommand{\seqK}{S(\KK)}
\newcommand{\seqQ}{\textbf{S}(\QQ)}

\newcommand{\id}{\operatorname{id}}

\newcommand{\ev}{\operatorname{ev}}

\allowdisplaybreaks[3]

\usepackage{stackengine,wasysym} 
\newcommand\opowerAlone{\!\stackMath\mathbin{\stackinset{c}{0ex}{c}{0ex}{\text{\scriptsize$\wedge$}}{\Circle}}}
\newcommand\opower{\!^{\stackMath\mathbin{\stackinset{c}{0ex}{c}{0ex}{\text{\scriptsize$\wedge$}}{\Circle}}}\!}

\makeindex             

\definecolor{blaugrau}{rgb}{0.796887, 0.789075, 0.871107}
\newcounter{mmacnt}
\def\restartmma{\setcounter{mmacnt}{0}}
\restartmma \catcode`|=\active
\def|#1|{\mathrm{#1}}
\catcode`|=12
\newenvironment{mma}{
	\par
	\catcode`|=\active
	\parskip=6pt\parindent=0pt 
	\small
	\def\In##1\\{%
		\def\linebreak{\hfill\break\null\qquad}%
		\refstepcounter{mmacnt}
		\hangindent=2.5em\hangafter=0
		\leavevmode
		\llap{\tiny\sffamily In[\arabic{mmacnt}]:=\kern.5em}%
		\mathversion{bold}\scriptsize$\tt\bf\displaystyle##1$\normalsize
		\mathversion{normal}\par
	}%
	\def\Print##1\\{%
		\def\linebreak{\hfill\break}%
		\hangindent=2.5em\hangafter=0
		\leavevmode\scriptsize ##1\par}%
	\def\Out##1\\{%
		\vspace*{-0.5cm}\def\linebreak{$\hfill\break\null\hfill$}%
		\kern\abovedisplayskip\par
		\hangindent=2.5em\hangafter=0
		\leavevmode
		\llap{\tiny\sffamily Out[\arabic{mmacnt}]=\kern.5em}
		\scriptsize$\displaystyle\tt##1$\normalsize\hfill\null\par
		\kern\belowdisplayskip\vspace*{-0.3cm}
	}%
	\def\Warning##1##2\\{%
		\def\linebreak{\hfill\break}%
		\hangindent=2.5em\hangafter=0
		\leavevmode
		{\scriptsize##1 : ##2}\par}%
}{%
	\par\smallskip
}

\newcommand{\LoadP}[1]{\fcolorbox{black}{blaugrau}{
		\begin{minipage}[t]{10.5cm}
			\scriptsize\normalfont #1
\end{minipage}}}

\newcommand{\myIn}[1]{{\sffamily In[#1]}}
\newcommand{\myOut}[1]{{\sffamily Out[#1]}}

\def\MLabel#1{{\refstepcounter{mmacnt}\label{#1}}\addtocounter{mmacnt}{-1}}
\newcommand{\MText}[1]{\textbf{\small\ttfamily#1}}

\begin{document}

\title*{Term Algebras, Canonical Representations and Difference Ring Theory for Symbolic Summation}
\titlerunning{Term Algebras, Canonical Representations and Difference Ring Theory}
\author{Carsten Schneider\thanks{This project has received funding from the European Union's Horizon 2020 research and innovation programme under the Marie Sk\l{}odowska--Curie grant agreement No. 764850, SAGEX and from the Austrian Science Foundation (FWF) 
		grant SFB F50 (F5005-N15) in the framework of the
		Special Research Program
		``Algorithmic and Enumerative Combinatorics''.}}
\institute{
	Carsten Schneider \at Johannes Kepler University Linz, Research Institute for Symbolic Computation, A-4040 Linz, Austria,
	\email{Carsten.Schneider@risc.jku.at}}
%
%

\maketitle

\abstract{A general overview of the existing difference ring theory for symbolic summation is given. Special emphasis is put on the user interface: the translation and back translation of the corresponding representations within the term algebra and the formal difference ring setting. In particular, canonical (unique) representations and their refinements in the introduced term algebra are explored by utilizing the available difference ring theory. Based on that, precise input-output specifications of the available tools of the summation package \texttt{Sigma} are provided.}

\section{Introduction}\label{Sec:Introduction}

In the last 40 years exciting results have been accomplished in symbolic summation as elaborated, e.g., in~\cite{Abramov:71,Abramov:75,Gosper:78,Karr:81,Karr:85,Zeilberger:91,PauleSchorn:95,Paule:95,AequalB,CJKS:13,Hou:11,vanHoeij:99,Abramov:94,Abramov:96,Singer:99,PauleRiese:97,APP:98,Bauer:99,Zeilberger:90a,Chyzak:00,Koutschan:13,Wilf:92,Wegschaider,AZ:06,KauersStirling:07,KS:08,CKS:09,Bron:00,Schneider:00,Schneider:01,Schneider:08c,DR1,SchneiderISSAC:16,DR2,DR3,Petkov:18,Paule:21} that will be sketched in more details below. In most cases, symbolic summation can be subsumed by the following problem description: given an algorithm that computes/represents a sequence, find a simpler algorithm that computes/represents (from a certain point on) the same sequence. Based on the context of a given problem, simpler can have different meanings: e.g., the output algorithm can be represented uniquely (by a canonical form in the sense of~\cite{Buchberger:83}), it might be computed more efficiently, or it can be formulated in terms of certain classes of special functions.
  
Often symbolic summation is subdivided in the following summation paradigms.

\noindent$\bullet$\hspace*{0.1cm}\textit{Telescoping:} Given an algorithm $F(k)$ that computes a sequence, find an algorithm $G(k)$, that is not more complicated than $F(k)$, such that
\begin{equation}\label{Equ:Tele}
F(k)=G(k+1)-G(k)
\end{equation}
holds for all $k\in\NN$ with $k\geq\delta$ for some $\delta\in\NN$. Then summing this equation over $k$ from $\delta$ to $n$ yields a simpler way to compute $S(n)=\sum_{k=\delta}^nF(k)$, namely
\begin{equation}\label{Equ:TeleSummed}
\sum_{k=\delta}^nF(k)=G(n+1)-G(\delta).
\end{equation}
\noindent$\bullet$\hspace*{0.1cm}\textit{Zeilberger's creative telescoping~\cite{Zeilberger:91}:} Given an algorithm $F(n,k)$ that computes a bivariate sequence, find an algorithm $G(n,k)$ that is not more complicated than $F(n,k)$, and algorithms $c_0(n),\dots,c_d(n)$ (for univariate sequences), such that
\begin{equation}\label{Equ:Crea}
c_0(n)\,F(n,k)+c_1(n)\,F(n+1,k)+\dots+c_d(n)\,F(n+d,k)=G(n,k+1)-G(n,k)
\end{equation}
holds for all $n,k\in\NN$ with $n,k\geq\delta$ for some $\delta\in\NN$. Then summing this equation over $k$ from $\delta$ to $n$ yields for the definite sum $S(n)=\sum_{k=\delta}^nF(n,k)$ the recurrence
\begin{equation}\label{Equ:RecGeneral}
c_0(n)\,S(n)+c_1(n)\,S(n+1)+\dots+c_d(n)\,S(n+d)=H(n)
\end{equation}
with $H(n)=G(n,n+1)-G(n,\delta)+\sum_{i=1}^dc_i(n)\sum_{j=1}^iF(n+i,n+j)$. In many cases $H(n)$ collapses to a rather simple ``algorithm'' and thus~\eqref{Equ:RecGeneral} yields (together with $d$ initial values and the assumption that $c_d(n)$ is nonzero for $n\geq\delta$) an efficient algorithm to compute the sequence $(S(n))_{n\geq\delta}$. 

\smallskip

\noindent$\bullet$\hspace*{0.1cm}\textit{Recurrence solving:} Given a recurrence of the form~\eqref{Equ:RecGeneral} where the algorithms $c_0(n),\dots,c_d(n)$ and $H(n)$ can be given by expressions in terms of certain classes of special functions (that can be evaluated accordingly) and given $d$ initial values, say $S(\delta),S(\delta+1),\dots,S(\delta+d-1)$ which determines the sequence $(S(n))_{n\geq\delta}$, find an expression that computes the sequence $(S(n))_{n\geq\delta}$ in terms of the same class of special functions or an appropriate extension of it.

\smallskip

We emphasize that all of the above summation paradigms are strongly interwoven (as illustrated, e.g., in the book~\cite{AequalB}) and they often yield a strong toolbox by combining them in a nontrivial way.

\medskip

Another natural classification of symbolic summation is based on the input class of algorithms and the focus how they can be formally represented. In most cases they are either given by evaluable expressions in terms of sums/products or linear recurrences accompanied with initial values that uniquely determine/enable one to calculate the underlying sequences.
The first breakthrough in this regard has been achieved by Abramov~\cite{Abramov:71,Abramov:75} who solved the telescoping problem for a rational function $F(x)\in\KK(x)$ and proposed an algorithm for finding all rational solutions of $\KK(x)$ of a given linear recurrence of the form~\eqref{Equ:RecGeneral} with $c_i(x),H(x)\in\KK(x)$.  In particular, Gosper's telescoping algorithm~\cite{Gosper:78} for hypergeometric products $F(n)=\prod_{k=l}^nH(k)$ with $H(x)\in\KK(x)$ and Zeilberger's extension to definite sums via his creative telescoping paradigm~\cite{Zeilberger:91,PauleSchorn:95,Paule:95,AequalB,Hou:11,CJKS:13} made symbolic summation highly popular in many areas of sciences; recently also the treatment of contiguous relations has been extensively explored in~\cite{Paule:21}. In particular, the interplay with Petkov{\v{s}}ek's algorithm Hyper~\cite{Petkov:92} or van Hoeij's improvements~\cite{vanHoeij:99} to find all hypergeometric product solutions enables one to simplify definite hypergeometric products to expressions given in terms of hypergeometric products; first methods are on the way to find even definite sum solutions~\cite{Petkov:18}. More generally, one can use these solvers as subroutines to hunt for all d'Alembertian solutions~\cite{Abramov:94,Abramov:96} (solutions that are expressible in terms of indefinite nested sums defined over hypergeometric products) and Liouvillian solutions~\cite{Singer:99,Petkov:2013} (incorporating in addition the interlacing operator).
This successful story has been pushed forward for indefinite and definite summation problems in terms of $q$-hypergeometric products and their mixed version~\cite{PauleRiese:97,APP:98,Bauer:99}. Further generalizations opened up substantially the class of applications, like the holonomic approach~\cite{Zeilberger:90a,Chyzak:00,Koutschan:13} dealing with objects that can be described by recurrence systems or the multi-summation approach of \hbox{($q$--)}hypergeometric products~\cite{Wilf:92,Wegschaider,AZ:06}. Even non-holonomic summation problems~\cite{KauersStirling:07,KS:08,CKS:09} involving, e.g., Stirling numbers, can be treated nowadays automatically.

In the following we will focus on the difference ring/field approach. It has been initiated by Karr's telescoping algorithm \cite{Karr:81,Karr:85} in \pisiE-fields which can be considered as the discrete analog of Risch's indefinite integration algorithm~\cite{Risch:69,Bron:97}. This pioneering work has been explored further in~\cite{Bron:00,Schneider:00,Schneider:01,Schneider:08c} and has been pushed forward to a general summation theory in the setting of \rpisiE-ring extensions~\cite{DR1,SchneiderISSAC:16,DR2,DR3} which is the driving engine of the summation package \texttt{Sigma}~\cite{Schneider:07a,Schneider:13a}.
In this setting, one can deal not only with expressions containing \hbox{($q$--)}hypergeometric products and their mixed versions, but also with those containing sums and products that are indefinite nested (that, depending on the ring or field setting, can appear also in the denominator). In particular, it covers a significant class of special functions that arise frequently, e.g., within the calculation of (massive) 2-loop and 3-loop Feynman integrals: harmonic sums~\cite{Bluemlein:99,Vermaseren:99}, generalized harmonic sums~\cite{Moch:02,ABS:13}, cyclotomic sums~\cite{ABS:11} and binomial sums~\cite{Davydychev:2003mv,Weinzierl:2004bn,ABRS:14}.

Internally, the following construction is performed in \texttt{Sigma}.
\begin{enumerate}\label{BasicSigmaStrategy}
	\item Rephrase the expression in terms of nested sums and products in an appropriate difference ring (built by \pisiE-field and \rpisiE-ring extensions).
	\item Solve the summation problems (given above) in this formal difference ring.
	\item Translate the obtained solution from the difference ring to the term algebra setting.
\end{enumerate}
The goal of this article is two-fold. First, we will present the existing algorithms in the difference ring setting (step~2) that have been implemented in large part within \texttt{Sigma}. In particular, we will summarize the available parameterized telescoping algorithms~\cite{Schneider:04a,Schneider:05f,Schneider:07d,Schneider:08c,Schneider:10a,Schneider:10b,Schneider:10c,Schneider:15}  (containing telescoping/creative telescoping as special cases), the multiplicative version of telescoping for the representation of products~\cite{Schneider:05c,Petkov:10,ZimingLi:11,DR2,OS:18,SchneiderProd:20,OS:20}
and recurrence solving algorithms~\cite{Schneider:01,Bron:00,Schneider:04b,Schneider:05a,Schneider:05b,MS:2018,ABPS:20} which generalize many contributions of the literature mentioned above. In addition, we will comment on further enhancements in order to treat new classes of summation objects, like unspecified sequences~\cite{Schneider:06d,Schneider:06e,PS:19} and radical objects~\cite{Schneider:07f}, or to combine the difference field/ring and holonomic approaches yielding a new toolbox for multi-summation~\cite{Schneider:05d,BRS:16}.

Besides these difference ring algorithms and the underlying difference ring theory (step 2), the translation mechanism between the summation objects and the formal representation (step 1 and 3) will be elaborated in detail. In particular, the summation package \texttt{Sigma} benefits strongly on this stable toolbox: the user can define expressions in terms of symbolic sums and products in a term algebra and obtains simplifications of the expressions by executing the rather technical difference ring/field machinery in the background. However, rigorous input/output specifications on the sum-product level are missing: many of the properties that one can extract on the formal level (step~2) are not properly carried over to the user level. The second main result of the article is a contribution towards closing this gap. In particular, inspired by~\cite{PauleNemes:97} and utilizing ideas from~\cite{Singer:97,Schneider:10c,DR2} we will show that the difference ring theory implies a canonical simplification in the sense of~\cite{Buchberger:83}. We can write the sums and products in a $\sigma$-reduced basis (see Definition~\ref{Def:SigmaReducedSet}) such that two expressions evaluate to the same sequence iff they are syntactically equal. 

In Section~\ref{Sec:TermAlgebra} we will define a term algebra in which we will represent our sequences in terms of indefinite nested sums and products. In particular, we will introduce one of the main features of~\texttt{Sigma} given in Problem~\texttt{SigmaReduce}: one can represent the expressions of our term algebra in canonical form. In Section~\ref{Sec:GeneralDRApproach} we will elaborate how this distinguished representation can be accomplished by exploiting the difference ring theory of \rpisiE-extensions. 
Here we will utilize the interplay (see Figure~\ref{Fig:UserInterface}) between the difference ring of sequences, 
the term algebra (equipped with an evaluation function) in which the sequences can be introduced by the user and the formal difference ring setting (also equipped with an evaluation function) in which the sequences can be modeled on the computer algebra level.
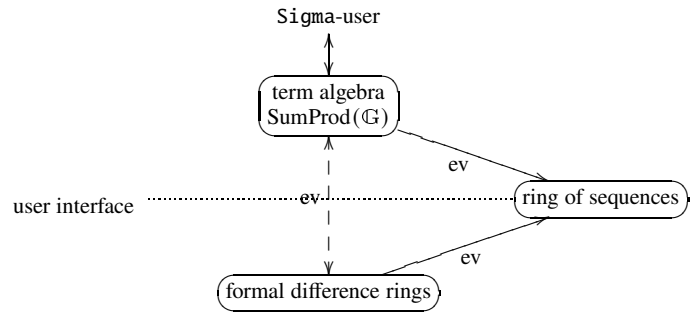
\begin{figure}

\vspace*{-0.4cm}

	$$
	\xymatrix@!R=0.4cm@C1cm{
		&\txt{\texttt{Sigma}-user}\ar@{<->}[d]&\\
		&*+[F-:<6pt>]{\txt{term algebra\\ $\Sum\Prod(\GG)$}}\ar[rd]_{\txt{ev}}\ar@{<-->}[dd]_{\txt{ev}}\\
		\txt{user interface}&\ar@{.}[]+<-8em,0em>;[r]+<-4em,0em>&*+[F-:<6pt>]{\txt{ring of sequences}}\\
		&*+[F-:<6pt>]{\txt{formal difference rings}}\ar[ru]_{\txt{ev}}
	}
	$$
	\caption{The symbolic summation framework for difference rings and fields}\label{Fig:UserInterface}

\vspace*{-0.4cm}

\end{figure}
In Section~\ref{Sec:RepProblem} we will make this construction precise by using the existing difference ring algorithms. In particular, we will concentrate on refined simplifications, like finding expressions with minimal nesting depth. Finally, we are in the position to specify in Section~\ref{Sec:DefiniteProblem} the above introduced summation paradigms of \texttt{Sigma} within the term algebra level. 
In Section~\ref{Sec:Applications} we present the main applications of the presented algorithms that support the evaluation of Feynman integrals. We conclude the article in Section~\ref{Sec:Conclusion}.

\section{The term algebra $\Sum\Prod(\GG)$}\label{Sec:TermAlgebra}

Inspired by~\cite{PauleNemes:97} we will refine the construction from~\cite{Schneider:10b} to introduce a term algebra
for a big class of indefinite nested sums and products. 

The basis of our construction (see also~\cite{Bauer:99}) will be the rational function field extension $\KK=K(q_1,\dots,q_v)$ over a field $K$ and on top of it the rational function field extension $\mGG:=\KK(x,x_1,\dots,x_v)$ over $\KK$. For any element $f=\frac{p}{q}\in \mGG$ with $p,q\in\KK[x,x_1,\dots,x_v]$ where $q\neq0$ and $p,q$ being coprime we define
	\begin{equation}\label{Equ:EvalMixed}
	\ev(f,k)=\begin{cases}
	0&\text{if }q(k,q_1^k,\dots,q_v^k)=0\\
	\frac{p(k,q_1^k,\dots,q_v^k)}{q(k,q_1^k,\dots,q_v^k)}&\text{if }q(k,q_1^k,\dots,q_v^k)\neq0.
	\end{cases}
	\end{equation}
	Note that there is a $\delta\in\NN$ with $q(k,q_1^k,\dots,q_v^k)\neq0$ for all $k\in\NN$ with $k\geq\delta$; for an algorithm that determines $\delta$ if one can factorize polynomials over $K$ see~\cite[Sec.~3.2]{Bauer:99}. We define $L(f)$ to be the minimal value $\delta\in\NN$ such that $q(k,q_1^k,\dots,q_v^k)\neq0$ holds for all $k\geq\delta$; further, we define $Z(f)=\max(L(1/p),L(1/q))$ for $f\neq0$. 
	Later we will call $L:\mGG\to\NN$ also an \emph{$o$-function} and\footnote{For a ring $\AR$ we denote by $\AR^*$ the set of units. If $\AR$ is a field, this means $\AR^*=\AR\setminus\{0\}$.} $Z:\mGG^*\to\NN$ a \emph{$z$-function}. $\mGG=\KK(x,x_1,\dots,x_v)$ represents the \emph{multibasic mixed sequences}. 	
	The special cases $\rGG=\KK(x)$ and $\bGG=\KK(x_1,\dots,x_v)$ represent the \emph{rational} and the \emph{multi-basic sequences}, respectively. If not specified further, $\GG$ will stand for one of the three cases $\mGG$, $\rGG$ or $\bGG$. 

Now we extend $\GG$ to expressions $\Sum\Prod(\GG)$ in terms of indefinite nested sums defined over indefinite nested products. 
For the set of nontrivial roots of unity 
$$\cRR=\{r\in\KK\setminus\{1\}\mid r\text{ is a root of unity}\}$$
we introduce the function $\ord:\cRR\to\NNP$ with 
$$\ord(r)=\min\{n\in\NNP\mid r^n=1\}.$$
Let $\opowerAlone$, $\oplus$, $\odot$, $\Sum$, $\Prod$ and $\RPow$ be operations with the signatures
$$\begin{array}[t]{llcl}
\opowerAlone: &\Sum\Prod(\GG)\times\ZZ&\rightarrow&\Sum\Prod(\GG)\\
\oplus: &\Sum\Prod(\GG)\times\Sum\Prod(\GG)&\rightarrow&\Sum\Prod(\GG)\\
\odot: &\Sum\Prod(\GG)\times\Sum\Prod(\GG)&\rightarrow&\Sum\Prod(\GG)\\
\Sum: &\NN\times\Sum\Prod(\GG)&\rightarrow&\Sum\Prod(\GG)\\
\Prod: &\NN\times\Sum\Prod(\GG)&\rightarrow&\Sum\Prod(\GG)\\
\RPow: &\cRR&\rightarrow&\Sum\Prod(\GG).
\end{array}$$
In the following we write $\opower$, $\oplus$ and $\odot$ in infix notation, and $\Sum$ and $\Prod$ in prefix notation.
Further, for $(\dots((f_1\text{{\tiny $\square$}} f_2)\text{{\tiny $\square$}} f_3)\text{{\tiny $\square$}}\dots\text{{\tiny $\square$}} f_r)$ with $\text{{\tiny $\square$}}\in\{\odot,\oplus\}$ and $f_1,\dots,f_r\in\Sum\Prod(\GG)$ we write $f_1\text{{\tiny $\square$}} f_2\text{{\tiny $\square$}} f_3\text{{\tiny $\square$}} \dots\text{{\tiny $\square$}} f_r$.

More precisely, we define the following chain of set inclusions:
\begin{equation}\label{Equ:ProdSumInclusion}
\begin{array}{ccccccccl}
&&&&\Prod_1(\GG)&\subset&\Sum\Prod_1(\GG)&&\txt{\small expressions with\\ \small single nested products}\\[-0.1cm]
&&&&\cap&&\cap\\[0.1cm]
&&\Prod^*(\GG)&\subset&\Prod(\GG)&\subset&\Sum\Prod(\GG)&&\txt{\small expressions\\ \small with nested products}\\[0.2cm]
&&\txt{\small power products\\ \small in products}&&\txt{\small expressions\\ \small in products}&&\txt{\small expressions in\\ \small \small sums and products.}
\end{array}
\end{equation}
\noindent Here we start with the \emph{set of power products of nested products $\Prod^*(\GG)$} which is the smallest set that contains $1$ 
with the following properties:
\begin{enumerate}
\item If $r\in\cRR$ then $\RPow(r)\in\Prod^*(\GG)$.
\item If $p\in\Prod^*(\GG)$, $f\in\GG^*$, $l\in\NN$ with $l\geq Z(f)$ then\footnote{We also write $p$ instead of $f\odot p$ if $f=1$; similarly we write $f$ instead of $f\odot p$ if $p=1$.} $\Prod(l,\!f\!\odot\! p)\in\Prod^*(\GG)$.
\item If $p,q\in\Prod^*(\GG)$ then $p\odot q\in\Prod^*(\GG)$.
\item If $p\in\Prod^*(\GG)$ and $z\in\ZZ\setminus\{0\}$ then $p\opower z\in\Prod^*(\GG)$.
\end{enumerate}
Later we will also use the sets 
\begin{align*}
\Pi(\GG)=&\{\RPow(r)\mid r\in\cRR\}\cup\{\Prod(l,f\odot p)\mid l,f,p\text{ as given in item 2}\}\\
\Pi_1(\GG)=&\{\RPow(r)\mid r\in\cRR\}\cup\{\Prod(l,f)\mid f\in \GG^*, l\in\NN\text{ with }l\geq Z(f)\}
\end{align*}
where $\Pi(\GG)$ and $\Pi_1(\GG)$ contains all nested and single nested products, respectively.

\begin{example}\label{Exp:NestedProdExpr}
In $\Prod^*(\GG)$ with $\GG=\QQ(q_1)(x,x_1)$ we get, e.g.,
$$P=(\underbrace{\Prod(1,\Prod(1,x)\opower(-2))}_{\in\Pi(\GG)}\opower2)\odot\underbrace{\Prod(1,\tfrac{x_1+x_1^2}{x})}_{\in\Pi_1(\GG)}\odot \underbrace{\RPow(-1)}_{\Pi_1(\GG)}\in\Prod^*(\GG).$$
\end{example}

\noindent Finally, we define \emph{$\Sum\Prod(\GG)$} as the smallest set containing $\GG\cup\Prod^*(\GG)$ with the following properties: 
\begin{enumerate}
\item For all $f,g\in\Sum\Prod(\GG)$ we have $f\oplus g\in\Sum\Prod(\GG)$.
\item For all $f,g\in\Sum\Prod(\GG)$ we have $f\odot g\in\Sum\Prod(\GG)$.
\item For all $f\in\Sum\Prod(\GG)$ and $k\in\NNP$ we have $f\opower k\in\Sum\Prod(\GG)$.
\item For all $f\in\Sum\Prod(\GG)$ and $l\in\NN$ we have $\Sum(l,f)\in\Sum\Prod(\GG)$.
\end{enumerate}
$\Sum\Prod(\GG)$ is also called the \emph{set of expressions in terms of nested sums over nested products}.
In addition, we define the following subsets:
\begin{enumerate}
	\item the \emph{set $\Prod(\GG)$ of expressions in terms of nested products (over $\GG$)}, i.e., all elements from $\Sum\Prod(\GG)$ which are free of sums;
	\item the \emph{set $\Prod_1(\GG)$ of expressions in terms of depth-1 products (over $\GG$)}, i.e., all elements from $\Prod(\GG)$ where the arising products are taken from $\Pi_1(\GG)$;
	\item the \emph{set $\Sum(\GG)$ of expressions in terms of nested sums (over $\GG$)}, i.e., all elements from $\Sum\Prod(\GG)$ where no products appear;
	\item the \emph{set $\Sum\Prod_1(\GG)$ of expressions in terms of nested sums over depth-1 products (over $\GG$)}, i.e., all elements from $\Sum\Prod(\GG)$ with products taken from $\Pi_1(\GG)$.
\end{enumerate}
In other words, besides the chain of set inclusions given in~\eqref{Equ:ProdSumInclusion} we also get
$$\Sum(\GG)\subset\Sum\Prod_1(\GG)\subset\Sum\Prod(\GG).$$ 
Furthermore, we introduce the \emph{set of nested sums over nested products} given by
\begin{equation*}
\Sigma(\GG)=\{\Sum(l,f)\mid l\in\NN \text{ and }f\in\Sum\Prod(\GG)\},
\end{equation*}
and the \emph{set of nested sums over single nested products} given by
\begin{equation*}
\Sigma_1(\GG)=\{\Sum(l,f)\mid l\in\NN \text{ and }f\in\Sum\Prod_1(\GG)\}.
\end{equation*}
For convenience we will also introduce the \emph{set $\Sigma\Pi(\GG)=\Sigma(\GG)\cup\Pi(\GG)$ of nested sums and products} and 
the \emph{set $\Sigma\Pi_1(\GG)=\Sigma_1(\GG)\cup\Pi_1(\GG)$ of nested sums and single-nested products}. In short, we obtain the following chain of sets:
$$\begin{array}{ccccccl}
\Pi_1(\GG)&\subset&\Sigma_1\Pi_1(\GG)&\supset&\Sigma_1(\GG)&&\txt{\small with single nested products}\\
\cap&&\cap&&\cap\\
\Pi(\GG)&\subset&\Sigma\Pi(\GG)&\supset&\Sigma(\GG)&&\txt{\small with nested products}\\[0.3cm]
\txt{\small products}&&\txt{\small products and\\ \small sums over products}&&\txt{\small sums over products}
\end{array}$$

\begin{example}\label{Exp:RatExpression}
	With $\GG=\KK(x)$ we get, e.g., the following expressions:
	\begin{align*}
	E_1&=\Sum(1,\Prod(1,x))\in\Sigma_1(\GG)\subset\Sum\Prod_1(\GG),\\
	E_2&=\Sum(1,\tfrac1{x+1}\odot\Sum(1,\tfrac1{x^3})\odot\Sum(1,\tfrac1{x}))\in\Sigma(\GG)\subset\Sum(\GG),\\
	E_3&=(E_1\oplus E_2)\odot E_1\in\Sum\Prod_1(\GG).
	\end{align*}
\end{example}

Finally, we introduce a function $\ev$ (a model of the term algebra) which evaluates a given expression of our term algebra to sequence elements. In addition, we also introduce the depth for our expressions. We start with the evaluation function $\ev:\GG\times\NN\to\KK$ given by~\eqref{Equ:EvalMixed} and the depth function $\depth:\GG\to\NN$ given by 
$$\depth(f)=\begin{cases}0&\text{if $f\in\KK$}\\
1&\text{if $f\in\GG\setminus\KK$}.
\end{cases}$$
\noindent Now $\ev$ and $\depth$ are extended recursively from $\GG$ to 
$\fct{\ev}{\Sum\Prod(\GG)\times\NN}{\Sum\Prod(\GG)}$ and $\depth:\Sum\Prod(\GG)\to\NN$ as follows.
\begin{enumerate}
	\item For $f,g\in\Sum\Prod(\GG)$ and $k\in\ZZ\setminus\{0\}$ ($k>0$ if $f\notin\Prod^*(\GG)$) we set
	\begin{align*}
	\ev(f\opower k,n)&:=\ev(f,n)^k,&\depth(f\opower k)&:=\depth(f),\\
	\ev(f\oplus g,n)&:=\ev(f,n)+\ev(g,n),&\depth(f\oplus g)&:=\max(\depth(f),\depth(g)),\\
	\ev(f\odot g,n)&:=\ev(f,n)\,\ev(g,n)&\depth(f\odot g)&:=\max(\depth(f),\depth(g));
	\end{align*}
	\item for $r\in\cRR$ and $\Sum(l,f), \Prod(\lambda,g)\in\Sum\Prod(\GG)$ we define
	\begin{align*}
	\ev(\RPow(r),n)&:=\prod_{i=1}^nr=r^n,&\depth(\RPow(r))&:=1,\\[-0.1cm]
	\ev(\Sum(l,f),n)&:=\sum_{i=l}^n\ev(f,i),&\depth(\Sum(l,f))&:=\depth(f)+1,\\[-0.1cm]
	\ev(\Prod(\lambda,g),n)&:=\prod_{i=\lambda}^n\ev(g,i),&\depth(\Prod(\lambda,g))&:=\depth(g)+1.
	\end{align*}
\end{enumerate}

\begin{remark}
(1) Since $\ev(\Prod(r,1),n)=\ev(\RPow(r),n)$, $\RPow$ is redundant. But it will be convenient for the treatment of canonical representations (see Definition~\ref{Def:CanonicalSet}).\\
(2) Any evaluation of $\Prod^*(\GG)$ is well defined and nonzero since the lower bounds of the products are set large enough via the $z$-function.\\  
(3) $\Sum\Prod_1(\rGG)$ covers as special cases generalized/cyclotomic harmonic sums \cite{Bluemlein:99,Vermaseren:99,Moch:02,ABS:11,ABS:13} and binomial sums~\cite{Davydychev:2003mv,Weinzierl:2004bn,ABRS:14}.
\end{remark}

\textit{In a nutshell, $\ev$ applied to $f\in\Sum\Prod(\GG)$ represents a sequence. In particular, $f$  can be considered as a simple program and $\ev(f,n)$ with $n\in\NN$ executes it (like an interpreter/compiler) yielding the $n$th entry of the represented sequence.}

\begin{definition}
For $F\in\Sum\Prod(\GG)$ and $n\in\NN$ we write
$F(n):=\ev(F,n).$
\end{definition}

\begin{example}
For $E_i\in\Sum\Prod(\KK(x))$ with $i=1,2,3$ in Ex.~\ref{Exp:RatExpression} we get $\depth(E_i)=3$ and
\begin{align*}
E_1(n)&=\ev(E_1,n)=\sum_{k=1}^n\prod_{i=1}^ki=\sum_{k=1}^nk!,&
E_2(n)&=\ev(E_2,n)=\sum_{k=1}^n \tfrac{1}{1+k}\Big(
\ssumB{i=1}k \frac{1}{i^3}\Big) 
\ssumB{i=1}k \frac{1}{i}
\end{align*}
and 
$E_3(n)=(E_1(n)+E_2(n))E_1(n)$. For $P\in\Sum\Prod(\KK(x,x_1))$ in Ex.~\ref{Exp:NestedProdExpr} we get
$$P(n)=\ev(P,n)=\Big(\prod_{k=1}^n\Big(\prod_{i=1}^ki\Big)^{-2}\Big)^2\Big(\prod_{k=1}^n\frac{q^k+q^{2k}}{k}\Big)(-1)^n,\quad \depth(P)=3.$$
\end{example}

\begin{example}
We show how the expressions of $\Sum\Prod(\GG)$ with $\ev$ are handled in
	
	\begin{mma}
		\In << Sigma.m \\
		\vspace*{-0.1cm}
		\Print \LoadP{Sigma - A summation package by Carsten Schneider
			\copyright\ RISC-JKU}\\
	\end{mma}
	
\noindent Instead of $F=\Sum(1,\tfrac1x)$ with $F(n)=\ev(F,n)=\sum_{k=1}^n\frac1k$ we introduce the sum by 
\begin{mma}
\In F=SigmaSum[\tfrac1k,\{k,1,n\}]\\
\Out \sum_{k=1}^n\frac1{k}\\
\end{mma}
\noindent where $n$ is kept symbolically. However, if the user replaces $n$ by a concrete integer, say $5$, the evaluation mechanism is carried out and we get $F(5)=\ev(F,5)$:
\begin{mma}
\In F/.n\to 5\\
\Out \frac{137}{60}\\
\end{mma}
\noindent Similarly, we can define $E_1$ from Example~\ref{Exp:RatExpression} as follows:
	\begin{mma}
		\In E_1=SigmaSum[SigmaFactorial[k],\{k,1,n\}]\\
		\Out \sum_{k=1}^nk!\\
	\end{mma}
\noindent Here \MText{SigmaFactorial} defines the factorials; its full definition is given by: 
\begin{mma}
\In GetFullDefinition[E_1]\\
\Out \sum_{k=1}^n\prod_{o_1=1}^ko_1\\
\end{mma}
\noindent Similarly, one can introduce as shortcuts powers, Pochhammer symbols, binomial coefficients, (generalized) harmonic sums~\cite{ABS:13} etc.\ with the function calls \MText{SigmaPower}, \MText{SigmaPochhammer}, \MText{SigmaBinomial} or \MText{S}, respectively; analogously $q$-versions are available. Together with Ablinger's package \texttt{HarmonicSums}, also function calls for cyclotomic sums~\cite{ABS:11} and binomial sums~\cite{ABRS:14} are available.

In the same fashion, we can define $E_2,E_3\in\Sum\Prod(\QQ(x))$ from Example~\ref{Exp:RatExpression} and $P\in\Sum\Prod(\QQ(q)(x,x_1))$ with $q=q_1$ from Example~\ref{Exp:NestedProdExpr} by 
	\begin{mma}
		\In E_2 = SigmaSum[
		SigmaSum[1/i, \{i, 1, k\}] SigmaSum[1/i^3, \{i, 1, k\}]/(k + 1), \{k, 1, n\}]\\	
		\Out \sum_{k=1}^n \frac{\displaystyle\Big(
			\ssumB{i=1}k \frac{1}{i^3}\Big) 
			\ssumB{i=1}k \frac{1}{i}}{1+k}\\
	\end{mma}

	\begin{mma}
		\In E_3=(E_1+E_2)E_1\\
		\Out \Big(\sum_{k=1}^nk!\Big)\Bigg(\sum_{k=1}^nk!+\sum_{k=1}^n \frac{\displaystyle\Big(
			\ssumB{i=1}k \frac{1}{i^3}\Big) 
			\ssumB{i=1}k \frac{1}{i}}{1+k}\Bigg)\\
	\end{mma}

	\begin{mma}
		\In P=SigmaProduct[SigmaProduct[i,\{i,1,k\}]^{-2},\{k,1,n\}]^2\newline \hspace*{1cm}SigmaProduct[(SigmaPower[q,k]+SigmaPower[q,k]^2)/k,\{k,1,n\}]SigmaPower[-1,n]\\
		\Out \Big(\prod_{k=1}^n\Big(\prod_{i=1}^ki\Big)^{-2}\Big)^2\Big(\prod_{k=1}^n\frac{q^k+(q^k)^2}{k}\Big)(-1)^n\\
	\end{mma}
\noindent Note that within \texttt{Sigma} the root of unity product $\RPow(\alpha)$ with $\alpha\in \cRR$ can be either defined by $\MText{SigmaPower[\alpha,n]}$ or $\MText{SigmaProduct[\alpha,\{k,1,n\}]}$. Whenever $\alpha$ is recognized as an element of $\cRR$, it is treated  as the special product $\RPow(\alpha)$.
\end{example}

Expressions in $\Sum\Prod(\GG)$ (similarly within Mathematica using \texttt{Sigma}) can be written in different ways such that they produce the same sequence. In the remaining part of this section we will elaborate on canonical (unique) representations~\cite{Buchberger:83}.

In a preprocessing step we can rewrite the expressions to a reduced representation; note that the equivalent definition in the ring setting is given in Definition~\ref{Def:ShapeOfElements}.

\begin{definition}\label{Def:SumProdReduced}
	An expression $A\in\Sum\Prod(\GG)$ is in \emph{reduced representation} if
	\begin{equation}\label{Equ:SumReduceRep}
	A=(f_1\odot P_1)\oplus(f_2\odot P_2)\oplus\dots\oplus (f_r\odot P_r)
	\end{equation}
	with $f_i\in\GG^*$ and 
	\begin{equation}\label{Equ:PPReduceRep}
	P_i=(a_{i,1}\opower z_{i,1})\odot (a_{i,2}\opower z_{i,2})\odot\dots\odot (a_{i,n_i}\opower z_{i,n_i})\in\Prod^*(\GG)
	\end{equation}
	for $1\leq i\leq r$ where 
	\begin{itemize}
	\item $a_{i,j}=\Sum(l_{i,j},f_{i,j})\in\Sigma(\GG)$ and $z_{i,j}\in\NNP$,
	\item $a_{i,j}=\Prod(l_{i,j},f_{i,j})\in\Pi(\GG)$ and $z_{i,j}\in\ZZ\setminus\{0\}$, or 
	\item $a_{i,j}=\RPow(f_{i,j})$ with $f_{i,j}\in\cRR$ and $1\leq z_{i,j}<\ord(r_{i,j})$ 
	\end{itemize}
	such that the following properties hold:
	\begin{enumerate}
		\item for each $1\leq i\leq r$ and $1\leq j<j'<n_i$ we have $a_{i,j}\neq a_{i,j'}$;
		\item for each $1\leq i<i'\leq r$ with $n_i=n_j$ there does not exist a $\sigma\in S_{n_i}$ with $P_{i'}=(a_{i,\sigma(1)}\opower z_{i,\sigma(1)})\odot (a_{i,\sigma(2)}\opower z_{i,\sigma(2)})\odot\dots\odot (a_{i,\sigma(n_i)}\opower z_{i,\sigma(n_i)})$.
	\end{enumerate}
We say that $H\in\Sum\Prod(\GG)$ is in \emph{sum-product reduced representation} (or in \emph{sum-product reduced form})
if it is in reduced representation and for each $\Sum(l,A)$ and $\Prod(l,A)$ that occur recursively in $H$ the following holds:
$A$ is in reduced representation as given in~\eqref{Equ:SumReduceRep}, 
$l\geq\max(L(f_1),\dots,L(f_r))$ (i.e. the first case of~\eqref{Equ:EvalMixed} is avoided during evaluations) and the lower bound $l$ is greater than or equal to the lower bounds of the sums and products inside of $A$.
\end{definition}


\begin{example}\quad
In \texttt{Sigma} the reduced representation of $E_3$ is calculated with the call
\begin{mma}
\In CollectProdSum[E_3]\\
\vspace*{-0.3cm}
\Out \Big(\sum_{k=1}^nk!\Big)^2+\Big(\sum_{k=1}^nk!\Big)\sum_{k=1}^n \frac{\displaystyle\Big(
	\ssumB{i=1}k \frac{1}{i^3}\Big) 
	\ssumB{i=1}k \frac{1}{i}}{1+k}\\
\end{mma}
\end{example}

\noindent Before we can state one of \texttt{Sigma}'s crucial features we need the following definitions.

\begin{definition}\label{Def:CanonicalSet}
Let $W\subseteq\Sigma\Pi(\GG)$.
We define 
\emph{$\Sum\Prod(W,\GG)$} as the set of elements from $\Sum\Prod(\GG)$ which are in reduced representation and where the arising sums and products are taken from $W$. More precisely, $A\in\Sum\Prod(W,\GG)$ if and only if it is of the form~\eqref{Equ:SumReduceRep}
with~\eqref{Equ:PPReduceRep} where $a_{i,j}\in W$. In the following we seek a $W$ with the following properties:
\begin{itemize}
\item $W$ is called \emph{shift-closed over $\GG$} if
for any $A\in\Sum\Prod(W,\GG)$, $s\in\ZZ$ there are 
$B\in\Sum\Prod(W,\GG)$ and $\delta\in\NN$ such that
$A(n+s)=B(n)$ holds for all $n\geq\delta$.
\item $W$ is called \emph{shift-stable over $\GG$} if for any product or sum in $W$ the multiplicand or summand is built by sums and products from $W$.
\item $W$ is called \emph{canonical reduced  over $\GG$} if for any $A,B\in\Sum\Prod(W,\GG)$ with
$A(n)=B(n)$ for all $n\geq\delta$
for some $\delta\in\NN$ the following holds: $A$ and $B$ are the same up to permutations of the operands in $\oplus$ and $\odot$.
\end{itemize}
\end{definition}

The sum-product reduced form is only a minor simplification, but it will be convenient to connect to the difference ring theory below; see Corollary~\ref{Cor:SwitchTiToAPS}. In Lemma~\ref{Lemma:shiftStableImpliesShiftClosed} we note further that shift-stability implies shift-closure. In particular, the shift operation can be straightforwardly carried out; the proof will be delivered later on page~\pageref{Proof:shiftStableImpliesShiftClosed}. 

\begin{lemma}\label{Lemma:shiftStableImpliesShiftClosed}
	If a finite set $W\subset\Sigma\Pi(\GG)$ is shift-stable and the elements are in sum-product reduced form\footnote{The sum-product reduced form is not necessary, but simplifies the proof given on page~\pageref{Proof:shiftStableImpliesShiftClosed}.}, then it is also shift-closed. If $\KK$ is computable then one can compute for $F\in\Sum\Prod(W,\GG)$ and $\lambda\in\ZZ$ a $G\in\Sum\Prod(W,\GG)$ such that $F(n+\lambda)=G(n)$ holds for all $n\geq\delta$ for some $\delta$. If one can factor polynomials over $\KK$, $\delta$ can be determined.
\end{lemma}

Based on this observation, we focus on $\sigma$-reduced sets which we define as follows.

\begin{definition}\label{Def:SigmaReducedSet}
$W\subseteq\Sigma\Pi(\GG)$ is called \emph{$\sigma$-reduced over $\GG$} if it is canonical reduced, shift-stable and the elements in $W$ are in sum-product reduced form. In particular, $A\in\Sum\Prod(W,\GG)$ is called \emph{$\sigma$-reduced (w.r.t.\ $W$)} if $W$ is $\sigma$-reduced over $\GG$.
\end{definition}

More precisely, we are interested in the following problem.

\begin{programcode}{Problem \textsf{SigmaReduce}: Compute a $\sigma$-reduced representation}
	\hspace*{-0.1cm}\begin{minipage}[t]{1.1cm}Given:\end{minipage}\begin{minipage}[t]{10.3cm}
		$A_1,\dots,A_u\in\Sum\Prod(\GG)$ with $\GG\in\{\rGG,\bGG,\mGG\}$, i,e., $\GG=\KK(x,x_1,\dots,x_v)$ or $\GG=\KK(x_1,\dots,x_v)$.
	\end{minipage}\\
	\begin{minipage}[t]{1.1cm}Find:\end{minipage}\begin{minipage}[t]{10.3cm}
		a $\sigma$-reduced set $W=\{T_1,\dots,T_e\}\subset\Sigma\Pi(\GG')$ in\footnote{In general, we might need a larger field $\GG'=\KK'(x,x_1,\dots,x_v)$ or $\GG'=\KK'(x_1,\dots,x_v)$ where the field $\KK$ is extended to $\KK'$.} $\GG'$,  $B_1\dots,B_u\in\Sum\Prod(W,\GG')$ and $\delta_1,\dots,\delta_u\in\NN$ such that for all $1\leq i\leq r$ we get
		$$A_i(n)=B_i(n)\quad n\geq\delta_i.$$
	\end{minipage}\\
\end{programcode}

\begin{example}\label{Exp:SigmaReduceBasic}
Consider the following two expressions from $\Sum\Prod(\QQ(x))$:
\begin{mma}
\In A_1 = SigmaSum[
SigmaSum[1/i, \{i, 1, k\}] SigmaSum[1/i^3, \{i, 1, k\}]/(k + 1), \{k, 1, n\}]\\
\Out \sum_{k=1}^n \frac{\displaystyle\Big(
	\ssumB{i=1}k \frac{1}{i^3}\Big) 
	\ssumB{i=1}k \frac{1}{i}}{1+k}\\
\end{mma}

\vspace*{-0.6cm}

\begin{mma}
\In A_2={\ssumB{i=1}n} \frac{1}{i^5}
-\frac{\displaystyle
	\ssumB{i=1}n \frac{1}{i^4}}{1+n}
-
{\ssumB{j=1}n} \frac{\displaystyle
	\ssumB{i=1}j \frac{1}{i^4}}{j}
-
{\ssumB{j=1}n} \frac{\displaystyle
	\ssumB{i=1}j \frac{1}{i^3}}{j^2}
+\frac{\displaystyle
	\ssumB{j=1}n \frac{\displaystyle
		\ssumB{i=1}j \frac{1}{i^3}}{j}}{1+n}
-
{\ssumB{j=1}n} \frac{\displaystyle
	\ssumB{i=1}j \frac{1}{i}}{j^4}
+\frac{\displaystyle
	\ssumB{j=1}n \frac{\displaystyle
		\ssumB{i=1}j \frac{1}{i}}{j^3}}{1+n}
+
{\ssumB{k=1}n}\frac{\displaystyle
	\ssumB{j=1}k \frac{\displaystyle
		\ssumB{i=1}j \frac{1}{i^3}}{j}}{k}
+
{\ssumB{k=1}n} \frac{\displaystyle
	\ssumB{j=1}k \frac{\displaystyle
		\ssumB{i=1}j \frac{1}{i}}{j^3}}{k};\\
\end{mma}

\noindent Then we solve Problem~\textsf{SigmaReduce} by executing:

\begin{mma}
\In \{B_1,B_2\}=SigmaReduce[\{A_1,A_2\},n]\\
\Out \{\sum_{k=1}^n \frac{\displaystyle\Big(
	\ssumB{i=1}k \frac{1}{i^3}\Big) 
	\ssumB{i=1}k \frac{1}{i}}{1+k},\sum_{k=1}^n \frac{\displaystyle\Big(
	\ssumB{i=1}k \frac{1}{i^3}\Big) 
	\ssumB{i=1}k \frac{1}{i}}{1+k}\}\\
\end{mma}

\noindent Since $B_1=B_2$, it follows $A_1=A_2$.
Note that the set $W$ pops up only implicitly. The set of all sums and products in the output, in our case 
$$W_0=\Big\{\sum_{k=1}^n \frac{1}{1+k}\Big(
	\sum_{i=1}^k \frac{1}{i^3}\Big) 
	\sum_{i=1}^k \frac{1}{i}\Big\}(=\big\{\Sum(1,\tfrac1{x+1}\odot\Sum(1,\tfrac1{x^3})\odot\Sum(1,\tfrac1{x}))\big\})$$ 
forms a canonical set in which $A_1$ and $A_2$ can be represented by $B_1$ and $B_2$ respectively. Adjoining in addition all sums and products that arise inside of the elements in $W_0$ we get
$W=\{\sum_{i=1}^n\frac{1}{i},\sum_{i=1}^n\frac{1}{i^3}\}\cup W_0$
which is a $\sigma$-reduced set.
Internally, \MText{SigmaReduce} parses the arising objects from left to right and constructs the underlying $\sigma$-reduced set $W$ in which the input expressions can be rephrased. 

Reversing the order of the input elements yields the following result:
\begin{mma}
	\In \{B_2,B_1\}=SigmaReduce[\{A_2,A_1\},n]\\
	\Out \Bigg\{-\Big(
	\sum_{k=1}^n \frac{1}{k^4}\Big) 
	\sum_{k=1}^n \frac{1}{k}
	+\frac{\displaystyle\Big(
		\ssumB{k=1}n \frac{1}{k^3}\Big) 
		\ssumB{k=1}n \frac{1}{k}}{1+n}
	-
	\sum_{k=1}^n \frac{\displaystyle
		\ssumB{k=1}k \frac{1}{k^3}}{k^2}
	+
	\sum_{k=1}^n \frac{\displaystyle
		\ssumB{k=1}k \frac{1}{k^4}}{k}
	+
	\sum_{k=1}^n \frac{\displaystyle\Big(
		\ssumB{k=1}k \frac{1}{k^3}\Big) 
		\ssumB{k=1}k \frac{1}{k}}{k},\newline
	-\Big(
	\sum_{k=1}^n \frac{1}{k^4}\Big) 
	\sum_{k=1}^n \frac{1}{k}
	+\frac{\displaystyle\Big(
		\ssumB{k=1}n \frac{1}{k^3}\Big) 
		\ssumB{k=1}n \frac{1}{k}}{1+n}
	-
	\sum_{k=1}^n \frac{\displaystyle
		\ssumB{k=1}k \frac{1}{k^3}}{k^2}
	+
	\sum_{k=1}^n \frac{\displaystyle
		\ssumB{k=1}k \frac{1}{k^4}}{k}
	+
	\sum_{k=1}^n \frac{\displaystyle\Big(
		\ssumB{k=1}k \frac{1}{k^3}\Big) 
		\ssumB{k=1}k \frac{1}{k}}{k}
	\Bigg\}\\
\end{mma}
\noindent In this case we get the $\sigma$-reduced set 
$$W=\Big\{
\sum_{j=1}^n \frac{1}{j^4},
\sum_{j=1}^n \frac{1}{j^3},
\sum_{j=1}^n \frac{1}{j},
\sum_{j=1}^n \frac{
	\ssumB{k=1}j \frac{1}{k^4}}{j},
\sum_{j=1}^n \frac{
	\ssumB{k=1}j \frac{1}{k^3}}{j^2},
\sum_{j=1}^n \frac{\big(
	\ssumB{k=1}j \frac{1}{k^3}\big) 
	\ssumB{k=1}j \frac{1}{k}}{j}\Big\}$$
(expressed in the \texttt{Sigma}-language)
and since $B_1=B_2$ we conclude again that $A_1=A_2$ holds for all $n\geq0$.
To check that $A_1=A_2$ holds, one can also execute
\begin{mma}
	\In SigmaReduce[A_1-A_2,n]\\
	\Out 0\\
\end{mma}
\noindent Here $W=\{\}$ is the $\sigma$-reduced set in which we can represent $A_1-A_2$ by $0$. 

\end{example}

Such a unique representation (up to trivial permutations) immediately gives rise to the following application: 
One can compare if two expressions $A_1$ and $A_2$ evaluate to the same sequences (from a certain point on): simply check if the resulting $B_1$ and $B_2$ in $\Sum\Prod(W,\GG)$ for a $\sigma$-reduced $W$ are the same (up to trivial permutations). Alternatively,
just check if $A_1-A_2$ can be reduced to zero. 
Besides that we will refine the above problem further. E.g., given $A\in\Sum\Prod(\GG)$, one can find an expression $B\in\Sum\Prod(W,\GG)$ and $\delta\in\NN$ such that $A(n)=B(n)$
holds for all $n\geq\delta$ and such that $B$ is as simple as possible. Here simple can mean that $\depth(B)$ is as small as possible. 
Other aspects might deal with the task of minimizing the number of elements in the set $W$. Finally, we want to emphasize that the above considerations can be generalized such that also unspecified/generic sequences can appear. The first important steps towards such a summation theory have been elaborated in~\cite{PS:19}.

As it turns out, the theory of difference rings provides all the techniques necessary to tackle the above problems. In the next section we introduce all the needed ingredients and will present our main result in Theorem~\ref{Thm:SigmaReduceInRPS} below.

\section{The difference ring approach for $\Sum\Prod(\GG)$}\label{Sec:GeneralDRApproach}

In the following we will rephrase expressions $H\in\Sum\Prod(\GG)$ as elements $h$ in a formal difference ring. More precisely, we will design
\begin{itemize}
	\item a ring $\AR$ with $\AR\supseteq\GG\supseteq\KK$ in which $H$ can be represented by $h\in\AR$;
	\item an evaluation function $\ev:\AR\times\NN\to\KK$ such that $H(n)=\ev(h,n)$ holds for sufficiently large $n\in\NN$;
	\item a ring automorphism $\sigma:\AR\to\AR$ which models the shift $H(n+1)$ with $\sigma(h)$.
\end{itemize}

\begin{example}\label{Exp:Q(x)[h]}
We will rephrase $F=\Sum(1,\frac1x)\in\Sum\Prod(\rGG)$ with $\rGG=\KK(x)$ where $\KK=\QQ$ in a formal ring. 
Namely, we take the polynomial ring $\AR=\rGG[s]=\QQ(x)[s]$ ($s$ transcendental over $\rGG$) and extend $\ev:\rGG\times\NN\to\QQ$ to $\ev':\AR\times\NN\to\QQ$ as  follows: for $h=\sum_{k=0}^df_k\,s^k$ with $f_k\in\rGG$ we set
\begin{equation}\label{Equ:QxsEv}
\ev'(h,n):=\sum_{k=0}^d\ev(f_k,n)\ev'(s,n)^k
\end{equation}
with
\begin{equation}\label{Equ:QxsEvs}
\ev'(s,n)=\sum_{i=1}^n\frac1i=:S_1(n)(=H_n);
\end{equation}
since $\ev$ and $\ev'$ agree on $\rGG$, we do not distinguish them anymore.
For any $$H=f_0\oplus (f_1\odot (F\opower 1))\oplus\dots\oplus (f_d\odot (F\opower d))$$
with $d\in\NN$ and $f_0,\dots,f_d\in\rGG$ we can take $h=\sum_{k=0}^df_ks^k\in\AR$ and get
$$H(n)=\ev(h,n)\quad\forall n\in\NN.$$
Further, we introduce the shift operator acting on the elements in $\AR$.
For the field $\rGG$ we simply define the field automorphism $\sigma:\rGG\to\rGG$ with $\sigma(f)=f|_{x\mapsto x+1}(=f(x+1))$.
Moreover, based on the observation that for any $n\in\NN$ we have
$$F(n+1)=\sum_{i=1}^{n+1}\frac1i=\sum_{i=1}^{n}\frac1i+\frac1{n+1},$$
we extend the automorphism $\sigma:\rGG\to\rGG$ to $\sigma':\AR\to\AR$ as follows: 
for $h=\sum_{k=0}^df_k\,s^k$ with $f_k\in\rGG$ we set
$\sigma'(h):=\sum_{k=0}^d\sigma(f_k)\sigma'(s)^k$
with
$\sigma'(s)=s+\tfrac1x$; since $\sigma'$ and $\sigma$ agree on $\rGG$, we do not distinguish them anymore.
We observe that
$$\ev(s,n+1)=\sum_{i=1}^{n+1}\frac1i=\sum_{i=1}^n\frac1i+\frac1{n+1}=\ev(s+\tfrac1{x+1},n)=\ev(\sigma(s),n)$$
holds for all $n\in\NN$ and more generally that
$\ev(h,n+l)=\ev(\sigma^l(h),n)$
holds for all $h\in\AR$, $l\in\ZZ$ and $n\in\NN$ with $n\geq\max(-l,0)$. 
\end{example}

As illustrated in the example above, the following definitions will be relevant.

\begin{definition}\label{Def:EvZ}
	A \emph{difference ring}/\emph{difference field} is a ring/field $\AR$ equipped with a ring/field automorphism $\sigma:\AR\to\AR $ which one also denotes by $\dfield{\AR}{\sigma}$. $\dfield{\AR}{\sigma}$ is \emph{difference ring/field extension} of a difference ring/field $\dfield{\HH}{\sigma'}$ if $\HH$ is a subring/subfield of $\AR$ and $\sigma|_{\HH}=\sigma'$. For a difference ring $\dfield{\AR}{\sigma}$ and a subfield $\KK$ of $\AR$ with\footnote{Note that $\dfield{\AR}{\sigma}$ is a difference ring extension of $\dfield{\KK}{\id}$.} $\sigma|_{\KK}=\id$ we introduce the following functions.	
	 \begin{enumerate}
		\item A function $\fct{\ev}{\AR\times\NN}{\KK}$ is called \emph{evaluation function} for $\dfield{\AR}{\sigma}$ if for all $f,g\in\AR$ and $c\in\KK$ there exists a $\lambda\in\NN$ with the following properties:
		\begin{align}
		\forall n\geq\lambda:&\,\ev(c,n)=c,\label{Equ:evC}\\
		\forall n\geq\lambda:&\,\ev(f+g,n)=\ev(f,n)+\ev(g,n),\label{Equ:evPlus}\\
		\forall n\geq\lambda:&\,\ev(f\,g,n)=\ev(f,n)\,\ev(g,n).\label{Equ:evTimes}\\
		\intertext{In addition, we require that for all $f\in\AR$ and $l\in\ZZ$ there exists a $\lambda$ with}
		\forall n\geq\lambda:&\,\ev(\sigma^l(f),n)=\ev(f,n+l).\label{Equ:evShift}
		\end{align}
		\item A function $\fct{L}{\AR}{\NN}$ is called an \emph{operation-function} (in short \emph{$o$-function}) for $\dfield{\AR}{\sigma}$ and an evaluation function $\ev$ if for any $f,g\in\AR$ with $\lambda=\max(L(f),L(g))$ the properties~\eqref{Equ:evPlus} and~\eqref{Equ:evTimes} hold and for any $f\in\AR$ and $l\in\ZZ$ with $\lambda=L(f)+\max(0,-l)$ property~\eqref{Equ:evShift} holds. 
		\item Let $G$ be a subgroup of $\AR^*$. $\fct{Z}{G}{\NN}$ is called a \emph{zero-function} (in short \emph{$z$-function}) for $\ev$ and $\GG$ if $\ev(f,n)\neq0$ holds for any $f\in\GG$ and integer $n\geq Z(f)$. 
	\end{enumerate}
\end{definition}

We note that a construction of a map $\ev:\AR\times\NN\to\KK$ with the properties~\eqref{Equ:evC} and~\eqref{Equ:evTimes} is straightforward. It is property~\eqref{Equ:evShift} that brings in extra complications: the evaluation of the elements in $\AR$ must be compatible with the automorphism $\sigma$. 

In this article we will always start with the following ground field; see~\cite{Bauer:99}.

\begin{example}\label{Exp:RationalDField1}
Take the rational function field	
$\mGG:=\GG=\KK(x,x_1,\dots,x_v)$ over $\KK=K(q_1,\dots,q_v)$, $v\geq0$, with the function~\eqref{Equ:EvalMixed},
together with the functions $L:\mGG\to\NN$ and $\fct{Z}{\mGG^*}{\NN}$ from the beginning of Section~\ref{Sec:TermAlgebra}. 
It is easy to see that $\ev:\mGG\times\NN\to\KK$ satisfies for all $c\in\KK$ and $f,g\in\GG$ the property~\eqref{Equ:evC} for $L(c)=0$ and the properties \eqref{Equ:evPlus} and~\eqref{Equ:evTimes} with $\lambda=\max(L(f),L(g))$. 
Finally, we take the automorphism $\fct{\sigma}{\mGG}{\mGG}$ defined by $\sigma|_{\KK}=\id$, $\sigma(x)=x+1$ and $\sigma(y_i)=q_i\,y_i$ for $1\leq i\leq v$. Then one can verify in addition that~\eqref{Equ:evShift} holds for all $f\in\mGG$ and $l\in\ZZ$ with $\lambda=\max(-l,L(f))$. Consequently, $\ev$ is an evaluation function for $\dfield{\mGG}{\sigma}$ and $L$ is an $o$-function for $\dfield{\mGG}{\sigma}$. In addition, $Z$ is a $z$-function for $\ev$ and $\mGG^*$ by construction. In the following we call $\dfield{\mGG}{\sigma}$ also a \emph{multibasic mixed difference field}. If $v=0$, i.e., $\rGG=\KK(x)=\KK'(x)$, we get the \emph{rational difference field} $\dfield{\rGG}{\sigma}$, and if we restrict to $\bGG=\KK(x_1,\dots,x_v)$, we get the \emph{multibasic difference field} $\dfield{\bGG}{\sigma}$.
\end{example}

We continue with the convention from above: if we write $\dfield{\GG}{\sigma}$, then it can be replaced by any of the difference rings $\dfield{\mGG}{\sigma}$, $\dfield{\rGG}{\sigma}$ or $\dfield{\bGG}{\sigma}$.

In the following we look for such a formal difference ring $\dfield{\AR}{\sigma}$ with a computable evaluation function $\ev$ and $o$-function $L$ in which we can model a finite set of expressions $A_1,\dots,A_u\in\Sum\Prod(\GG)$ with $a_1,\dots,a_u\in\AR$.

\begin{definition}
	Let $F\in\Sum\Prod(\GG)$ and $\dfield{\AR}{\sigma}$ be a difference ring extension of $\dfield{\GG}{\sigma}$ equipped with an evaluation function $\ev:\AR\times\NN\to\KK$. We say that \emph{$f\in\AR$ models $F$} if $\ev(f,n)=F(n)$ holds for all $n\geq\lambda$ for some $\lambda\in\NN$.
\end{definition}


\subsection{The naive representation in $APS$-extensions}

As indicated in Example~\ref{Exp:Q(x)[h]} our sum-product expressions will be rephrased in a tower of difference field and ring extensions. We start with the field version which will lead later to \pisiE-fields~\cite{Karr:81,Karr:85}.

\begin{definition}\label{Def:PSFieldExt}
	A difference field $(\FF,\sigma)$ is called a \emph{$PS$-field extension} of a difference field
	$\dfield{\HH}{\sigma}$ if 
	$\HH=\HH_0\leq\HH_1\leq\dots\leq\HH_e=\FF$
	is a tower of field extensions where for all $1\leq i\leq
	e$ one of the following holds:
	\begin{itemize}
		\item $\HH_i=\HH_{i-1}(t_i)$ is a rational function field extension with $\frac{\sigma(t_i)}{t_i}\in(\HH_{i-1})^*$ ($t_i$ is called a \emph{$P$-field monomial});
		\item $\HH_i=\HH_{i-1}(t_i)$ is a rational function extension with $\sigma(t_i)-t_i\in\HH_{i-1}$ ($t_i$ is called an \emph{$S$-field monomial}).
	\end{itemize}
\end{definition}

\begin{example}\label{Exp:RationalDField2}
Following Example~\ref{Exp:RationalDField1}, $\dfield{\mGG}{\sigma}$ with $\mGG=\KK(x,x_1,\dots,x_v)$ is a $PS$-field extension of $\dfield{\KK}{\sigma}$ with the $S$-field monomial $x$ and the $P$-monomials $x_1,\dots,x_v$. Similarly, $\dfield{\bGG}{\sigma}$ with $\bGG=\KK(x_1,\dots,x_v)$ forms a tower of $P$-field extensions of $\dfield{\KK}{\sigma}$ and $\dfield{\rGG}{\sigma}$ with $\rGG=\KK(x)$ is an $S$-field extension of $\dfield{\KK}{\sigma}$.
\end{example}

In addition, we will modify the field version to obtain the following ring version (allowing us to model also products over roots of unity).

\begin{definition}\label{Def:APSExt}
	A difference ring $(\EE,\sigma)$ is called an \emph{$APS$-extension} of a difference ring
	$\dfield{\AR}{\sigma}$ if 
	$\AR=\AR_0\leq\AR_1\leq\dots\leq\AR_e=\EE$
	is a tower of ring extensions where for all $1\leq i\leq
	e$ one of the following holds:
	\begin{itemize}
		\item $\AR_i=\AR_{i-1}[t_i]$ is a ring extension subject to the relation $t_i^{\nu}=1$ for some $\nu>1$ where $\frac{\sigma(t_i)}{t_i}\in(\AR_{i-1})^*$ is a primitive  $\nu$th root of unity ($t_i$ is called an \emph{$A$-monomial}, and $\nu$ is called the \emph{order of the $A$-monomial});
		\item $\AR_i=\AR_{i-1}[t_i,t_i^{-1}]$ is a Laurent polynomial ring extension with $\frac{\sigma(t_i)}{t_i}\in(\AR_{i-1})^*$ ($t_i$ is called a \emph{$P$-monomial});
		\item $\AR_i=\AR_{i-1}[t_i]$ is a polynomial ring extension with $\sigma(t_i)-t_i\in\AR_{i-1}$ ($t_i$ is called an \emph{$S$-monomial}).
	\end{itemize}
	Depending on the occurrences of the $APS$-monomials such an extension is also called an \emph{$A$-/$P$-/$S$-/$AP$-/$AS$/-/$PS$-extension}.
\end{definition}

\begin{example}\label{Exp:Qxs1}
	Take the rational difference ring $\dfield{\QQ(x)}{\sigma}$ with $\sigma(x)=x+1$ and $\sigma|_{\QQ}=\id$. Then the difference ring $\dfield{\QQ(x)[s]}{\sigma}$ with $\sigma(s)=s+\frac1{x+1}$ defined in Example~\ref{Exp:Q(x)[h]} is an $S$-extension of $\dfield{\QQ(x)}{\sigma}$ and $s$ is an $S$-monomial over $\dfield{\QQ(x)}{\sigma}$.
\end{example}

For an $APS$-extension $\dfield{\EE}{\sigma}$ of a difference ring $\dfield{\AR}{\sigma}$ we will also write $\EE=\AR\lr{t_1}\dots\lr{t_e}$. Depending on whether $t_i$ with $1\leq i\leq e$ is an $A$-monomial, a $P$-monomial or an $S$-monomial, $\GG\lr{t_i}$ with $\GG=\AR\lr{t_1}\dots\lr{t_{i-1}}$ stands for the algebraic ring extension $\GG[t_i]$ with $t_i^{\nu}$ for some $\nu>1$, for the ring of Laurent polynomials $\GG[t_1,t_1^{-1}]$ or for the polynomial ring $\GG[t_i]$, respectively.

For such a tower of $APS$-extensions we can use the following lemma iteratively to construct an evaluation function; for the corresponding proofs see~\cite[Lemma~5.4]{DR3}.

\begin{lemma}\label{Lemma:EvConstruction}
Let $\dfield{\AR}{\sigma}$ be a difference ring with a subfield $\KK\subseteq\AR$ where $\sigma|_{\KK}=\id$ that is equipped with an evaluation function $\ev:\AR\times\NN\to\KK$ and $o$-function $L$.
Let $\dfield{\AR\lr{t}}{\sigma}$ be an $APS$-extension of $\dfield{\AR}{\sigma}$ with $\sigma(t)=\alpha\,t+\beta$ ($\alpha=1$,  $\beta\in\AR$ or $\alpha\in\AR^*$, $\beta=0$). Further, suppose that $\ev(\sigma^{-1}(\alpha),n)\neq0$ for all $n\geq \mu$ for some $\mu\in\NN$. 
Then the following holds.
\begin{enumerate}
	\item Take $l\in\NN$ with $l\geq\max(L(\sigma^{-1}(\alpha),L(\sigma^{-1}(\beta)),\mu)$; if $t^{\lambda}=1$ for some $\lambda>1$ ($t$ is an $A$-monomial), set $l=1$.
	Then $\ev':\AR\lr{t}\times\NN\to\KK$ given by
	\begin{align}\label{Equ:DefineEvExt}
	\ev'(\sum_{i=a}^bf_i\,t^i,n)&=\sum_{i=a}^b\ev(f_i,n)\ev'(t,n)^i\quad\forall n\in\NN\\[-0.3cm]
	\intertext{with $f_i\in\AR$ for $a\leq i\leq b$ and} 
	\label{Equ:SumProdHom}
	\ev'(t,n)&=\begin{cases}
	\displaystyle \prod_{i=l}^n\ev(\sigma^{-1}(\alpha),i)&\text{ if\footnotemark $\sigma(t)=\alpha\,t$}\\
	\displaystyle\sum_{i=l}^n\ev(\sigma^{-1}(\beta),i)&\text{ if  $\sigma(t)=t+\beta$}
	\end{cases}
	\end{align}
	\footnotetext{If $t$ is an $A$-monomial, we have $\ev(t,n)=\alpha^n$.}
	is an evaluation function for $\dfield{\AR\lr{t}}{\sigma}$.
	\item There is an $o$-function $\fct{L'}{\AR\lr{t}}{\NN}$ for $\ev'$ defined by
	\begin{equation}\label{Equ:LDef}
	L'(f)=\begin{cases}
	L(f)&\text{if }f\in\AR,\\
	\max(l-1,L(f_a),\dots,L(f_b))&\text{if
		$f=\sum_{i=a}^b f_it^i\notin\AR\lr{t}\setminus\AR$}.
	\end{cases}
\end{equation}
\end{enumerate}
\end{lemma}

\begin{example}
In Example~\ref{Exp:Q(x)[h]} we followed precisely the construction (1) of the above lemma to construct for $\dfield{\QQ(x)[s]}{\sigma}$ an evaluation function. For this $\ev$ we can now apply also the construction (2) to enhance the $o$ function $L:\QQ(x)\to\NN$ (given in Example~\ref{Exp:RationalDField1} with $v=0$) to $L:\QQ(x)[s]\to\NN$ by setting $L(f)=\max(0,L(f_0),\dots,L(f_b))$ for $f=\sum_{i=0}^bf_is^i$. 
\end{example}

More precisely, the main idea is to apply the above lemma iteratively to extend the evaluation function $\ev$ from $\AR$ to $\EE$. However, if one wants to treat, e.g., the next $P$-monomial $t$ with $\frac{\sigma(t)}{t}=\alpha\in\EE^*$, one has to check if there is a $\mu\in\NN$ such that $\ev(\sigma^{-1}(\alpha),n)\neq0$ holds for all $n\geq \mu$. So far, we are not aware of a general algorithm that can accomplish this task. In order to overcome these difficulties, we will restrict $APS$-extensions further to a subclass which covers all summation problems that we have encountered in concrete problems so far.

Let $G$ be a multiplicative subgroup of $\AR^{*}$. 
Following~\cite{DR1,DR3}
we call 
\begin{align*}
\{G\}_{\AR}^{\EE} &:= \{ h\,t_1^{m_1}\dots t_e^{m_e}|\, h\in G \text{ and } m_{i}\in\ZZ\text{ where $m_i=0$ if $t_i$ is an $S$-monomial}\}\\
\intertext{the \emph{simple product group} over $G$ and}
[G]_{\AR}^{\EE} &:= \{ h\,t_1^{m_1}\dots t_e^{m_e}|\, h\in G \text{ and } m_{i}\in\ZZ\text{ where $m_i=0$ if $t_i$ is an $AS$-monomial}\}
\end{align*}
the \emph{basic product group} over $G$ for the nested $APS$--extension $\dfield{\EE}{\sigma}$ of $\dfield{\AR}{\sigma}$. Note that we have the chain of subgroups $[G]_{\AR}^{\EE}\leq\{G\}_{\AR}^{\EE}\leq\EE^*$.
In the following we will restrict ourselves to the following subclass of $APS$-extensions.
\begin{definition}\label{defn:simpleNestedAExtension}
	Let $\dfield{\AR}{\sigma}$ be a difference ring and let $G$ be a subgroup of $\AR^{*}$. Let $\dfield{\EE}{\sigma}$ be an $APS$-extension of $\dfield{\AR}{\sigma}$ with $\EE=\AR\langle t_{1}\rangle\dots\langle t_{e}\rangle$.
	\begin{enumerate}
	\item The extension is called \emph{$G$-basic} if for any $P$-monomial $t_i$ we have $\tfrac{\sigma(t_{i})}{t_{i}}\in [G]_{\AR}^{\AR\langle t_{1}\rangle\dots\langle t_{i-1}\rangle}$ and for any $A$-mon.\ $t_i$ we have  $\alpha_i=\frac{\sigma(t_{i})}{t_{i}}\in G$ with $\sigma(\alpha_i)=\alpha_i$.
	\item It is called \emph{$G$-simple} if for any $AP$-monomial $t_i$ we have $\tfrac{\sigma(t_{i})}{t_{i}}\in \{G\}_{\AR}^{\AR\langle t_{1}\rangle\dots\langle t_{i-1}\rangle}$.
	\end{enumerate}
	If $G=\AR^*$, it is also called \emph{basic} (resp.\ \emph{simple}) instead of $\AR^*$-basic (resp.\ $\AR^*$-simple).
\end{definition}

\noindent 
By definition any simple $APS$-extension is also a basic $APS$-extension. 
We will start with the more general setting of simple extensions, but will restrict later mostly to basic extensions.
For both cases we can supplement Lemma~\ref{Lemma:EvConstruction} as follows.

\begin{lemma}\label{Lemma:ZConstruction}
Let $\dfield{\AR}{\sigma}$ be a difference ring with a subfield $\KK\subseteq\AR$ where $\sigma|_{\KK}=\id$ that is equipped with an evaluation function $\ev$ and $o$-function $L$. Let $G$ be a subgroup of $\AR^*$ and let $\dfield{\AR\lr{t}}{\sigma}$ be an $APS$-extension of $\dfield{\AR}{\sigma}$ with $\sigma(t)=\alpha\,t+\beta$ with $\alpha\in G$ and $\beta\in\AR$. 
Suppose that there is in addition a $z$-function for $\ev$ and $G$. 
Take $l\in\NN$ with
\begin{equation}\label{Equ:DefineLowerBound}
l\geq\begin{cases}
\max(L(\sigma^{-1}(\alpha)),Z(\sigma^{-1}(\alpha)))&\text{if $t$ is an $AP$-monomial}\\
L(\sigma^{-1}(\beta))&\text{if $t$ is an $S$-monomial}.
\end{cases}
\end{equation}
Then we obtain an evaluation function $\ev'$ and $o$-function $L'$ for $\dfield{\AR\lr{t}}{\sigma}$ as given in Lemma~\ref{Lemma:EvConstruction}.
In addition, we can construct a $z$-function $Z'$ for $\{G\}_{\AR}^{\AR\lr{t}}$. If $\ev$, $L$ and $Z$ are computable, $\ev'$, $L'$ and $Z'$ are computable. 
\end{lemma}
\begin{proof}
For $r$ as defined in~\eqref{Equ:DefineLowerBound} the assumptions in Lemma~\ref{Lemma:EvConstruction} are fulfilled and the $\ev'$ with $L'$ defined in the lemma yield an evaluation function together with an $o$-function.
If $t$ is an $S$-monomial, $\{G\}_{\AR}^{\AR\lr{t}}=G$ and we can set $Z':=Z$. Otherwise, if $t$ is an $AP$-monomial, we have $\ev'(t,n)\neq0$ for all $n\in\NN$ by construction.
Thus for 
$f=g\,t^{m}\in\{G\}_{\AR}^{\AR\lr{t}}$ with $g\in G$ and $m\in\ZZ$ we have  $\ev(f,n)\neq0$ for all $n\geq Z(g)$. Thus we can define $Z'(f)=Z(g)$. If $L$ and $Z$ are computable, also $L'$ and $Z'$ are computable. In addition, if we can compute $\ev$, then clearly also $\ev'$ is computable.
\end{proof}

In general, suppose that we are given a difference ring $\dfield{\AR}{\sigma}$ with a subfield $\KK\subseteq\AR$ where $\sigma|_{\KK}=\id$. Assume in addition that we are given a (computable) evaluation function $\ev:\AR\times\NN\to\KK$ together with a (computable) $o$-function $L:\AR\to\NN$ and a (computable) $z$-function $Z:\AR^*\to\NN$. Furthermore, suppose that we are given a simple $APS$-extension $\dfield{\EE}{\sigma}$ of $\dfield{\AR}{\sigma}$ with $\EE=\AR\lr{t_1}\dots\lr{t_e}$. Then we can apply iteratively Lemmas~\ref{Lemma:EvConstruction} and~\ref{Lemma:ZConstruction} and get a (computable) evaluation function $\ev:\EE\times\NN\to\KK$ together with a (computable) $o$-function $L:\EE\to\NN$ and a (computable) $z$-function for $\{\AR^*\}_{\AR}^{\AR\langle t_{1}\rangle\dots\langle t_{e}\rangle}$; note that $\{\{\AR^*\}_{\AR}^{\HH}\}_{\HH}^{\HH\lr{t_i}}=\{\AR^*\}_{\AR}^{\HH\lr{t_i}}$ for all $\HH=\AR\lr{t_{1}}\dots\lr{t_{i-1}}$ with $1\leq i<e$.

It is natural to define the evaluation function iteratively using Lemma~\ref{Lemma:EvConstruction} but it is inconvenient to compute the $o$-function in this iterative fashion. Here the following lemma provides a shortcut for expressions which are given in reduced representation; for the corresponding representation in $\Sum\Prod(\GG)$ see Definition~\ref{Def:SumProdReduced}. 

\begin{definition}\label{Def:ShapeOfElements}
	Let $\dfield{\EE}{\sigma}$ be an $APS$-extension of $\dfield{\AR}{\sigma}$ with $\EE=\AR\lr{t_1}\dots\lr{t_e}$. Then we say that $f\in\EE$ is in \emph{reduced representation} if it is written in the form
	\begin{equation}\label{Equ:ShapeOfElements}
	f=\sum_{(m_1,\dots,m_e)\in S}f_{(m_1,\dots,m_e)}t_1^{m_1}\dots t_e^{m_e}
	\end{equation}
	with $f_{(m_1,\dots,m_e)}\in\AR$ and $S\subseteq M_1\times\dots\times M_e$ finite where 
	$$M_i=\begin{cases}
	\{0,\dots,\nu_i-1\}&\text{if $t_i$ is an $A$-extension of order $\nu_i$},\\
	\ZZ&\text{if $t_i$ is a $P$-monomial},\\
	\NN&\text{if $t_i$ is an $S$-monomial}.\\
	\end{cases}$$
\end{definition}

\begin{lemma}\label{Lemma:ShortCutforL}
Take a difference ring $\dfield{\AR}{\sigma}$ with a subfield $\KK\subseteq\AR$ where $\sigma|_{\KK}=\id$ that is equipped with an evaluation function $\ev:\AR\times\NN\to\KK$ together with an $o$-function $L$ and $z$-function $Z$. Let $\dfield{\EE}{\sigma}$ with $\EE=\AR\lr{t_1}\dots\lr{t_e}$ be a simple $APS$-extension of $\dfield{\AR}{\sigma}$ and let $\ev$ be an evaluation function and $Z$ be a $z$-function (using iteratively Lemmas~\ref{Lemma:EvConstruction} and~\ref{Lemma:ZConstruction}). Here 
the $l_i\in\NN$ for $1\leq i\leq e$ are the lower bounds of the corresponding sums/products in~\eqref{Equ:SumProdHom} with $t=t_i$.
Then for any $f\in\EE$ with~\eqref{Equ:ShapeOfElements} where $f_{(m_1,\dots,m_e)}\in\EE$ and $S\subseteq\ZZ^e$ we have
$$L(f)=\max(\max_{s\in S}L(f_s),\max_{j\in \sup(f)}l_j-1)$$
where $\sup(f)=\{1\leq j\leq e\mid t_j\text{  depends on } f\}$.
\end{lemma}
\begin{proof}
We show the statement by induction on $e$. If $e=0$, the statement holds trivially. Now suppose that the statement holds for $e\geq0$ extensions and let $f\in\EE\lr{t_{e+1}}$ with $\EE=\GG\lr{t_1}\dots\lr{t_e}$ where $f=\sum_{i=a}^{b}f_it_{e+1}^i$ with $f_i\in\EE$. If $f\in\EE$, then the statement holds by the induction assumption. Otherwise write $f_i=\sum_{(s_1,\dots,s_e)\in S_i}f^{(i)}_{s_1,\dots,s_e}t_1^{s_1}\dots t_e^{s_e}$ with $S_i\subseteq\ZZ^e$ and $f^{(i)}_{s}$ for $s\in S_i$ in reduced representation. In particular, we get $f=\sum_{(s_1,\dots,s_{e+1})\in S}h_{(s_1\dots,s_{e+1})} t_1^{s_1}\dots t_{e+1}^{s_{e+1}}$ with $h_{(s_1\dots,s_{e+1})}=f^{(s_{e+1})}_{(s_1,\dots,s_e)}$ and $S=\cup_{a\leq i\leq b}\{(s_1,\dots,s_e,i)\mid (s_1,\dots,s_e)\in S_i\}$.
Then by the induction assumption we get $L(f_i)=\max(\max_{s\in S_i}L(f^{(i)}_s),\max_{j\in \sup(f_i)}l_j-1)$. 
Thus by the definition in~\eqref{Equ:LDef} we get
\begin{align*}
L(f)&=\max(\max_{a\leq i\leq b}L(f_i),l_{e+1}-1)\\
&=\max(\max_{s\in S_a}L(f^{(a)}_s),\max_{s\in S_{a+1}}L(f^{(a+1)}_s),\dots,\max_{s\in S_b}L(f^{(b)}_s),\max_{j\in \sup(f)}l_j-1)\\
&=\max(\max_{s\in S} L(h_s),\max_{j\in \sup(f)}l_j-1).\tag*{$\square$}
\end{align*}
\end{proof}

Utilizing the above constructions with $\AR:=\GG$, we are now ready to show in Lemmas~\ref{Lemma:SwitchAPSToA} and~\ref{Lemma:SwitchAToAPS} given below that 
the representations in $\Sum\Prod(\GG)$ and in a basic $APS$-extension are closely related. Their proofs are rather technical (but not very deep). Still we will present all the details, since this construction will be crucial for further refinements. This will finally lead to a strategy to solve Problem~\textsf{SigmaReduce}.

\begin{lemma}\label{Lemma:SwitchAPSToA}
	Take the difference field $\dfield{\GG}{\sigma}$ with $\GG\in\{\rGG,\bGG,\mGG\}$ with the evaluation function $\ev$, $o$-function $L$ and $z$-function $Z$ from  Example~\ref{Exp:RationalDField1}. 
	Let $\dfield{\EE}{\sigma}$ with $\EE=\GG\lr{t_1}\dots\lr{t_e}$ be a basic $APS$-extension of $\dfield{\GG}{\sigma}$ and let $\ev$, $L$ and $Z$ be extended versions for $\dfield{\EE}{\sigma}$ (using Lemmas~\ref{Lemma:EvConstruction} and~\ref{Lemma:ZConstruction}).	
	Then for each $1\leq i\leq e$ one can construct $T_i\in\Sigma\Pi(\GG)$ in sum-product reduced representation with $\ev(t_i,n)=T_i(n)$ for all $n\geq L(t_i)$. In particular, if $f\in\EE\setminus\{0\}$, then there is $0\neq F\in\Sum\Prod(\{T_1,\dots,T_e\},\GG)$ with $F(n)=\ev(f,n)$ for all $n\geq L(f)$.\\
	If $\KK$ is computable and polynomials can be factored over $\KK$, all components can be computed.
\end{lemma}
\begin{proof}
	First suppose that we can construct such $T_i\in\Sigma\Pi(\GG)$ with $T_i(n)=\ev(t_i,n)$ for all $n\geq L(t_i)$ and $1\leq i\leq e$. Now take $f\in\EE$ in reduced representation, i.e., it is given in the form~\eqref{Equ:ShapeOfElements} with $S\subseteq\ZZ^e$. Now replace each $f_{(m_1,\dots,m_e)}\cdot t_1^{m_1}\dots t_e^{m_e}$ by $f_{(m_1,\dots,m_e)}\odot (T_1\opower{m_1})\odot\dots\odot(T_e\opower{m_e})$ and replace $+$ by $\oplus$ in $f$ yielding $F\in\Sum\Prod(\GG)$ in reduced representation. Then for each $n\geq L(f)$ we get
	\begin{equation}\label{Equ:fTransform}
	\begin{split}
	\ev(f,n)=&\ev(\sum_{(m_1,\dots,m_e)\in S}f_{(m_1,\dots,m_e)}t_1^{m_1}\dots t_e^{m_e},n)\\
	=&\sum_{(m_1,\dots,m_e)\in S}\ev(f_{(m_1,\dots,m_e)},n)\ev(t_1,n)^{m_1}\dots \ev(t_e,n)^{m_e}\\
	=&\sum_{(m_1,\dots,m_e)\in S}\ev(f_{(m_1,\dots,m_e)},n)\ev(T_1,n)^{m_1}\dots \ev(T_e,n)^{m_e}\\
	=&\ev(F,n).
	\end{split}
	\end{equation}
	Note: if $f\neq0$, we can find $(m_1,\dots,m_e)\in S$ with $f_{(m_1,\dots,m_e)}\in\GG^*$ which implies that $F\neq0$. This shows the second part of statement (1).\\
	Finally we show the existence of the $T_i$ by induction on $e$. For $e=0$ nothing has to be shown. Suppose that the statement holds for $e\geq0$ extensions and consider the $APS$-monomial $t_{e+1}$ over $\EE=\GG\lr{t_1}\dots\lr{t_e}$. By assumption we can take $T_i\in\Sum\Prod(\GG)$ in sum-product reduced representation with $T_i(n)=\ev(t_i,n)$ for all $n\geq L(t_i)$ and $1\leq i\leq e$. Now consider the $APS$-monomial $t_{e+1}$ with $\sigma(t_{e+1})=\alpha\,t_{e+1}+\beta$. By assumption we have~\eqref{Equ:SumProdHom} ($\ev'$ replaced by $\ev$) with $l\in\NN$ where~\eqref{Equ:DefineLowerBound} and $L$ is defined by~\eqref{Equ:LDef} ($L'$ replaced by $L$). In particular, we have $l\geq\max(L(\sigma^{-1}(\alpha)),L(\sigma^{-1}(\beta)),\mu)$ with $\mu\geq Z(\sigma^{-1}(\alpha))$, and
	$L(t_{e+1})=l-1$.\\
	\textbf{$A$-monomial case:} If $t_{e+1}$ is an $A$-monomial, we have $\sigma(t_{e+1})=\alpha\,t_{e+1}$ with $\alpha\in\cRR$. In particular, we have $\ev(t_{e+1},n)=\alpha^n$. Thus we set $T_{e+1}=\RPow(\alpha)$  and get $\ev(t_{e+1},n)=T_{e+1}(n)$ for all $n\geq L(t_{e+1})=0$. \\
	\textbf{$S$-monomial case:} If $t_{e+1}$ is an $S$-monomial, we have $\sigma(t_{e+1})=t_{e+1}+\beta$ with $\beta\in\EE$.  
	Now take $f=\sigma^{-1}(\beta)$ in reduced representation. Then by construction $l\geq\max(L(\sigma^{-1}(\beta)),0)=L(f)$. Further, we can take $F\in\Sum\Prod(\GG)$ as constructed above with~\eqref{Equ:fTransform} for all $n\geq l\geq L(f)$.
	Thus for $T_{e+1}=\Sum(l,F)$ we get $\ev(t_{e+1},n)=T_{e+1}(n)$ for all $n\geq l-1=L(t_{e+1})$.\\
	\textbf{$P$-monomial case:} If $t_{e+1}$ is a $P$-monomial, we have $\sigma(t_{e+1})=\alpha\,t_{e+1}$ with $\alpha\in[\GG]_{\GG}^{\EE}$, i.e.,
	$\alpha=g\,t_1^{n_1}\dots t_e^{n_e}$ with $g\in\GG^*$ and $n_1,\dots,n_e\in\ZZ$ with $n_i=0$ if $t_i$ is an $AS$-monomial. Thus $f=\sigma^{-1}(\alpha)=h\,t_1^{m_1}\dots t_e^{m_e}$ with $h:=f_{(m_1,\dots,m_e)}\in\GG^*$ and $m_1,\dots,m_e\in\ZZ$ with $m_i=0$ if $t_i$ is an $AS$-monomial. By construction, $l\geq\max(L(f),Z(f))=\max(L(f),Z(h))$.
	As above we get $F=h\odot(T_1\opower{m_1})\odot \dots\odot(T_e\opower{m_e})\in\Sum\Prod(\GG)$ such that $\ev(f,n)=F(n)$ holds for all $n\geq L(f)$ and $\ev(f,n)=F(n)\neq0$ for all $n\geq l$.
	Thus for $T_{e+1}=\Prod(l,F)\in\Prod(\GG)$ we get $\ev(t_{e+1},n)=T_{e+1}(n)$ for all $n\geq l-1=L(t_{e+1})$.\\
	We note that in the last two cases $T_{e+1}$ is in sum-product reduced representation: the arising sums and products in $F$	
	are in sum-product reduced representation by induction, $F$ given by~\eqref{Equ:fTransform} is in reduced representation and we have $l\geq\max_{k\in S}L(f_k)$ where $l$ is larger than all the lower bounds of the sums and product in $F$ due to Lemma~\ref{Lemma:ShortCutforL}. This completes the induction step.\\
	If $\KK$ is computable and one can factorize polynomials over $\KK$, the functions $Z$ and $L$ are computable and thus all the ingredients can be computed.
\end{proof}

\begin{definition}
	Given $\dfield{\GG}{\sigma}$ where $\GG\in\{\rGG,\bGG,\mGG\}$ with $\ev$, $L$ and $Z$ from Example~\ref{Exp:RationalDField1},
	let $\dfield{\EE}{\sigma}$ be a basic $APS$-extension with an evaluation function $\ev$ together with $L$ and $Z$ given by iterative application of Lemmas~\ref{Lemma:EvConstruction} and~\ref{Lemma:ZConstruction}. Let $a\in\EE$ be in reduced representation. Then following the construction of Lemma~\ref{Lemma:SwitchAPSToA} one obtains $A\in\Sum\Prod(\GG)$ in sum-product reduced representation with $A(n)=\ev(a,n)$ for all $n\geq L(a)$. The derived $A$ is also called the \emph{canonical induced sum-product expression} of $a$ w.r.t.\ $\dfield{\AR}{\sigma}$ and $\ev$ and we write $\expr(a):=A$.
\end{definition}

\begin{example}[Cont. of Ex.~\ref{Exp:Q(x)[h]}]
	For $a=x+\frac{x+1}{x}s^4\in\QQ(x)[s]$ with our evaluation function $\ev$ we obtain the canonical induced sum-product expression $$\expr(a)=A=x\oplus\left(\tfrac{x+1}{x}\odot(\Sum(1,\tfrac1x)\opower4)\right)\in\Sum(\QQ(x))$$ 
	with $A(n)=\ev(a,n)$ for all $n\geq1$.
\end{example}

\begin{lemma}\label{Lemma:SwitchAToAPS}
Take the difference field $\dfield{\GG}{\sigma}$ with $\GG\in\{\rGG,\bGG,\mGG\}$ with the evaluation function $\ev$, $o$-function $L$ and $z$-function $Z$ from  Example~\ref{Exp:RationalDField1}. 
Let $\dfield{\HH}{\sigma}$ be a basic $APS$-extension of $\dfield{\GG}{\sigma}$ and let $\ev$, $L$ and $Z$ be extended versions for $\dfield{\HH}{\sigma}$ (using Lemmas~\ref{Lemma:EvConstruction} and~\ref{Lemma:ZConstruction}). Let $A\in\Sum\Prod(\GG)$. Then there is an $APS$-extension $\dfield{\EE}{\sigma}$ of $\dfield{\HH}{\sigma}$ which forms a basic $APS$-extension of $\dfield{\GG}{\sigma}$ together with the extended functions $\ev$, $L$ and $Z$ (using  Lemmas~\ref{Lemma:EvConstruction} and~\ref{Lemma:ZConstruction}) in which one can model $A$ by $a\in\EE$: i.e., $\ev(a,n)=A(n)$ holds for all $n\geq\delta$ for some $\delta\in\NN$.\\
If $\KK$ is computable and one can factorize polynomials over $\KK$, all the ingredients can be computed.
\end{lemma}
\begin{proof}
We prove the lemma by induction on the depth of the arising sums ($\Sum$) and products ($\Prod$ and $\RPow$) in $A\in\Sum\Prod(\GG)$. If no sums and products arise in $A$, then $A\in\GG$ and the statement clearly holds. Now suppose that the statement holds for all expressions with sums/products whose depth is smaller than or equal to $d\geq0$. Take all products and sums $T_1,\dots,T_r\in\Sigma\Pi(\GG)$ that arise in $A$. We proceed stepwise for $i=1,\dots,r$ with the starting field $\HH$. Suppose that we have constructed an $APS$-extension $\dfield{\AR}{\sigma}$ of $\dfield{\HH}{\sigma}$ which forms a basic $APS$-extension of $\dfield{\GG}{\sigma}$. Suppose in addition that we are given an extended evaluation function $\ev$, $o$-function $L$ and $z$-function $Z$ function (using Lemmas~\ref{Lemma:EvConstruction} and~\ref{Lemma:ZConstruction}) in which we find $b_1,\dots,b_{i-1}$ with $\ev(b_j,n)=T_j(n)$ for all $n\geq L(b_j)$ and all $1\leq j<i$. Now we consider $T_i$.\\
\textbf{Bookkeeping\footnote{This step is not necessary for the proof, but avoids unneccesary copies of $APS$-monomials. When we refine this construction later, this step will be highly relevant.}:} If $T_i$ has been treated earlier (i.e., by handling sums and products of depth $\leq d$), we get $b_i\in\AR$ with $\ev(b_i,n)=T_i(n)$ for all $n\geq L(a_i)$.\\
\textbf{$\RPow$-case:} If $T_i=\RPow(\alpha)$, we take the $A$-extension $\dfield{\AR\lr{t}}{\sigma}$ of $\dfield{\AR}{\sigma}$ with $\sigma(t)=\alpha t$ of order $\ord(\alpha)$ and extend $\ev$ to $\AR\lr{t}$ by $\ev(t,n)=\alpha^n$.
Further, we extend $L:\AR\lr{t}\to\NN$ with~\eqref{Equ:LDef} and get $L(t)=0$. Thus we can take $b_{i}=t$ and get $\ev(b_i,n)=T_{i}(n)$ for all $n\geq L(b_i)=0$.

Otherwise, we can write $T_i=\Sum(\lambda,H)$ or $T_i=\Prod(\lambda,H)$ where the sums and products in $H\in\Sum\Prod(\GG)$ have depth at most $d$. By assumption we can construct an $APS$-extension $\dfield{\AR'}{\sigma}$ of $\dfield{\AR}{\sigma}$ which is a basic $APS$-extension of $\dfield{\GG}{\sigma}$ and we can extend $\ev$, $L$ and $Z$ (using Lemmas~\ref{Lemma:EvConstruction} and~\ref{Lemma:ZConstruction})
and get $h\in\AR'$ with $\ev(h,n)=H(n)$ for all $n\geq \delta$ for some $\delta\in\NN$ with $\delta\geq L(h)$.\\ 
\textbf{Sum-case:} If $T_i=\Sum(\lambda,H)$, we take the $S$-extension $\dfield{\AR'\lr{t}}{\sigma}$ of $\dfield{\AR'}{\sigma}$ with $\sigma(t)=t+\sigma(h)$. In addition, we extend $\ev$ to $\AR'\lr{t}$ by $\ev(t,n)=\sum_{k=l}^n\ev(h,k)$ with
$l=\max(\delta,\lambda)$; note that~\eqref{Equ:DefineLowerBound} is satisfied. Further, we extend $L:\AR'\lr{t}\to\NN$ with~\eqref{Equ:LDef} and get $L(t)=l-1$. 
Finally, we set $c=\sum_{k=\lambda}^{l-1}H(k)\in\KK$. Then we get $b_{i}=t+c$ with 
$\ev(b_i,n)=\sum_{k=\lambda}^nH(k)=\ev(\Sum(\lambda,H),n)$ for all $n\geq L(b_i)=l-1$.\\
\textbf{Product-case:} If $T_i=\Prod(\lambda,H)$, we take the $P$-extension $\dfield{\AR'\lr{t}}{\sigma}$ of $\dfield{\AR'}{\sigma}$ with $\sigma(t)=\sigma(h) t$. 
In addition we extend $\ev$ to $\AR'\lr{t}$ by $\ev(t,n)=\prod_{k=l}^n\ev(h,k)$ with $l=\max(L(h),Z(h),\lambda)$; note that~\eqref{Equ:DefineLowerBound} is satisfied.
Further, we extend $L:\AR\lr{t}\to\NN$ with~\eqref{Equ:LDef} and get $L(t)=l-1$. Thus we can take $b_{i}=c\,t$ with $c=\prod_{k=\lambda}^{l-1}H(k)\in\KK^*$ (the product evaluation is nonzero by assumption of $\Pi(\GG)$) and get $\ev(b_i,n)=\prod_{k=\lambda}^n H(k)=\ev(\Prod(\lambda,H),n)$ for all $n\geq L(b_i)=l-1$.\\
In all three cases we can follow Lemma~\ref{Lemma:ZConstruction} and extend the $z$-function accordingly.
After carrying out the steps $i=1,\dots,r$ we get a basic $APS$-extension $\dfield{\EE}{\sigma}$ of $\dfield{\HH}{\sigma}$ together with an evaluation function $\ev$, $o$-function $L$ and $z$-function $Z$ (using Lemmas~\ref{Lemma:EvConstruction} and~\ref{Lemma:ZConstruction}) and $b_1,\dots,b_r$ such that $T_i(n)=\ev(b_i,n)$ holds for all $1\leq i\leq r$ and $n\geq L(b_i)$. Finally, let $f_{1},\dots,f_{s}\in\GG$ be all arising elements in $A$ (that do not arise within $\Prod$ and $\Sum$). Define $\delta=\max(L(f_1),\dots,L(f_s))\in\NN$. Then for each $n\in\NN$ with $n\geq \delta$ we have that $\ev(A,n)$ can be carried out without catching poles in the second case of~\eqref{Equ:EvalMixed}. Now replace each $T_i$ with $1\leq i\leq r$ 
in $A$ by $b_i$ and replace $\oplus,\odot,\opower$ by $+$, $\cdot$, and $\hat{}\,$, respectively. This yields $a\in\EE$ which we can write in reduced representation. Note that in $a$ some $f_{k}$ remain and others are combined by putting elements over a common denominator which lies in $\KK[x,x_1,\dots,x_v]$ (or in $\KK[x_1,\dots,x_v]$). Further, some factors of the denominators might cancel. Thus $L(a)\leq \delta$. In particular, when carrying out the evaluations $\ev(a,n)$ and $\ev(A,n)$ for $n\geq\delta$ no poles arise  and thus by the homomorphic property of the evaluation it follows that $\ev(a,n)=\ev(A,n)$ for all $n\geq \delta$. This completes the induction step.\\
If $\KK$ is computable and one can factorize polynomials over $\KK$, then the $z$- and $o$-function for $\GG$ are computable. Thus all the components of the iterative construction (using Lemmas~\ref{Lemma:EvConstruction} and~\ref{Lemma:ZConstruction}) are computable.
\end{proof}

As consequence, we can establishes with Lemma~\ref{Lemma:SwitchAPSToA} above and the following corollary a 1-1 correspondence between basic $APS$-extensions and shift-stable sets whose expressions are in sum-reduced representation.

\begin{corollary}\label{Cor:SwitchTiToAPS}
Let $W=\{T_1,\dots,T_e\}\subset\Sigma\Pi(\GG)$ be in sum-product reduced representation and shift-stable. More precisely, for each $1\leq i\leq e$ the arising sums/products in $T_i$ are contained in $\{T_1,\dots,T_{i-1}\}$. Then there is a basic $APS$-extension $\dfield{\EE}{\sigma}$ of $\dfield{\GG}{\sigma}$ with $\EE=\GG\lr{t_1}\dots\lr{t_e}$ equipped with an evaluation function $\ev$ (using Lemmas~\ref{Lemma:EvConstruction} and~\ref{Lemma:ZConstruction}) such that $T_{i}=\expr(t_i)\in\Sigma\Pi(\GG)$ holds for $1\leq i\leq e$. 
\end{corollary}
\begin{proof}
We can treat the elements $T_1,\dots,T_e$ following the construction of  Lemma~\ref{Lemma:SwitchAToAPS} iteratively. 
Let us consider the $i$th step with $T_i=\Sum(\lambda,H)$ or $T_i=\Prod(\lambda,H)$.
Since the $T_i$ are in sum-product reduced form it follows from Lemma~\ref{Lemma:ShortCutforL} that within the sum-case (resp.\  product-case) we can guarantee $l=\lambda$, i.e, $c=0$ (resp.\ $c=1$). Thus $\ev(t_i,n)=T_i(n)$ for all $n\geq l$ and hence $\expr(t_i)=T_i(n)$.
\end{proof}

\noindent In addition, we can provide the following simple proof of Lemma~\ref{Lemma:shiftStableImpliesShiftClosed}.

\begin{proof}[of Lemma~\ref{Lemma:shiftStableImpliesShiftClosed}]\label{Proof:shiftStableImpliesShiftClosed}
	 Let $W=\{T_1,\dots,T_e\}\subseteq W$ be shift-stable and the $T_i$ in sum-product reduced form. Take $F\in\Sum\Prod(W,\GG)$ and $\lambda\in\ZZ$. W.l.o.g.\ we may assume that the $T_i$ are given as in Corollary~\ref{Cor:SwitchTiToAPS}. Thus we can take an $APS$-extension 
	$\dfield{\GG\lr{t_1}\dots\lr{t_e}}{\sigma}$ of $\dfield{\GG}{\sigma}$ equipped with an evaluation function  $\ev$ and $o$-function $L$ such that $\expr(t_i)=T_i$ for $1\leq i\leq e$. Then
	we can take $f\in\EE$ with~\eqref{Equ:ShapeOfElements} and get $F(n+\lambda)=\ev(F,n+\lambda)=\ev(\sigma^{\lambda}(t_i),n)=G(n)$ for all $n\in L(f)+\max(0,-\lambda)$ with $G(n):=\expr(\sigma^{\lambda}(t_i))\in\Sum\Prod(W,\GG)$. Thus $W$ is shift-closed. If $\KK$ is computable and one can factor polynomials over $\KK$, then one can compute the $o$-function $L$ and all the above components are computable. 
\end{proof}

\noindent In short, the naive construction of $APS$-extensions will not gain any substantial simplification (except a transformation to a sum-product reduced representation). In the next section we will refine this construction further to solve Problem~\textsf{SigmaReduce}. 

\subsection{The embedding into the ring of sequences and \rpisiE-extensions}

Let $\dfield{\AR}{\sigma}$ be a difference ring with a subfield $\KK\subseteq\AR$ where $\sigma|_{\KK}=\id$ that is equipped with an evaluation function $\ev:\AR\times\NN\to\KK$. Then $\ev$ naturally produces sequences in the commutative ring $\KK^{\NN}$ with the identity element $\vect{1}=(1,1,1,\dots)$ with component-wise addition and multiplication. More precisely, we can define the function
$\fct{\tau}{\AR}{\KK^{\NN}}$ with
\begin{equation}\label{Equ:EvToTau}
\tau(f)=(\ev(f,n))_{n\geq0}=(\ev(f,0),\ev(f,1),\ev(f,2),\dots).
\end{equation}
Due to \eqref{Equ:evPlus} and~\eqref{Equ:evTimes} the map
$\tau$ can be turned to a ring homomorphism by defining the equivalence relation 
$(f_n)_{n\geq0}\equiv (g_n)_{n\geq0}$ with $f_j=g_j$ for all $j\geq\lambda$ for some $\lambda\in\NN$; compare~\cite{AequalB}. It is easily seen that the set of equivalence classes $[f]$ with $f\in\KK^{\NN}$ forms 
with $[f]+[g]:=[f+g]$ and $[f][g]:=[f g]$
again a commutative ring with the identity element $[\vect{1}]$ which we will denote by $\seqK$. In the following we will simply write $f$ instead of $[f]$. In this setting, $\fct{\tau}{\AR}{\seqK}$ forms a ring homomorphism. In addition the shift operator $S:\seqK\to\seqK$ defined by
$$\shiftS((a_0,a_1,a_2,\dots))=(a_1,a_2,a_3,\dots)$$
turns to a ring automorphism. In the following we call \emph{$(\seqK,\shiftS)$ also the (difference) ring of sequences over $\KK$}. Finally, we observe that property~\eqref{Equ:evShift} implies that
\begin{equation}\label{Equ:DRHomo}
\tau(\sigma(f))=S(\tau(f))
\end{equation}
holds for all $f\in\AR$, i.e., $\tau$ turns to a difference ring homomorphism. Finally, property~\eqref{Equ:evC} implies
\begin{equation}\label{Equ:KHomo}
\tau(c)=\vect{c}=(c,c,c,\dots)
\end{equation}
for all $c\in\KK$. In the following we call a ring homomorphism $\tau:\AR\to\seqK$ with~\eqref{Equ:DRHomo} and~\eqref{Equ:KHomo} also a \emph{$\KK$-homomorphism}.

We can now link these notions to our construction from above with $\GG\in\{\rGG,\bGG,\mGG\}$. Let $\dfield{\EE}{\sigma}$ with $\EE=\GG\lr{t_1}\dots\lr{t_e}$ be a basic $APS$-extension of $\dfield{\GG}{\sigma}$ and take an evaluation function $\ev:\EE\times\NN\to\KK$ with $o$-function $L$. Such a construction can be accomplished by iterative application of Lemmas~\ref{Lemma:EvConstruction} and~\ref{Lemma:ZConstruction}. Then the function $\tau:\EE\to\AR$ with~\eqref{Equ:EvToTau}
for $f\in\EE$ yields a $\KK$-homomorphism. 

If we find two different elements $a,b\in\EE$ with $\tau(a)=\tau(b)$, then we find two different sum-product reduced representations $\expr(a)$ and $\expr(b)$ in terms of the sums and products given in $W=\{\expr(t_1),\dots,\expr(t_e)\}\subseteq\Sigma\Pi(\GG)$ which evaluates to the same sequence. In short, $W$ is not canonical reduced (and thus not $\sigma$-reduced) over $\GG$. This shows that a solution of Problem~\textsf{SigmaReduce} can be only accomplished if $\tau$ is injective.

In this context, the set of constants plays a decisive role.

\begin{definition}
	For a difference ring $\dfield{\AR}{\sigma}$ the \emph{set of constants} is defined by
	$\const{\AR}{\sigma}=\{c\in\AR\mid \sigma(c)=c\}.$
\end{definition}

\noindent In general, $\const{\AR}{\sigma}$ is a subring of $\AR$. If $\AR$ is a field, then $\const{\AR}{\sigma}$ itself is a field which one also calls the \emph{constant field of $\dfield{\AR}{\sigma}$}.  

With this extra notion we can state now the following remarkable property that is based on results from~\cite{DR3}; compare also~\cite{Singer:97}.

\begin{theorem}\label{Thm:InjectiveProp}
Let $\dfield{\EE}{\sigma}$ be a basic $APS$-extension of a difference field $\dfield{\FF}{\sigma}$ with $\KK=\const{\FF}{\sigma}$ and let $\tau$ be a $\KK$-homomorphism. 
Then $\tau$ is injective iff $\const{\EE}{\sigma}=\KK$.
\end{theorem}
\begin{proof}
	Suppose that $\const{\EE}{\sigma}=\KK$. By Theorem~\cite[Thm~3.3]{DR3} it follows that $\dfield{\EE}{\sigma}$ is simple (i.e., the only difference ideals in $\EE$ are $\{0\}$ or $\EE$) and thus by \cite[Lemma~5.8]{DR3} we can conclude that $\tau$ is injective. Conversely, if $\tau$ is injective, it follows by~\cite[Lemma~5.13]{DR3} that $\const{\EE}{\sigma}=\KK$.
\end{proof}

\noindent This result gives rise to the following refined definition of $PS$-field/$APS$-extensions.

\begin{definition}\label{Def:PiSiFieldExt}
	Let $(\FF,\sigma)$ be a $PS$-field extension of $\dfield{\HH}{\sigma}$ as defined in Definition~\ref{Def:PSFieldExt}. Then this is called a \emph{\pisiE-field extension} if $\const{\FF}{\sigma}=\const{\HH}{\sigma}$. The arising $P$-field and $S$-field monomials are also called \emph{\piE-field and \sigmaE-field monomials}, respectively. In particular, we call it a \emph{\piE-/\sigmaE-/\piE\sigmaE-field extension} if it is built by the corresponding monomials.
	$\dfield{\FF}{\sigma}$ is called a \emph{\pisiE-field over $\KK$} if $\dfield{\FF}{\sigma}$ is a \pisiE-field extension of $\dfield{\KK}{\sigma}$ and $\const{\KK}{\sigma}=\KK$.
\end{definition}

\begin{example}\label{Exp:RationalDField3}
	As mentioned in Examples~\ref{Exp:RationalDField1} and~\ref{Exp:RationalDField2}, the difference fields $\dfield{\rGG}{\sigma}$, $\dfield{\bGG}{\sigma}$ and $\dfield{\mGG}{\sigma}$ are $PS$-field extensions of $\dfield{\KK}{\sigma}$. Using the technologies given in Theorems~\ref{Thm:PiTest} and~\ref{Thm:SigmaTest} below one can show that they are all \pisiE-field extensions. Since $\const{\KK}{\sigma}=\KK$, they are also \pisiE-fields over $\KK$; compare also~\cite{OS:18}.
\end{example}

\begin{definition}\label{Def:RPiSiExt}
	Let $(\EE,\sigma)$ be an $APS$-extension of $\dfield{\AR}{\sigma}$ as defined in Definition~\ref{Def:APSExt}. Then this is called an \emph{\rpisiE-extension} if $\const{\EE}{\sigma}=\const{\AR}{\sigma}$. The arising $A$-monomials are also called \emph{\rE-monomials}, the $P$-monomials are called \emph{\piE-monomials} and the $S$-monomials are called \emph{\sigmaE-monomials}. In particular, we call it an \emph{\rE-/\piE-/\sigmaE-/\rE\piE-/\rE\sigmaE/-\piE\sigmaE-extension} if it is built by the corresponding monomials.
\end{definition}

\begin{example}[Cont. of Ex.~\ref{Exp:Qxs1}]\label{Exp:Qxs2}
Consider the difference ring $\dfield{\QQ(x)[s]}{\sigma}$ from Example~\ref{Exp:Qxs1}. Since $\ev:\QQ(x)[s]\to\QQ$ defined by~\eqref{Equ:QxsEv} and~\eqref{Equ:QxsEvs} (with $\ev'=\ev$) is an evaluation function of $\dfield{\QQ(x)[s]}{\sigma}$ we can construct the $\QQ$-homomorphism $\tau:\QQ(x)[s]\to\seqQ$ defined by~\eqref{Equ:EvToTau}. 
Since $s$ is a \sigmaE-monomial over $\QQ(x)$, we get $\const{\QQ(x)[s]}{\sigma}=\QQ$. Thus we can apply Theorem~\ref{Thm:InjectiveProp} and it follows that
$$\tau(\QQ(x))[\tau(s)]=\tau(\QQ(x))[(\ev(s,n))_{n\geq0}]=\tau(\QQ(x))[(S(n))_{n\geq0}]$$
with 
$S=\expr(s)=\Sum(1,\tfrac1x)\in\Sigma(\QQ(x))$
is isomorphic to the polynomial ring $\QQ(x)[s]$.
Further, $(S(n))_{n\geq0}$ with
$S(n)=\sum_{k=1}^n\frac1k$
is transcendental over $\tau(\QQ(x))$.
\end{example}

Example~\ref{Exp:Qxs2} generalizes as follows. Suppose that we are given a  
basic \rpisiE-extension $\dfield{\EE}{\sigma}$ of $\dfield{\GG}{\sigma}$ with $$\GG[\rho_1]\dots[\rho_l][p_1,p_1^{-1}]\dots[p_u,p_u^{-1}][s_1]\dots[s_r]$$ 
where the $\rho_i$ are \rE-monomials with $\zeta_i=\frac{\sigma(\rho_i)}{\rho_i}\in\cRR$ being primitive roots of unity, $p_i$ are \piE-monomials and the $s_i$ are \sigmaE-monomials. In addition, take an evaluation function $\ev$ with $o$-function $L$ by iterative applications of Lemmas~\ref{Lemma:EvConstruction} and~\ref{Lemma:ZConstruction}. Here we may assume that
\begin{itemize}
\item $\ev(\rho_i,n)=\zeta_i^n$ for all $1\leq i\leq l$, 
\item $\ev(p_i,n)=P_i(n)$ with $\expr(p_i)=P_i\in\Pi(\GG)$ for all $1\leq i\leq u$, and 
\item $\ev(s_i,n)=S_i(n)$ with $\expr(s_i)=S_i\in\Sigma(\GG)$ for all $1\leq i\leq r$.
\end{itemize}
Then $\tau:\EE\to\seqK$ with~\eqref{Equ:EvToTau} is a $\KK$-homomorphism. By Theorem~\ref{Thm:InjectiveProp} it follows that $\tau$ is injective and thus
\begin{eqnarray*}
\tau(\EE)=&\tau(\GG)&[\tau(\rho_1)]\dots[\tau(\rho_l)]\\
&&\times[\tau(p_1),\tau(p_1)^{-1}]\dots[\tau(p_u),\tau(p_u)^{-1}]\\
&&\times[\tau(s_1)]\dots[\tau(s_r)]\\
=&\tau(\GG)&[(\zeta_1^n)_{n\geq0}]\dots[(\zeta_l^n)_{n\geq0}]\\
&&\times[(P_1(n))_{n\geq0},(\tfrac1{P_1(n)})_{n\geq0}]\dots[(P_u(n))_{n\geq0},(\tfrac1{P_u(n)})_{n\geq0}]\\
&&\times[(S_1(n))_{n\geq0}]\dots[(S_r(n))_{n\geq0}]
\end{eqnarray*}
forms a (Laurent) polynomial ring extension over the ring of sequences $R=\tau(\GG)[(\zeta_1^n)_{n\geq0}]\dots[(\zeta_l^n)_{n\geq0}]$.
In particular, we conclude that the sequences 
\begin{align*}
&(P_1(n))_{n\geq0}, (\tfrac1{P_1(n)})_{n\geq0},\dots(P_u(n))_{n\geq0},(\tfrac1{P_u(n)})_{n\geq0},
(S_1(n))_{n\geq0},\dots,(S_r(n))_{n\geq0}
\end{align*}
are, up to the trivial relations $(P_i(n))_{n\geq0}\cdot (\tfrac1{P_i(n)})_{n\geq0}=1$ for $1\leq i\leq u$,
algebraically independent among each other over the ring $R$.

We are now ready to state the main result of this section that connects $\Sum\Prod(\GG)$ with difference ring theory.

\begin{theorem}\label{Thm:SigmaReduceInRPS}
Let $\dfield{\EE}{\sigma}$ be a basic $APS$-extension of $\dfield{\GG}{\sigma}$ with $\GG\in\{\rGG,\bGG,\mGG\}$ and $\AR=\EE\lr{t_1}\dots\lr{t_e}$ equipped with an evaluation function $\ev:\EE\times\NN\to\KK$ (using Lemmas~\ref{Lemma:EvConstruction} and~\ref{Lemma:ZConstruction}). Take the $\KK$-homomorphism $\tau:\EE\to\seqK$ with $\tau(f)=(\ev(f,n))_{n\geq0}$ and $T_i=\expr(t_i)\in\Sigma\Pi(\GG)$ for $1\leq i\leq e$. Then the following statements are equivalent.
\begin{enumerate}
\item $\dfield{\EE}{\sigma}$ is an \rpisiE-extension of $\dfield{\GG}{\sigma}$.
\item $\tau$ is a $\KK$-isomorphism between $\dfield{\EE}{\sigma}$ and $(\tau(\EE),S)$; in particular all sequences generated by the \pisiE-monomials are algebraically independent over the ring given by the sequences of $\tau(\GG)$ adjoined with the sequences generated by \rE-monomials. 
\item $W=\{T_1,\dots,T_e\}$ is canonical-reduced over $\GG$.
\item The zero recognition problem is trivial, i.e., for any $F\in\Sum\Prod(W,\GG)$ the following holds: $F=0$ if and only if $\ev(F,n)=0$ for all $n\geq\delta$ for some $\delta\in\NN$.
\end{enumerate}
\end{theorem}
\begin{proof}
$(1)\Leftrightarrow(2)$ is an immediate consequence of Theorem~\ref{Thm:InjectiveProp}.\\	
$(2)\Rightarrow(3)$: 
Let $F,F'\in\Sum\Prod(W,\GG)$ with $F(n)=F'(n)$ for all $n\geq\delta$ for some $\delta\in\NN$. Replace in $F,F'$ any occurrences of $T_i\opower z_i$ for $1\leq i\leq e$ with $z_i\in\ZZ$ by $t_i^{z_i}$, $\oplus$ by $+$, and $\odot$ by $\cdot$. This yields $f,f'\in\EE$ with $\ev(f,n)=F(n)$ for all $n\geq L(f)$ and $\ev(f',n)=F'(n)$ for all $n\geq L(f')$. Hence $\tau(f)=\tau(f')$. Since $\tau$ is injective, $f=f'$. But this implies that $F$ and $F'$ are the same up to trivial permutations of the operands in $\odot$ and $\oplus$. Consequently $W$ is canonical reduced.\\
$(3)\Rightarrow(4)$: Suppose that $W$ is canonical reduced and take $F\in\Sum\Prod(W,\GG)$ with $F(n)=0$ for all $n\geq\delta$ for some $\delta\in\NN$. Since $\ev(0,n)=0$ for all $n\geq0$ and $W$ is canonical reduced, it follows that $F=0$.\\
$(4)\Rightarrow(1)$: Suppose that $\tau$ is not injective and take $f\in\EE\setminus\{0\}$ with $\tau(f)=\vect{0}$. 
By Lemma~\ref{Lemma:SwitchAPSToA}  we can take 
$0\neq F\in\Sum\Prod(\{T_1,\dots,T_e\},\GG)$ and $\delta\in\NN$ with $\ev(f,n)=F(n)=0$ for all $n\geq\delta$. Thus statement~(4) does not hold.
\end{proof}

In order to derive the equivalences in Theorem~\ref{Thm:SigmaReduceInRPS} we assumed that an $APS$-extension is given.
We can relax this assumption if the set $W$ is shift-stable.

\begin{corollary}\label{Cor:RPS=Canonical}
	Let $W=\{T_1,\dots,T_e\}\subset\Sigma\Pi(\GG)$ be in sum-product reduced representation and shift-stable. More precisely, for each $1\leq i\leq e$ the arising sums/products in $T_i$ are contained in $\{T_1,\dots,T_{i-1}\}$. Then the following two statements are equivalent:
	\begin{enumerate}
		\item 	
		There is a basic \rpisiE-extension $\dfield{\EE}{\sigma}$ of $\dfield{\GG}{\sigma}$ with $\EE=\GG\lr{t_1}\dots\lr{t_e}$ equipped with an evaluation function $\ev$ (using Lemmas~\ref{Lemma:EvConstruction} and~\ref{Lemma:ZConstruction}) with $T_{i}=\expr(t_i)\in\Sigma\Pi(\GG)$ for $1\leq i\leq e$. 
		\item $W=\{T_1,\dots,T_e\}$ is $\sigma$-reduced over $\GG$.
	\end{enumerate}
\end{corollary}
\begin{proof}
	$(1)\Rightarrow(2)$: By assumption $W$ is sum-product reduced and shift-stable, and by $(1)\Rightarrow(3)$ of Theorem~\ref{Thm:SigmaReduceInRPS} it is canonical-reduced. Thus $W$ is $\sigma$-reduced.\\
	$(2)\Rightarrow(1)$: By Corollary~\ref{Cor:SwitchTiToAPS} we get an $APS$-extension $\dfield{\EE}{\sigma}$ of $\dfield{\GG}{\sigma}$ with $\EE=\GG\lr{t_1}\dots\lr{t_e}$ equipped with an evaluation function $\ev$ (using Lemmas~\ref{Lemma:EvConstruction} and~\ref{Lemma:ZConstruction}) with $T_{i}=\expr(t_i)\in\Sigma\Pi(\GG)$ for $1\leq i\leq e$. Since $W$ is canonical reduced, it follows by $(3)\Rightarrow(1)$ of Theorem~\ref{Thm:SigmaReduceInRPS} that $\dfield{\EE}{\sigma}$ is an \rpisiE-extension of $\dfield{\GG}{\sigma}$.
\end{proof}

Corollary~\ref{Cor:RPS=Canonical} yields immediately a strategy (actually the only strategy for shift-stable sets) to solve Problem~\textsf{SigmaReduce}.

\begin{programcode}{Strategy to solve Problem~\textsf{SigmaReduce}}\label{Method:SigmaReduce}
\hspace*{-0.1cm}\begin{minipage}[t]{1.1cm}Given:\end{minipage}\begin{minipage}[t]{10.4cm}
	$A_1,\dots,A_u\in\Sum\Prod(\GG)$ with $\GG\in\{\rGG,\bGG,\mGG\}$, i.e., $\GG=\KK(x,x_1,\dots,x_v)$ or $\GG=\KK(x_1,\dots,x_v)$.
\end{minipage}\\
\begin{minipage}[t]{1.1cm}Find:\end{minipage}\begin{minipage}[t]{10.4cm}
	a $\sigma$-reduced set $W=\{T_1,\dots,T_e\}\subset\Sigma\Pi(\GG')$ with $B_1\dots,B_u\in\Sum\Prod(W,\GG')$
	and $\delta_1,\dots,\delta_u\in\NN$ such that
	$A_i(n)=B_i(n)$ holds for all $n\geq\delta_i$ and $1\leq i\leq r$.
\end{minipage}
	\begin{enumerate}
		\item Construct		
		an \rpisiE-extension $\dfield{\EE}{\sigma}$ of\footnote{Here we get $\GG'=\KK'(x,x_1,\dots,x_v)$ or $\GG'=\KK'(x_1,\dots,x_v)$ where $\KK'$ is a field extension of $\KK$; if $A_1,\dots,A_u\in\Sum\Prod_1(\GG)$, one can restrict to the special case $\GG=\GG'$.} $\dfield{\GG'}{\sigma}$ with $\EE=\GG'\lr{t_1}\dots\lr{t_e}$ 
		 equipped with an evaluation function $\ev:\EE\times\NN\to\KK'$ and $o$-function $L$ (using Lemmas~\ref{Lemma:EvConstruction} and~\ref{Lemma:ZConstruction}) in which $A_1,\dots,A_u$ are modeled by $a_1,\dots, a_u\in\EE$. More precisely, for $1\leq i\leq u$ we compute in addition $\delta_i\in\NN$ with $\delta_i\geq L(a_i)$ such that
		\begin{equation}\label{Equ:ModelAi}
		A_i(n)=\ev(a_i,n)\quad\forall n\geq\delta_i.
		\end{equation}
		\item Set $W=\{T_1,\dots,T_e\}$ with $T_i:=\expr(t_i)\in\Sigma\Pi(\GG')$ for $1\leq i\leq e$.
		\item Set $B_i:=\expr(a_i)\in\Sum\Prod(W,\GG')$ for $1\leq i\leq u$.
		\item Return $W$, $(B_1,\dots,B_u)$ and $(\delta_1,\dots,\delta_u)$.
	\end{enumerate}

\vspace*{-0.2cm}

\end{programcode}

\noindent What remains open is to enrich this general method with the construction required in step~(1). This task will be considered in detail in the next section.

\section{The representation problem}\label{Sec:RepProblem}

In this section we will give an overview of the existing algorithms that assist in the task of solving the open subproblem given in step (1) of our general method \textsf{SigmaReduce}. The resulting machinery can be summarized as follows.

\begin{theorem}\label{Thm:SigmaReduce1}
	Given $A_1,\dots,A_u\in\Sum\Prod_1(\GG)$ with $\GG\in\{\rGG,\bGG,\mGG\}$
	where $\KK$ is a rational function field over an algebraic number field. Then one can compute
	a $\sigma$-reduced set $W=\{T_1,\dots,T_e\}\subset\Sigma\Pi_1(\GG)$ with $B_1\dots,B_u\in\Sum\Prod(W,\GG)$ and $\delta_1,\dots,\delta_u\in\NN$ such that
	$A_i(n)=B_i(n)$ holds for all $n\geq\delta_i$ and $1\leq i\leq u$.
\end{theorem}

\begin{theorem}\label{Thm:SigmaReduceNested}
	Given $A_1,\dots,A_u\in\Sum\Prod(\KK(x))$ where $\KK=\calA(y_1,\dots,y_o)$ is a rational function field over an algebraic number field $\calA$. Then one can take $\KK'=\calA'(y_1,\dots,y_o)$ where $\calA'$ is an algebraic extension of $\calA$
	and can compute
	a $\sigma$-reduced set $W=\{T_1,\dots,T_e\}\subset\Sigma\Pi(\KK'(x))$ with 
	$B_1\dots,B_u\in\Sum\Prod(W,\GG')$ and $\delta_1,\dots,\delta_u\in\NN$ such that
	$A_i(n)=B_i(n)$ holds for all $n\geq\delta_i$ and $1\leq i\leq u$.
\end{theorem}

Here we will start with the problem to represent products in \rpiE-monomials (see Subsection~\ref{Sec:ProductCase}). More precisely, we will show various tactics that enable one to represent expressions of $\Prod_1(\rGG)$, $\Prod_1(\bGG)$, $\Prod_1(\mGG)$ and $\Prod(\rGG)$. Afterwards, we will consider the problem to represent nested sums over such products (i.e., expressions of $\Sum(\GG)$, $\Sum\Prod_1(\rGG)$, $\Sum\Prod_1(\bGG)$, $\Sum\Prod_1(\mGG)$ and $\Sum\Prod(\rGG)$) in \sigmaE-monomials (see Subsection~\ref{Sec:SumCase}). 

\begin{remark}
\texttt{Sigma} can represent fully algorithmically single nested products in \rpiE-extensions; in addition, Ocansey's package \texttt{NestedProducts}~\cite{OS:18,OS:20} can deal with the case $\Prod(\rGG)$. Expressions from more general domains (e.g., sums and products that arise nontrivially in denominators) also work with the function call \MText{SigmaReduce} of \texttt{Sigma}. But for these cases the underlying summation mechanisms (like those given in Lemmas~\ref{Lemma:SwitchAPSToA} and~\ref{Lemma:SwitchAToAPS}) are only partially developed and the back translation from the difference ring setting to the term algebra might fail.
\end{remark}

In general, it suffices in our proposed construction to compute an \rpisiE-extension in which a finite set of sums and products are modeled. However, in some important instances it is possible to perform this constructions stepwise. 

\begin{definition}
	Fix $X$ as one of the term algebras $\Prod_1(\GG)$, $\Prod(\GG)$, $\Sum(\GG)$, $\Sum\Prod_1(\GG)$, $\Sum\Prod(\GG)$, and	
	let $\cal D$ be a subclass of basic \rpisiE-extensions of $\dfield{\GG}{\sigma}$. Then $\cal D$ is called \emph{X-extension-stable} if for any $\dfield{\HH}{\sigma}\in\cal D$ and any $A\in X$ one can construct an \rpisiE-extension $\dfield{\EE}{\sigma}$ of $\dfield{\HH}{\sigma}$ with $\dfield{\EE}{\sigma}\in\cal D$ and $a\in\EE$ such that one can model $A$ with $a$.
\end{definition}

\noindent We note that within such an extension-stable class of \rpisiE-extensions one does not have to treat the arising sums and products in one stroke, but can consider them iteratively. This is in particular interesting, when unforeseen sums and products arise in a later step, that have to be considered in addition.
In a nutshell, we will provide a general overview of the existing tools to design basic \rpisiE-extensions. In particular, we will emphasize the available algorithms to construct extension-stable versions.
 
\subsection{Representation of products in \rpiE-extensions}\label{Sec:ProductCase}

We start with algorithmic tools that enable one to test if a $P$-extension forms a \piE-extension. Based on these tools  we present (without further details) the existing techniques to represent a finite set of products in an \rpiE-extension.

\subsubsection{Algorithmic tests}

In~\cite[Theorem~9.1]{Schneider:10c} based on Karr's work~\cite{Karr:81,Karr:85} 
a general criterion for \piE-field extensions is elaborated. Here we present a more flexible version in the ring setting.
\begin{theorem}\label{Thm:PiTest}
Let $\dfield{\EE}{\sigma}$ be a $P$-extension of a difference ring $\dfield{\HH}{\sigma}$ with $\EE=\HH\lr{t_1}\dots\lr{t_d}$ and 
$f_i=\frac{\sigma(t_i)}{t_i}\in\HH^*$ for $1\leq i\leq d$. Suppose that 
\begin{equation}\label{Equ:SemiConstantProp}
\{g\in\HH\setminus\{0\}\mid \sigma(g)=u\,g\text{ for some }u\in\HH^*\}\subseteq\HH^*
\end{equation}
holds. Then the following statements are equivalent:
\begin{enumerate}
	\item $\dfield{\EE}{\sigma}$ is a \piE-extension of $\dfield{\HH}{\sigma}$, i.e., $\const{\EE}{\sigma}=\const{\HH}{\sigma}$.
	\item There do not exist $g\in\HH\setminus\{0\}$ and $(z_1,\dots,z_d)\in\ZZ^d\setminus\{\vect{0}\}$ with
	$$\sigma(g)=f_1^{z_1}\dots f_d^{z_d}g.$$
\end{enumerate}
\end{theorem}
\begin{proof}
$(1)\Rightarrow(2)$: Suppose that there is a $g\in\HH\setminus\{0\}$ and $(z_1,\dots,z_d)\in\ZZ^d\setminus\{\vect{0}\}$ with
$\sigma(g)=f_1^{z_1}\dots f_d^{z_d}g.$ Let $i$ be maximal such that $z_i\neq0$. Then we can take $h=g\,f_1^{-z_1}\dots f_{i-1}^{-z_{i-1}}$ and get $\sigma(h)=f_i^{z_i}\,h$. With part~(2) of Theorem~2.12 in \cite{DR1} it follows that $\dfield{\HH\lr{t_1}\dots\lr{t_i}}{\sigma}$ is not a \piE-extension of $\dfield{\HH\lr{t_1}\dots\lr{t_{i-1}}}{\sigma}$.\\
$(2)\Rightarrow(1)$: Let $i$ with $1\leq i\leq d$ be minimal such that $\dfield{\HH\lr{t_1}\dots\lr{t_i}}{\sigma}$ is not a \piE-extension of $\dfield{\HH\lr{t_1}\dots\lr{t_{i-1}}}{\sigma}$. Then $\sigma(g)=\alpha_i^{z_i}\, g$ for some $g\in\HH\lr{t_1}\dots\lr{t_{i-1}}\setminus\{0\}$ and $z_i\in\ZZ\setminus\{0\}$ by part~(2) of Theorem~2.12 in \cite{DR1}. In particular, with property~\eqref{Equ:SemiConstantProp} we can apply Theorem~22 of~\cite{DR1} and it follows that $g=h\,t_1^{-z_1}\dots t_{i-1}^{-z_{i-1}}$ for some $z_i\in\ZZ$ and $h\in\HH^*$. Thus we get $\sigma(h)=\alpha_1^{z_1}\dots\alpha_i^{z_i}\,h$ with $z_i\neq0$ which proves statement~(1). 
\end{proof}

\begin{remark}
(1) Theorem~\ref{Thm:SigmaTest} contains the following special case (see~\cite{Karr:85} for the field and~\cite{DR1} for the ring case): a $P$ extension $\dfield{\AR\lr{p}}{\sigma}$ of $\dfield{\AR}{\sigma}$ with $f:=\frac{\sigma(p)}{p}\in\AR^*$ is a \piE-extension iff there are no $g\in\AR$, $m\in\ZZ\setminus\{0\}$ with $\sigma(g)=f\,g$.\\
(2) Often Theorem~\ref{Thm:PiTest} is applied to the special case when the ground ring $\dfield{\HH}{\sigma}$ forms a field. Note that in this particular instance, the assumption~\eqref{Equ:SemiConstantProp} trivially holds.
\end{remark}

Let $\dfield{\AR}{\sigma}$ be a difference ring and $\vect{f}=(f_1,\dots,f_d)\in(\AR^*)^d$. Then we define
$$M(\vect{f},\AR):=\{(m_1,\dots,m_d)\in\ZZ^d|\,\sigma(g)=f_1^{m_1}
\dots f_d^{m_d}\,g\text{ for some }g\in \AR\setminus\{0\}\};$$
see also~\cite{Karr:81}.
Note that Theorem~\ref{Thm:PiTest} states that the $P$-extension $\dfield{\EE}{\sigma}$ of the difference ring $\dfield{\HH}{\sigma}$ with $\EE=\HH\lr{t_1}\dots\lr{t_d}$ and 
$f_i=\frac{\sigma(t_i)}{t_i}\in\HH^*$ for $1\leq i\leq d$ is a \piE-extension if and only if $M(\vect{f},\HH)=\{\vect{0}\}$.
If $\vect{f}\in([\FF^*]_{\FF}^{\HH})^d$ (which holds for $\FF^*$-basic $P$-extensions), this latter property can be checked by utilizing the following result.
\begin{theorem}\label{Thm:SolveProblemM}
	Let $\dfield{\HH}{\sigma}$ be a basic \rpisiE-extension of
	a difference field $\dfield{\FF}{\sigma}$ and $\vect{f}\in([\FF^*]_{\FF}^{\HH})^d$. Then the following holds: 
	\begin{enumerate}
	\item $M(\vect{f},\HH)$ is a $\ZZ$-module over $\ZZ^d$.
	\item If one can compute a basis of $M(\vect{h},\FF)$ for any $\vect{h}\in(\FF^*)^m$ with $m\geq1$, then one can compute a basis of $M(\vect{f},\HH)$.
	\end{enumerate}
\end{theorem}
\begin{proof}
Part (1) follows by Lemma~2.6 and Theorem~2.22 of~\cite{DR1} and part (2) by~\cite[Theorem~2.23]{DR1}.
\end{proof}

In other words, we can apply Theorem~\ref{Thm:PiTest} to test if a basic $P$-extension over $\FF$ is a \piE-extension if one can compute a basis of $M(\vect{h},\FF)$ in a difference field $\dfield{\FF}{\sigma}$. In particular, using the algorithms from~\cite{Karr:81} this is possible if $\dfield{\FF}{\sigma}$ is a \pisiE-field over $\KK$ where the constant field satisfies certain algorithmic properties.

\begin{definition}
A field $\KK$ is called \emph{$\sigma$-computable} if the following holds:
\begin{enumerate}
\item One can factorize multivariate polynomials over $\KK$;
\item given $(f_1,\dots,f_d)\in(\KK^*)^d$ one can compute for $\{(z_1,\dots,z_d)\in\ZZ^d\mid f_1^{z_1}\dots f_d^{z_d}=1\}$ a $\ZZ$-basis;
\item one can decide if $c\in\KK$ is an integer.
\end{enumerate}
\end{definition}

More precisely, the following holds if $\dfield{\FF}{\sigma}$ is a \pisiE-field over a $\sigma$-computable constant field; special cases are $\rGG$, $\bGG$ or $\mGG$ where $\KK$ is $\sigma$-computable.

\begin{corollary}\label{Cor:MPiSiRPiSi}
Let $\dfield{\EE}{\sigma}$ be a basic \rpisiE-extension of a \pisiE-field $\dfield{\FF}{\sigma}$ over $\KK$. If $\KK$ is $\sigma$-computable, one can compute a basis of $M(\vect{h},\EE)$ for any $\vect{h}\in([\FF^*]_{\FF}^{\EE})^d$ with $d\geq1$. This in particular is the case, if $\KK=\calA(y_1,\dots,y_o)$ is a rational function over an algebraic number field $\calA$.
\end{corollary}
\begin{proof}
If $\KK$ is $\sigma$-computable, it follows by~\cite[Theorem~9]{Karr:81} that one can compute a basis of $M(\vect{f},\FF)$ for any $\vect{f}\in(\FF^*)^m$ with $m\geq1$. Thus by part~2 of Theorem~\ref{Thm:SolveProblemM} one can compute a basis of $M(\vect{h},\EE)$ for any $\vect{h}\in([\FF^*]_{\FF}^{\EE})^d$ with $d\geq1$. In particular, it follows by~\cite[Thm.~3.5]{Schneider:05c} (based on the algorithm of~\cite{Ge:93}) that a rational function field over an algebraic number field is $\sigma$-computable.
\end{proof}

\begin{remark}\label{Remark:ProdSigma}
(1) By~\cite[Theorem~2.26]{DR1} Corollary~\ref{Cor:MPiSiRPiSi} is also valid for $\vect{f}\in(\{\FF^*\}_{\FF}^{\EE})^d$ in simple \rpisiE-extension defined over a \pisiE-field. As elaborated in~\cite[Sect.~2.3.3]{DR1} (using ideas of~\cite{Schneider:06d}) it holds even in the more general setting that $\dfield{\FF}{\sigma}$ is a \pisiE-field extension of a difference field $\dfield{\FF_0}{\sigma}$ where all roots of unity in $\FF$ are constants and $\dfield{\FF_0}{\sigma}$ is $\sigma$-computable; for the definition of these algorithmic properties we refer to~\cite[Def.~1]{Schneider:06d}. Further aspects can be also found in~\cite{AbraBron2000}.
In particular, all these properties hold, if $\dfield{\FF_0}{\sigma}$ is a free difference field~\cite{Schneider:06d,Schneider:06e} (covering generic/unspecified sequences $X_n$) or is built by radical extensions~\cite{Schneider:07f} (covering objects like $\sqrt{n}$). For the underlying implementations enhancing \texttt{Sigma} we refer to~\cite{Schneider:06d,Schneider:07f}.\\
(2) Within \texttt{Sigma} the case of \pisiE-fields is implemented properly where the constant field is given by a rational function field over the rational numbers. In parts also algebraic numbers work, but here we rely on sub-optimal routines of Mathematica. 
\end{remark}

\subsubsection{Algorithmic representations}

In this section we present several algorithms that provide proofs of Theorems~\ref{Thm:SigmaReduce1} and~\ref{Thm:SigmaReduceNested} if one restricts to the cases $\Prod_1(\GG)$ with $\GG\in\{\rGG,\bGG,\mGG\}$ or $\Prod(\rGG)$, i.e., if one drops expressions where sums arise. More precisely, we will introduce several solutions of step (1) for our method \textsf{SigmaReduce}. 

First, we treat the case $\Prod_1(\GG)$. In this setting (where also sums can arise) single-basic \rpisiE-extensions, a subclass of basic \rpisiE-extensions, are sufficient.

\begin{definition}
	An \rpisiE-extension $\dfield{\EE}{\sigma}$ of a difference ring $\dfield{\AR}{\sigma}$ with $\EE=\AR\lr{t_1}\dots\lr{t_e}$ is called \emph{single-basic} if for any $R$-monomial $t_i$ we have $\frac{\sigma(t_i)}{t_i}\in\const{\AR}{\sigma}^*$ and for any $P$-monomial $t_i$ we have $\frac{\sigma(t_i)}{t_i}\in\AR^*$.
\end{definition}

We will present the following two main strategies.

\medskip

\noindent$\bullet$ \textit{Optimal product representations.} In~\cite[Theorem~69]{SchneiderProd:20} we showed that one can construct \rpiE-extensions with minimal extension degree and minimal order.

\begin{theorem}\label{Thm:OptimalProdExt}
	Let $\GG\in\{\rGG,\bGG,\mGG\}$ and $A_1,\dots,A_u\in\Prod_1(\GG)$. Then there is a single-basic \rpiE-extension $\dfield{\EE}{\sigma}$ of $\dfield{\GG}{\sigma}$ with $\EE=\GG\lr{t_1}\dots\lr{t_e}$ together with an evaluation function $\ev$ and $o$-function $L$ (based on the construction given in Lemmas~\ref{Lemma:EvConstruction} and~\ref{Lemma:ZConstruction}) with the following properties:
	\begin{enumerate}
		\item $A_1,\dots,A_u$ are modeled by $a_1,\dots,a_u\in\EE$, i.e., for all $1\leq i\leq u$ we have~\eqref{Equ:ModelAi} for some explicitly given $\delta_i\in\NN$ with $\delta_i\geq L(a_i)$.
		\item There is at most one \rE-monomial in $\EE$. This implies that the order $\lambda$ is minimal among all such extensions in which one can model $a_1,\dots,a_u$.
		\item The number of \piE-monomials in $\EE$ is minimal among all such extensions in which one can model $a_1,\dots,a_u$.
	\end{enumerate}
If the constant field of $\dfield{\KK}{\sigma}$ is a rational function field over an algebraic number field, then the above components are computable.
\end{theorem}

\begin{example}\label{Exp:OptimalProdRep}
For the following products in $\Prod_1(\QQ[\iiota](x)$ with the imaginary unit $\iiota$:
\begin{align*}
A_1&=\Prod\big(1,\tfrac{-13122 x (1+x)}{(3+x)^3}\big),&
A_2&=\Prod\big(1,\tfrac{26244 x^2 (2+x)^2}{(3+x)^2}\big),\\
A_3&=\Prod\big(1,\tfrac{\iiota k (2+x)^3}{729 (5+x)}\big),&
A_4&=\Prod\big(1,\tfrac{-162 x (2+x)}{5+x}\big),
\end{align*}
we compute the alternative expressions
$B_1=\tfrac{5 (1+x)^2 (2+x)^5 (3+x)^8}{52488 (4+x) (5+x)} T_1
T_2 T_3^{-2}$,
$B_2=\tfrac{(4+x)^2 (5+x)^2}{400} T_2^2$,
$B_3=\tfrac{2754990144 (4+x)^2 (5+x)^2}{25 (1+x)^4 (2+x)^{10} (3+x)^{16}} T_3^3$ and
$B_4=T_2$ in terms of the $\sigma$-reduced set $W=\{T_1,T_2,T_3\}$ with
$$T_1=\RPow(-1),\quad T_2=\Prod(1,\tfrac{-162 x (2+x)}{5+x}),\quad T_3=\Prod(1,\tfrac{-\iiota (3+x)^6}{9 x (1+x)^2 (2+x) (5+x)});$$
internally, $T_1$ is modeled by an \rE-monomial of order $2$ and $T_2,T_3$ are modeled by two \piE-monomials. Details on this construction are given in~\cite[Ex.~70]{SchneiderProd:20}.
\end{example}

We remark that this optimal representation has one essential drawback: if further products have to be treated in a later situation, the existing difference ring cannot be reused, but a completely new difference ring has to be designed. 

\medskip

\noindent$\bullet$ \textit{Extension stable representations for completely factorizable constant fields.}
In the following we will follow another approach: instead of computing the smallest ring in which one can model a finite set of single nested products, we design a difference ring where the multiplicands are as small as possible such that the constructed difference rings are $\Prod_1(\GG)$-extension-stable. In order to accomplish this task, we will restrict the constant field $\KK$ further as follows.

A \emph{ring $R$ is called completely factorizable} if $R$ is a unique factorization domain (UFD) and all units in $R$ are roots of unity. In particular, any element $a\in R\setminus\{0\}$ can be written in the form
$a=u\,a_1^{n_1}\dots a_l^{n_l}$ with a root of unity $u$, $n_1,\dots,n_l\in\NNP$ and $a_1,\dots,a_l\in R$ being coprime irreducible elements.
In addition, a \emph{field $K$ is called completely factorizable} if it is the quotient field of a completely factorizable ring $R$. In such a field any element $a\in K^*$ can be written in the form $a=u\,a_1^{n_1}\dots a_l^{n_l}$ with a root of unity $u$,  $n_1,\dots,n_l\in\ZZ\setminus\{0\}$ and $a_1,\dots,a_l\in R$ being coprime irreducible elements. We call $K$ \emph{completely factorizable of order $\lambda\in\NN$}, if the set of roots of unity is finite and the maximal order is $\lambda$. We say that \emph{complete factorizations are computable over such a field $K$} if for any rational function from $K(x_1,\dots,x_r)$ a complete factorization can be computed.

The following lemma allows to lift the property of completely factorizable rings.

\begin{lemma}
If a ring (resp.\ field) $\AR$ is completely factorizable, the polynomial ring $\AR[x_1,\dots,x_r]$ (resp.\ rat.\ function field $\AR(x_1,\dots,x_r)$) is completely factorizable.	
\end{lemma}

\begin{example}
The ring $\ZZ$ and the Gaussian ring $\ZZ[\iiota]$ with the roots of unity $1,-1$ and $1,-1,\iiota,-\iiota$, respectively, are examples of completely factorizable rings. Thus $\ZZ$, $\ZZ[\iiota]$ and, in particular $\ZZ[x_1,\dots,x_r]$ and $\ZZ[\iiota][x_1,\dots,x_r]$ are completely factorizable rings. Furthermore, their quotient fields $\QQ$, $\QQ[\iiota]$, $\QQ(x_1,\dots,x_r)$ and $\QQ[\iiota](x_1,\dots,x_r)$ are completely factorizable of order $2$ or $4$, respectively. In particular, one can compute complete factorizations over $\QQ$ and $\QQ[\iiota]$.
\end{example}

\begin{definition}
Let $\FF$ be the quotient field of a completely factorizable ring $R$ of order $\lambda$. 
A single-basic \rpisiE-extension $\dfield{\EE}{\sigma}$ of $\dfield{\FF}{\sigma}$ with $\EE=\FF\lr{t_1}\dots\lr{t_e}$ is called \emph{completely factorized} if there is at most one \rE-monomial $\rho$ with $\frac{\sigma(\rho)}{\rho}\in(\const{\FF}{\sigma})^*$ of order $\lambda$ and for any \piE-monomial $t_i$ we have that $\frac{\sigma(t_i)}{t_i}\in R$ is irreducible.
\end{definition}

We are now ready to state the following result implemented within \texttt{Sigma}; the case $\rGG$ is covered by \cite[Theorem~2]{DR2};  the extension to $\bGG$ and $\mGG$ is straightforward.

\begin{theorem}\label{Thm:ExtensionStableProd}
	Let $\GG\in\{\rGG,\bGG,\mGG\}$ where $\KK$ is completely factorizable of order $\lambda$. Then the class of completely factorized \rpiE-extensions over $\dfield{\GG}{\sigma}$ is $\Prod_1(\GG)$-extension-stable. More precisely,
	let $\dfield{\HH}{\sigma}$ be a completely factorized \rpiE-extension of $\dfield{\GG}{\sigma}$ equipped with an evaluation function $\ev$ an $o$-function $L$. Let $A\in\Prod_1(\GG)$. Then
	there is an \rpiE-extension $\dfield{\EE}{\sigma}$ of $\dfield{\HH}{\sigma}$ with an extended evaluation function $\ev$ and $o$-function $L$ (using Lemmas~\ref{Lemma:EvConstruction} and~\ref{Lemma:ZConstruction}) with the following properties:
	\begin{enumerate}
		\item $\dfield{\EE}{\sigma}$ is a completely factorizable \rpiE-extension of $\dfield{\GG}{\sigma}$.
		\item $A$ is modeled by $a\in\EE$, i.e., $A(n)=\ev(a,n)$ for all $n\geq\delta$ for some $\delta\in\NN$.
	\end{enumerate}
	If complete factorization over $\KK$ can be computed, all components are computable.
\end{theorem}

\begin{example}
Given the products~\eqref{Exp:OptimalProdRep}, we can split the multiplicands into irreducible factors and get (after some technical details) the product representations
$B_1=\frac{216 T_1^2 T_2
		T_3^8}{(n+1)^2 (n+2)^3 (n+3)^3
		T_4}$, 
	$B_2=\frac{9 T_2^2
		T_3^8 T_4^2}{(n+3)^2}$, 
	$B_3=\frac{15
		(n+1)^2 (n+2)^2 T_1^2 T_4^3}{(n+3) (n+4) (n+5)
		T_3^6}$ and $B_4=\frac{60
		T_1^2 T_2 T_3^4
		T_4}{(n+3) (n+4) (n+5)}$
in terms of the $\sigma$-reduced set $W=\{T_1,T_2,T_3,T_4\}$ with
$$T_1=\RPow(\iiota),\quad T_2=\Prod(1,2),\quad T_3=\Prod(1,3),\quad
T_4=\Prod(1,x);$$
internally, $T_1$ is modeled by an \rE-monomial of order $4$, and $T_2,T_3,T_4$ are modeled by three \piE-monomials.
\end{example}

It would be interesting to see extension-stable difference ring constructions that work in more general settings.  A first step in this direction has been elaborated in~\cite[Theorem~6.2]{OS:18}. Here a toolbox (implemented within \texttt{NestedProducts}) is summarized where one tries to follow the above construction of completely factorized \rpiE-extensions as much as possible. In this way, a  modification of the existing \rpiE-extension will arise only for products whose multiplicands are taken from an algebraic number field.

\medskip

\noindent$\bullet$ \textit{Representation of nested products.}
We obtained the first algorithm in~\cite[Theorem~9]{OS:20} (implemented in \texttt{NestedProducts}) to represent products from $\Prod(\rGG)$ fully algorithmically in a basic \rpiE-extension. This result can be stated as follows.

\begin{theorem}\label{Thm:NestedProdExt}
	Let $\GG=\rGG=\KK(x)$ where $\KK={\cal A}(y_1,\dots,y_{o})$ with $o\geq0$ is a rational function field over an algebraic number field ${\cal A}$. Then for $A_1,\dots,A_u\in\Prod(\GG)$ one can compute a basic \rpiE-extension $\dfield{\EE}{\sigma}$ of $\dfield{\GG'}{\sigma}$ with an evaluation function $\ev$ and $o$-function $L$ (using Lemmas~\ref{Lemma:EvConstruction} and~\ref{Lemma:ZConstruction}) with the following properties:
	\begin{enumerate}
	\item  The ground field $\GG$ is extended to $\GG'=\KK'(x)$ where $\KK'={\cal A'}(y_1,\dots,y_{o})$ with ${\cal A'}$ being an algebraic field extension of ${\cal A}$.	
		\item Within the \rpiE-monomials in $\dfield{\EE}{\sigma}$ there is at most one \rE-monomial.
		\item $A_1,\dots,A_u$ are modeled by $a_1,\dots,a_u\in\EE$, i.e., for all $1\leq i\leq u$ we have~\eqref{Equ:ModelAi} for some explicitly given $\delta_i\in\NN$ with $\delta_i\geq L(a_i)$.
	\end{enumerate}
\end{theorem}

\begin{remark}
Theorem~\ref{Thm:OptimalProdExt} holds also for general ground rings $\dfield{\GG}{\sigma}$ with certain algorithmic properties; see~\cite{SchneiderProd:20}. 
Fascinating structural properties of mixed hypergeometric products (and related objects within the differential case) are presented in~\cite{ZimingLi:11}. Further simplification aspects within \pisiE-fields (e.g., finding products where the degrees of the top most sum or product in the numerator and denominator of a multiplicand are minimal) are elaborated in~\cite{Schneider:05c,Petkov:10}. In addition, methods to find algebraic relations of sequences built by products are given in~\cite{Kauers:08,Schneider:10c,Singer:16,SchneiderProd:20,OS:20}.
\end{remark}

\subsection{Representation of sums}\label{Sec:SumCase}

\subsubsection{Algorithmic tests via (parameterized) telescoping}

We will proceed as in the product case.
The additive version of Theorem~\ref{Thm:PiTest}, which is nothing else than parameterized telescoping (see Section~\ref{Sec:ParaTele}), reads as follows.

\begin{theorem}[{\cite[Thm.~7.10]{DR3}}]\label{Thm:SigmaTest}
	Let $\dfield{\EE}{\sigma}$ be an $S$-extension of a difference ring $\dfield{\HH}{\sigma}$ with $\EE=\HH\lr{t_1}\dots\lr{t_d}$ and
	$f_i=\sigma(t_i)-t_i\in\HH$ for $1\leq i\leq d$. If $\KK:=\const{\HH}{\sigma}$ is a field, then the following statements are equivalent:
	\begin{enumerate}
		\item $\dfield{\EE}{\sigma}$ is a \sigmaE-extension of $\dfield{\HH}{\sigma}$, i.e., $\const{\EE}{\sigma}=\const{\HH}{\sigma}$.
		\item There do not exist $g\in\HH$ and $(c_1,\dots,c_d)\in\KK^d\setminus\{\vect{0}\}$ with
		$$\sigma(g)-g=c_1\,f_1+\dots +c_d\,f_d.$$ 
	\end{enumerate}
\end{theorem}

Note that Theorem~\ref{Thm:SigmaTest} contains the following special case (compare~\cite{Karr:81} for the field case and~\cite{DR1} for the ring case): an $S$ extension $\dfield{\AR[s]}{\sigma}$ of $\dfield{\AR}{\sigma}$ with $f:=\sigma(s)-s\in\AR$ is a \sigmaE-extension if and only if there is no $g\in\AR$ such that the telescoping equation $\sigma(g)-g=f$ holds; this property will be crucial for the construction that establishes Theorem~\ref{Thm:SigmaExtensionStable} given below.

Let $\dfield{\AR}{\sigma}$ be a difference ring with constant field $\KK$, $u\in\AR\setminus\{0\}$ and
$\vect{f}=(f_1,\dots,f_d)\in\AR^d$. Then following~\cite{Karr:81} we define the \emph{set of solutions of parameterized first-order linear difference equations}:
$$V_1(u,\vect{f},\AR)=\{(c_1,\dots,c_d,g)\in\KK^d\times
\AR\mid\,\sigma(g)-u\,g=c_1\,f_1+\dots+c_d\,f_d\}.$$

With this notion, Theorem~\ref{Thm:SigmaTest} can be restated as follows: $\dfield{\EE}{\sigma}$ is a \sigmaE-extension of $\dfield{\HH}{\sigma}$ if and only if $V_1(1,(f_1,\dots,f_d),\HH)=\{0\}^d\times\KK$. In order to check that this is the case, we can utilize the following theorem. 

\begin{theorem}\label{Thm:SolveProblemV1}
	Let $\dfield{\HH}{\sigma}$ be a basic \rpisiE-extension of
	a difference field $\dfield{\FF}{\sigma}$ with constant field $\KK$, $u\in[\FF^*]_{\FF}^{\HH}$ and $\vect{f}\in\HH^d$. Then the following holds: 
	\begin{enumerate}
		\item $V_1(u,\vect{f},\HH)$ is a $\KK$-vector space of dimension $\leq d+1$.
		\item If one can compute a basis of
		 $M(\vect{h},\FF)$ for any $\vect{h}\in(\FF^*)^n$ and a basis of $V_1(v,\vect{h},\FF)$ for any $v\in\FF^*$, $\vect{h}\in\FF^n$, then one can compute a basis of $V_1(u,\vect{f},\HH)$.
	\end{enumerate}
\end{theorem}
\begin{proof}
Lemma~2.17 and Thm.~2.22 of~\cite{DR1} gives (1); \cite[Thm.~2.23]{DR1}\footnote{For an alternative algorithm we refer to~\cite[Section~6]{DR3}.} shows (2).
\end{proof}

In particular, we can activate this machinery if $\dfield{\FF}{\sigma}$ is a \pisiE-field over a $\sigma$-computable constant field; a special case is, e.g., $\FF=\mGG$.

\begin{corollary}\label{Cor:VPiSiRPiSi}
	Let $\dfield{\EE}{\sigma}$ be an \rpisiE-extension of a \pisiE-field $\dfield{\FF}{\sigma}$ over $\KK$. If $\KK$ is $\sigma$-computable, one can compute a basis of $V_1(1,\vect{f},\EE)$ for any $\vect{f}\in(\EE^*)^d$. This in particular is the case, if $\KK$ is a rational function field over an algebraic number field.
\end{corollary}
\begin{proof}
	If $\KK$ is $\sigma$-computable, it follows by~\cite{Karr:81} (or~\cite{Schneider:05c}) that one can compute a basis of $V_1(u,\vect{f},\FF)$ for any $u\in\FF^*$, $\vect{f}\in(\FF^*)^d$. Thus by part~2 of Theorem~\ref{Thm:SolveProblemV1} one can compute a basis of $V_1(1,\vect{h},\EE)$ for any $\vect{h}\in(\EE^*)^d$. In particular, it follows by~\cite[Thm.~3.5]{Schneider:05c} (based on the algorithm of~\cite{Ge:93}) that a rational function field over an algebraic number field is $\sigma$-computable. 
\end{proof}

\begin{remark}
(1) By~\cite[Thm.~2.26]{DR1}, Corollary~\ref{Cor:VPiSiRPiSi} is also valid for $\vect{f}\in(\{\FF^*\}_{\FF}^{\EE})^d$ in simple \rpisiE-extensions over a \pisiE-field. As elaborated in~\cite[Sect.~2.3.3]{DR1} it holds even in the more general setting where $\dfield{\FF}{\sigma}$ is a \pisiE-field extension of a difference field $\dfield{\FF_0}{\sigma}$ which is $\sigma^*$-computable (see \cite[Def.~1]{Schneider:06d}) and one can compute a basis of $V(u,\vect{f})$ in $\dfield{\FF_0}{\sigma^k}$ for any\footnote{If the extension is basic, we only need the case $k=1$.} $k>0$, $u\in\FF^*$ and $\vect{f}\in\FF_0^m$; see also Remark~\ref{Remark:ProdSigma}.(1).
\end{remark}

\subsubsection{Basic representations}\label{Sec:BasicSumRep}

The following theorem (based on Theorem~\ref{Thm:SigmaTest} and the property that one can solve the telescoping problem~\eqref{Equ:TeleDF} given below) enables one to lift the results of $\Prod_1(\GG)$ and $\Prod(\rGG)$ form Section~\ref{Sec:ProductCase} to the cases $\Sum(\GG)$, $\Sum\Prod_1(\GG)$ and $\Sum\Prod(\rGG)$.

\begin{theorem}\label{Thm:SigmaExtensionStable}
Let $\GG\in\{\rGG,\bGG,\mGG\}$ and $A_1,\dots,A_u\in\Sum\Prod(\GG)$. Let $\dfield{\HH}{\sigma}$ be a basic \rpisiE-extension of $\dfield{\GG}{\sigma}$ equipped with an evaluation function $\ev$ and an $o$-function $L$ where all arising products in $A_1,\dots,A_u$ can be modeled. Then there is a \sigmaE-extension $\dfield{\EE}{\sigma}$ of $\dfield{\HH}{\sigma}$ with an extended evaluation function $\ev$ and $o$-function $L$ (using Lemmas~\ref{Lemma:EvConstruction} and~\ref{Lemma:ZConstruction}) such that $a_1,\dots,a_u\in\EE$ model $A_1,\dots,A_u$, i.e., for all $1\leq i\leq u$ we have~\eqref{Equ:ModelAi} for some explicitly given $\delta_i\in\NN$ with $\delta_i\geq L(a_i)$.\\
If $\KK$ is $\sigma$-computable, and $L:\HH\to\NN$ and $\ev:\HH\times\NN\to\KK$ are computable, the above components can be computed.
\end{theorem}
\begin{proof}
This result follows from the construction given in~\cite[pp. 657--658]{DR3} which can be summarized as follows. We suppose that we have constructed already a basic \rpisiE-extension of $\dfield{\GG}{\sigma}$ equipped with an evaluation function $\ev$ and an $o$-function $L$ where all arising products in $A_1,\dots,A_u$ can be modeled. Then we can adapt the construction of Lemma~\ref{Lemma:SwitchAToAPS} and deal with all arising sums and products arising in the $A_1,\dots,A_u$. Suppose that we have constructed already a \sigmaE-extension $\dfield{\AR}{\sigma}$ of $\dfield{\GG}{\sigma}$ and we are treating now the product or sum $T_i$. If it is a product, we sort it out in the bookkeeping step and obtain an element $b_i\in\HH^*\subseteq\EE^*$ that models $T_i$ by assumption. Otherwise, $T_i=\Sum(\lambda,H)$. By induction (on the depth of the arising sums) we can construct a \sigmaE-extension $\dfield{\AR'}{\sigma}$ of $\dfield{\AR}{\sigma}$ together with an extended evaluation function $\ev$ and $o$-function $L$ such that we can take $h\in\AR'$ with $\ev(h,n)=H(n)$ for all $n\geq L(h)$. 
Now we enter the sum-case and perform the following extra test. We check if there is a $g\in\AR'$ with 
\begin{equation}\label{Equ:TeleDF}
\sigma(g)=g+f\quad\Leftrightarrow\quad \sigma(g)-g=f
\end{equation}
for $f:=\sigma(h)$.
Suppose there is such a $g$. We define $\delta_i:=\max(L(f),L(g),\lambda)$. Then for $b_i:=g+\sum_{j=\lambda}^{\delta_i}H(j)-\ev(g,\delta_i)\in\AR'$ we get
$\ev(b_i,n+1)-\ev(b_i,n)=\ev(g,n+1)-\ev(g,n)=H(n+1)$ and $T_i(n+1)=T_i(n)+H(n+1)$ for all $n\geq\delta_i$. Since $\ev(b_i,\delta_i)=\sum_{j=\lambda}^{\delta_i}F(j)=\ev(T_i,\delta_i)$, we get $\ev(b_i,n)=\ev(T_i,n)$ for all $n\geq\delta_i$.\\ 
Otherwise, if there is no such $g$, we proceed as in the sum-case of Lemma~\ref{Lemma:SwitchAToAPS}: we adjoin the \sigmaE-monomial $t$ to $\AR'$ with $\sigma(t)=t+f$ with $f=\sigma(h)$ and get the claimed $b_{i}=t+c$ with $c\in\KK$ such that 
$\ev(b_i,n)=\ev(T_i,n)$ holds for all $n\geq L(b_i)=\delta_i$.\\
Summarizing, we can construct a nested \sigmaE-extension in which the elements from $\Sum\Prod(\GG)$ can be modeled. If $\KK$ is $\sigma$-computable, one can decide constructively by Corollary~\ref{Cor:VPiSiRPiSi} if there exists such a $g$. Furthermore, if $L:\HH\to\NN$ and $\ev:\HH\times\NN\to\KK$ are computable also their extensions for $\dfield{\EE}{\sigma}$ are computable by recursion. Consequently, all components are computable. 
\end{proof}

We get immediately the following result for $\Sum(\GG)$-stable extensions.

\begin{corollary}
Let $\GG\in\{\rGG,\bGG,\mGG\}$. The class of \sigmaE-extensions over $\dfield{\GG}{\sigma}$ is $\Sum(\GG)$-extension-stable. More precisely, let $\dfield{\HH}{\sigma}$ be a \sigmaE-extension of $\dfield{\GG}{\sigma}$ with an evaluation function $\ev$ and an $o$-function $L$, and let $A\in\Sum(\GG)$. Then there is a \sigmaE-extension $\dfield{\EE}{\sigma}$ of $\dfield{\HH}{\sigma}$ with an extended evaluation function $\ev$ and an $o$-function $L$ (using Lemmas~\ref{Lemma:EvConstruction} and~\ref{Lemma:ZConstruction}) together with $a\in\EE$ and $\delta\in\NN$ with $A(n)=\ev(a,n)$ for all $n\geq\delta$.\\ 
If $\KK$ is $\sigma$-computable, these components can be computed.
\end{corollary}

Combining Theorems~\ref{Thm:ExtensionStableProd} and~\ref{Thm:SigmaExtensionStable} we get \texttt{Sigma}'s main translation mechanism.

\begin{corollary}\label{Cor:FactorizableFieldConstruction}
Let $\GG\in\{\rGG,\bGG,\mGG\}$ where $\KK$ is completely factorizable of order $\lambda$. Then the class of completely factorized \rpisiE-extensions is $\Sum\Prod_1(\GG)$-extension-stable. More precisely,
let $\dfield{\HH}{\sigma}$ be a completely factorized \rpisiE-extension of $\dfield{\GG}{\sigma}$ equipped with an evaluation function $\ev$ an $o$-function $L$. Let $A\in\Sum\Prod_1(\GG)$. Then
there is an \rpisiE-extension $\dfield{\EE}{\sigma}$ of $\dfield{\HH}{\sigma}$ with an extended evaluation function $\ev$ and $o$-function $L$ (using Lemmas~\ref{Lemma:EvConstruction} and~\ref{Lemma:ZConstruction}) with the following properties:
\begin{enumerate}
	\item $\dfield{\EE}{\sigma}$ is a completely factorizable \rpisiE-extension of $\dfield{\GG}{\sigma}$.
	\item $A$ is modeled by $a\in\EE$, i.e., $A(n)=\ev(a,n)$ for all $n\geq\delta$ for some $\delta\in\NN$.
\end{enumerate}
If $\KK$ is $\sigma$-computable and complete factorizations over $\KK$ can be computed, all the components can be given explicitly.
\end{corollary}
\begin{proof}
We can write $\HH=\GG\lr{t_1}\dots\lr{t_e}[s_1,\dots,s_u]$ where the $t_i$ are \rpiE-monomials and the $s_i$ are \sigmaE-monomials. 
Take all products that arise in $A$. Since $\dfield{\HH_0}{\sigma}$ with $\HH_0=\GG\lr{t_1}\dots\lr{t_e}$ is a completely factorized \rpiE-extension of $\dfield{\GG}{\sigma}$, we can apply Theorem~\ref{Thm:ExtensionStableProd} and get an \rpiE-extension $\dfield{\HH_1}{\sigma}$ of $\dfield{\HH_0}{\sigma}$ with $\HH_1=\HH_0\lr{p_1}\dots\lr{p_v}$ together with an extended evaluation function $\ev$ and $o$-function $L$ such that $\dfield{\HH_1}{\sigma}$ is a completely factorized \rpiE-extension of $\dfield{\GG}{\sigma}$ and such that all products in $A$ can be modeled in $\HH_1$. By~\cite[Cor.~6.5]{DR3} (together with~\cite[Prop~3.23]{DR3}) it follows that also  $\dfield{\HH_2}{\sigma}$ with $\HH_2=\HH\lr{p_1}\dots\lr{p_v}$ is a \piE-extension of $\dfield{\HH}{\sigma}$. In particular, $\dfield{\HH_2}{\sigma}$ is a completely factorized \rpisiE-extension of $\dfield{\GG}{\sigma}$ and we can merge the evaluation functions and $o$-functions to $\ev:\HH_2\times\NN\to\KK$ and $L:\HH_2\to\NN$. Finally, we apply Theorem~\ref{Thm:SigmaExtensionStable} and get a \sigmaE-extension $\dfield{\EE}{\sigma}$ of $\dfield{\HH_2}{\sigma}$ with an appropriately extended evaluation function $\ev$ and $o$-function $L$ together with $a\in\EE$ and $\delta\in\NN$ such that $\ev(a,n)=A(n)$ holds for all $n\geq\delta$. By definition $\dfield{\EE}{\sigma}$ is a completely factorized \rpisiE-extension of $\dfield{\GG}{\sigma}$.\\
If $\KK$ is $\sigma$-computable and one can compute complete factorizations over $\KK$, Theorems~\ref{Thm:ExtensionStableProd} and~\ref{Thm:SigmaExtensionStable} are constructive and all components can be computed.
\end{proof}

Furthermore, combining Theorems~\ref{Thm:OptimalProdExt} and~\ref{Thm:SigmaExtensionStable} gives the following result (we omit the optimality properties given in Theorem~\ref{Thm:OptimalProdExt}).

\begin{corollary}\label{Cor:NestedFieldConstruction}
	Let $\GG\in\{\rGG,\bGG,\mGG\}$ where $\KK$ is built by a rational function field defined over an algebraic number field.
	Then for $A_1,\dots,A_u\in\Sum\Prod_1(\GG)$
	there is a single-basic \rpisiE-extension $\dfield{\EE}{\sigma}$ of $\dfield{\GG}{\sigma}$ together with an extended evaluation function $\ev:\EE\times\NN\to\KK$ and $o$-function $L:\EE\to\NN$ (using Lemmas~\ref{Lemma:EvConstruction} and~\ref{Lemma:ZConstruction}) with the following properties: $A_1,\dots,A_u$ are modeled by $a_1,\dots,a_u\in\EE$, i.e., for all $1\leq i\leq u$ we have~\eqref{Equ:ModelAi} for some explicitly given $\delta_i\in\NN$ with $\delta_i\geq L(a_i)$.
\end{corollary}

In addition, the applications of Theorems~\ref{Thm:NestedProdExt} and~\ref{Thm:SigmaExtensionStable} yield the following statement.

\begin{corollary}
	Let $\rGG=\KK(x)$ with $\KK={\cal A}(y_1,\dots,y_{o})$ be a rational function field over an algebraic number field ${\cal A}$. Then for $A_1,\dots,A_u\in\Sum\Prod(\GG)$ there is a basic \rpiE-extension $\dfield{\EE}{\sigma}$ of $\dfield{\rGG'}{\sigma}$ with an evaluation function $\ev:\EE\times\NN\to\KK'$ and $o$-function $L:\EE\to\NN$ (using Lemmas~\ref{Lemma:EvConstruction} and~\ref{Lemma:ZConstruction}) with the following properties:
\begin{enumerate}
	\item  The ground field $\rGG$ is extended to $\rGG'=\KK'(x)$ where $\KK'={\cal A'}(y_1,\dots,y_{o})$ with ${\cal A'}$ being an algebraic field extension of ${\cal A}$.	
	\item Within the \rpisiE-monomials in $\dfield{\EE}{\sigma}$ there is at most one \rE-monomial.
	\item $A_1,\dots,A_u$ are modeled by $a_1,\dots,a_u\in\EE$, i.e., for all $1\leq i\leq u$ we have~\eqref{Equ:ModelAi} for some explicitly given $\delta_i\in\NN$ with $\delta_i\geq L(a_i)$.
\end{enumerate}
\end{corollary}

In particular, activating our method \MText{SigmaReduce} in combination with Corollaries~\ref{Cor:FactorizableFieldConstruction} and~\ref{Cor:NestedFieldConstruction} establishes Theorems~\ref{Thm:SigmaReduce1} and~\ref{Thm:SigmaReduce1}, respectively.

Most of the above results are implemented within the summation package \texttt{Sigma} or are available by using in addition the package \texttt{NestedProducts}. Further details can be found in the following remark.

\begin{programcode}{Technical details of the summation package {\tt Sigma}}

\vspace*{-0.8cm}

\begin{remark}\label{Remark:ControlTower}
	(1) Within \texttt{Sigma} the function call \texttt{SigmaReduce} follows the method given on page~\pageref{Method:SigmaReduce}. 
	Note that in this construction the $\sigma$-reduced set $W$ is constructed by treating stepwise the sums and products that occur in the $A_i$. \\
	(2)	The user can control the $\sigma$-reduced set $W$ manually by introducing extra sums and products with the option {\tt Tower$\to\{S_1,\dots,S_v\}$} that will be parsed before the arising sums in $A_1,\dots,A_u$ are considered; as an example we refer to \myIn{\ref{MMA:TowerExp}} in Example~\ref{Exp:SigmaReduceVar}. \\
	(3) \texttt{Sigma} is tuned for expressions from $\Sum\Prod_1(\GG)$ where the constant field $\KK$ is a completely factorizable field. In particular for the case that $\KK$ is a rational function field over the rational numbers, the machinery given in Corollary~\ref{Cor:FactorizableFieldConstruction} is highly robust. \texttt{Sigma} also works partially with rational function fields over algebraic number fields; but  here it depends on the stability of the subroutines in Mathematica.\\ 
	(4) For nested products the machinery of \texttt{SigmaReduce} works if the objects can be transformed straightforwardly to \rpisiE-extensions. For more complicated situations the objects $\Sum\Prod(\rGG)$ can be handled fully algorithmically in combination with Ocansey's package \texttt{NestedProducts}. 
\end{remark}

\vspace*{-0.3cm}

\end{programcode}

\begin{remark}
As observed in~\cite{Bluemlein:04} an algebraically independent basis of certain classes of indefinite nested sums can be obtained by exploiting the underlying quasi-shuffle algebra. In~\cite{Bluemlein:04} this aspect has been utilized for the class of harmonic sums, and it has been enhanced for generalized, cyclotomic and binomial sums in~\cite{ABS:11,ABS:13,ABRS:14}. Later it has been shown in~\cite{AS:18} that the relations in the class of cyclotomic harmonic sums produced by difference ring theory (compare Theorem~\ref{Thm:InjectiveProp}) and by the quasi-shuffle algebra are equivalent. As a consequence, the quasi-shuffle algebra of cyclotomic sums induces a canonical representation.
We emphasize that many of the above aspects can be carried over to a summation theory of unspecified sequences~\cite{PS:19}.
\end{remark}

\subsubsection{Depth-optimal representations}\label{Sec:DepthOptimalRep}

In~\cite{Schneider:05f,Schneider:08c} we have refined Karr's definition of \pisiE-field extensions to depth-optimal \pisiE-field extensions and have developed improved telescoping algorithms therein. In this way, we could provide a general toolbox in~\cite{Schneider:10b} that can find representations such that the nesting depths of the arising sums are minimal. As it turns out, the underlying telescoping algorithms can be adapted (and even simplified) for \rpisiE-extensions. 
For the specification of the refined representation (without entering into technical details) we need the following definition.

\begin{definition}
A finite set $W\subset\Sigma\Pi(\GG)$ is called \emph{depth-optimal} if for any $G\in\Sum\Prod(W,\GG)$ and $G'\in\Sum\Prod(\GG)$ with $G(n)=G'(n)$ for all $n\geq\delta$ for some $\delta\in\NN$ it follows that $\delta(G)\leq\delta(G')$ holds.
\end{definition}

Then combining the results from Section~\ref{Sec:BasicSumRep} with the tools from~\cite{Schneider:05f,Schneider:08c,Schneider:10b} we obtain algorithms that can solve the following problem if $\KK$ is $\sigma$-computable; for simplicity we skipped the general case $\Sum\Prod(\GG)$. Further technical details concerning the implementation in \texttt{Sigma} can be found in Remark~\ref{Remark:ControlTower}.

\begin{programcode}{Problem DOS: \textsf{Depth-optimal SigmaReduce}}
\hspace*{-0.1cm}\noindent\begin{minipage}[t]{1.1cm}Given:\end{minipage}\begin{minipage}[t]{10.4cm}
$A_1,\dots,A_u\in\Sum\Prod_1(\mGG)$.
\end{minipage}\\
\noindent\begin{minipage}[t]{1.1cm}Find:\end{minipage}\begin{minipage}[t]{10.4cm}
 a finite $\sigma$-reduced depth-optimal set $W\subset\Sigma\Pi_1(\mGG)$ together with $B_1,\dots,B_u\in\Sum\Prod(W,\mGG)$ and $\delta_1,\dots,\delta_u\in\NN$ such that $A_i(n)=B_i(n)$ holds for all $n\geq\delta_i$ and $1\leq i\leq u$
\end{minipage}\\[0.2cm]
\end{programcode}

\begin{example}\quad
Given the sums $A_1,A_2,A_3\in\Sum(\QQ(x))$ defined by
\begin{mma}
	\In \{A_1,A_2,A_3\}=\Bigg\{\sum_{k=1}^n \frac{\displaystyle\Big(
		\ssumB{i=1}k \frac{1}{i^2}\Big) 
		\ssumB{i=1}k \frac{(-1)^i}{i}}{1+k},\sum_{k=1}^n \frac{\displaystyle\Big(
		\ssumB{i=1}k \frac{1}{i^2}\Big) 
		\ssumB{i=1}k \frac{(-1)^i}{i}}{2+k},\sum_{k=1}^n \frac{\displaystyle\Big(
		\ssumB{i=1}k \frac{1}{i^2}\Big) 
		\ssumB{i=1}k \frac{(-1)^i}{i}}{3+k}
	\Bigg\};\\
\end{mma}
\noindent we get the alternative expressions $B_1,B_2,B_3\in\Sum\Prod(W,\QQ(x))$ by executing
\begin{mma}
	\In \{B_1,B_2,B_3\}=SigmaReduce[\{A_1,A_2,A_3\},n]\\
	\Out \Bigg\{
	\sum_{k=1}^n \frac{\displaystyle\Big(
		\ssumB{i=1}k \frac{1}{i^2}\Big) 
		\ssumB{i=1}k \frac{(-1)^i}{i}}{1+k},
	\sum_{k=1}^n \frac{\displaystyle\Big(
		\ssumB{i=1}k \frac{1}{i^2}\Big) 
		\ssumB{i=1}k \frac{(-1)^i}{i}}{2+k},\newline
	\frac{3}{16}
	+\frac{(-3-2 n) (-1)^n}{8 (1+n) (2+n)}
	+\frac{(-1)^n}{2 (2+n)}\sum_{i=1}^n \frac{1}{i^2}+
	\frac{1}{2} 
	\sum_{i=1}^n \frac{\displaystyle(-1)^i}{i^2}
	+\frac{
		-3+2 n+2 n^2}{4 (1+n) (2+n)}\sum_{i=1}^n \frac{(-1)^i}{i}\newline
	-\frac{(1+n) (5+2 n)}{2 (2+n) (3+n)}\Big(
	\sum_{i=1}^n \frac{1}{i^2}\Big) 
	\sum_{i=1}^n \frac{(-1)^i}{i}
	+\frac{1}{2} 
	\sum_{i=1}^n \frac{\displaystyle\Big(
		\ssumB{j=1}i \frac{1}{j^2}\Big) 
		\ssumB{j=1}i \frac{(-1)^j}{j}}{1+i}
	+\frac{1}{2} 
	\sum_{i=1}^n \frac{\displaystyle\Big(
		\ssumB{j=1}i \frac{1}{j^2}\Big) 
		\ssumB{j=1}i \frac{(-1)^j}{j}}{2+i}
	\Bigg\}\\
\end{mma}
\noindent with the $\sigma$-reduced set
$$W=\Big\{\sum_{k=1}^n\frac1{k^2},\sum_{k=1}^n\frac{(-1)^k}{k},\sum_{i=1}^n\frac{(-1)^i}{i^2},\sum_{k=1}^n \frac{\displaystyle\Big(
	\ssumB{i=1}k \frac{1}{i^2}\Big) 
	\ssumB{i=1}k \frac{(-1)^i}{i}}{1+k},\sum_{k=1}^n \frac{\displaystyle\Big(
	\ssumB{i=1}k \frac{1}{i^2}\Big) 
	\ssumB{i=1}k \frac{(-1)^i}{i}}{2+k}\Big\}.$$
Note: instead of $A_3$ (a sum of nesting depth 3) the simpler sum $\sum_{i=1}^n\frac{(-1)^i}{i^2}$ (with nesting depth $2$) has been introduced automatically.
\end{example}

\begin{remark}
Further refined \pisiE-extensions, such as reduced \pisiE-extensions, have been elaborated in~\cite{Schneider:10a} 
(based on improved telescoping algorithms given in~\cite{Schneider:04a,Schneider:15}).
\end{remark}

%
%
%

\section{The summation paradigms}\label{Sec:DefiniteProblem}

We have explained in detail how sums and products can be modeled automatically within \rpisiE-extensions. Thus steps~1 and~3 on page~\eqref{BasicSigmaStrategy} are settled and we focus on step~2: We will introduce the summation paradigms in difference rings and fields; further details how these problems are handled in \texttt{Sigma} are given below.

\subsection{Refined telescoping}\label{Sec:RefinedTelescoping}

As indicated in Section~\ref{Sec:BasicSumRep}, in particular in Theorem~\ref{Thm:SigmaExtensionStable}, the construction of basic \rpisiE-extensions for the representation of $\Sum\Prod(\GG)$ is based on algorithms that solve the telescoping problem~\eqref{Equ:TeleDF}. In particular, the quality of the constructed extensions and the used telescoping algorithms are mutually intertwined. As illustrated for instance in Section~\ref{Sec:DepthOptimalRep}, the underlying telescoping algorithms could be refined further (using~\cite{Schneider:05f,Schneider:08c,Schneider:10b}) to compute depth-optimal representations.

In the following we will focus on the available telescoping technologies in \texttt{Sigma} (based on
~\cite{Schneider:04a,Schneider:05f,Schneider:07d,Schneider:08c,Schneider:10a,Schneider:10b,Schneider:10c,Schneider:15})
that enable one to simplify sums further. For simplicity we will focus on sums from $\Sigma\Pi_1(\mGG)$ and skip, e.g., the case $\Sigma\Pi(\rGG)$.

\begin{programcode}{Problem RT: Refined Telescoping}
	\hspace*{-0.1cm}\noindent\begin{minipage}[t]{1.1cm}Given:\end{minipage}\begin{minipage}[t]{10.4cm}
		$F\in\Sum\Prod_1(\mGG)$.
	\end{minipage}\\
	\noindent\begin{minipage}[t]{1.1cm}Find:\end{minipage}\begin{minipage}[t]{10.4cm}
		$\delta\in\NN$ and 
		a $\sigma$-reduced set $W=\{T_1,\dots,T_e\}\subset\Sigma\Pi_1(\mGG)$ where $\depth(T_1)\leq\depth(T_2)\leq\dots\leq\depth(T_e)$ together with 
		$F',G\in\Sum\Prod(W,\mGG)$ such that for all $k\geq\delta$ we have 
		$F(k)=F'(k)$ and

\vspace*{-0.2cm}

		$$G(k+1)-G(k)=F'(k).$$
	\end{minipage}
	\begin{itemize}
		\item\textsf{Refinement 1:} $W$ is depth-optimal (by using \MText{\small SimplifyByExt$\to$MinDepth}).
		\item\textsf{Refinement 2:} In addition, if $\depth(G)=\depth(F')+1$, then $\depth(T_{e-1})<\depth(T_e)=\depth(G)$ and $T_e=\Sum(\delta,H)$ with $H\in\Sum\Prod( \{T_1,\dots,T_{i}\},\mGG)$ where $i$ with $1\leq i<e$ is minimal (by using \MText{\small SimplifyByExt$\to$DepthNumber}).
		\item\textsf{Refinement 3:} One can compute, among all possible choices with $i$ minimal, $H$ such that also $\deg_{T_i}$ is minimal (by using \MText{\small SimplifyByExt$\to$DepthNumberDegree}).

	\vspace*{-0.3cm}

	\end{itemize} 
\end{programcode}

\noindent Given such $G$ and $\delta\in\NN$ for $F$ we obtain the simplification~\eqref{Equ:TeleSummed} for all $n\geq\delta$.

\begin{example}\label{Exp:SigmaReduceVar}
	We start with the following sum:
	\begin{mma}
		\In mySum1=\sum_{k=1}^n \Big(
		\sum_{j=1}^k \frac{(-1)^j}{j^2}
		\Big)
		\Big(
		\sum_{j=1}^k \frac{(-1)^j}{j}\Big)^2;\\
	\end{mma}
	\noindent Telescoping without any refinements (by setting  \MText{SimplifyByExt$\to$None}) does not yield a simplification.
	\noindent However, by activating the first refinement with the option \MText{SimplifyByExt$\to$MinDepth} (which actually is the default option) we get
	\begin{mma}
		\In SigmaReduce[mySum1,n,SimplifyByExt\to MinDepth]\\
		\Out \frac{1}{3} 
		\sum_{i=1}^n \frac{(-1)^i}{i^3}
		+(-1)^{1+n} \Big(
		\sum_{j=1}^n \frac{(-1)^j}{j^2}\Big) 
		\sum_{j=1}^n \frac{(-1)^j}{j}
		+(1+n) \Big(
		\sum_{j=1}^n \frac{(-1)^j}{j^2}
		\Big)
		\Big(
		\sum_{j=1}^n \frac{(-1)^j}{j}\Big)^2
		-\frac{1}{3} \Big(
		\sum_{j=1}^n \frac{(-1)^j}{j}\Big)^3\\
	\end{mma}
	We illustrate the second refinement with the sum:
	\begin{mma}	
		\In mySum2=\sum_{k=1}^n \Big(
		\sum_{j=1}^k \frac{(-1)^j}{j^2}
		\Big)
		\Big(
		\sum_{j=1}^k \frac{(-1)^j}{j}\Big)^3;\\
	\end{mma}
	\noindent \begin{mma}\MLabel{MMA:TowerExp}
		\In SigmaReduce[mySum2,n,SimplifyByExt\to DepthNumber,Tower\to\Big\{\ssumB{i=1}n\frac{(-1)^i}{i},\ssumB{i=1}n\frac{(-1)^i}{i^2}\Big\}\newline \hspace*{1.5cm}SimpleSumRepresentation\to False]\\
		\Out \frac{1}{4} \Big(
		\sum_{j=1}^n \frac{(-1)^j}{j^2}\Big)^2
		-\frac{3}{2} (-1)^n \Big(
		\sum_{j=1}^n \frac{(-1)^j}{j^2}
		\Big)
		\Big(
		\sum_{j=1}^n \frac{(-1)^j}{j}\Big)^2
		+(1+n) \Big(
		\sum_{j=1}^n \frac{(-1)^j}{j^2}
		\Big)
		\Big(
		\sum_{j=1}^n \frac{(-1)^j}{j}\Big)^3\newline
		-\frac{1}{4} 
		\sum_{i=1}^n \Big(
		\frac{1}{i^4}
		-\frac{6 \Big(
			\ssumB{j=1}i \frac{(-1)^j}{j}\Big)^2}{i^2}
		+\frac{4 (-1)^i \Big(
			\ssumB{j=1}i \frac{(-1)^j}{j}\Big)^3}{i}
		\Big)\\
	\end{mma}
	\noindent Namely, within the given extension (specified by \MText{Tower}$\to\big\{\ssumB{i=1}n\frac{(-1)^i}{i},\ssumB{i=1}n\frac{(-1)^i}{i^2}\big\}$, compare Remark~\ref{Remark:ControlTower}) we find a sum extension which is free of $\ssumB{i=1}n\tfrac{(-1)^i}{i^2}$. Without the option \MText{SimpleSumRepresentation$\to$False} further simplifications on the found sum (using in addition partial fraction decomposition) are applied and one gets:
	\begin{mma}
		\In SigmaReduce[mySum2,n,SimplifyByExt\to DepthNumber,Tower\to\Big\{\ssumB{i=1}n\frac{(-1)^i}{i},\ssumB{i=1}n\frac{(-1)^i}{i^2}\Big\}]\\
		\Out -\frac{1}{4} 
		\sum_{i=1}^n \frac{1}{i^4}
		+\frac{1}{4} \Big(
		\sum_{j=1}^n \frac{(-1)^j}{j^2}\Big)^2
		-\frac{3}{2} (-1)^n \Big(
		\sum_{j=1}^n \frac{(-1)^j}{j^2}
		\Big)
		\Big(
		\sum_{j=1}^n \frac{(-1)^j}{j}\Big)^2\newline
		+(1+n) \Big(
		\sum_{j=1}^n \frac{(-1)^j}{j^2}
		\Big)
		\Big(
		\sum_{j=1}^n \frac{(-1)^j}{j}\Big)^3
		+\frac{3}{2} 
		\sum_{i=1}^n \frac{\Big(
			\ssumB{j=1}i \frac{(-1)^j}{j}\Big)^2}{i^2}
		-
		\sum_{i=1}^n \frac{(-1)^i \Big(
			\ssumB{j=1}i \frac{(-1)^j}{j}\Big)^3}{i}\\
	\end{mma}
	\noindent If one changes the order of the extension with the option \MText{Tower}$\to\big\{\ssumB{i=1}n\frac{(-1)^i}{i^2},\ssumB{i=1}n\frac{(-1)^i}{i}\big\}$, no simplification is possible with the option \MText{SimplifyByExt$\to$DepthNumber}. However, using the option \MText{SimplifyByExt$\to$DepthNumberDegree} one finds a sum extension where in the summand the degree w.r.t.\ $T=\ssumB{i=1}n\frac{(-1)^i}{i}$ is minimal. In this case we find
	\begin{mma}
		\In SigmaReduce[mySum2,n,SimplifyByExt\to DepthNumberDegree,Tower\to\Big\{\ssumB{i=1}n\frac{(-1)^i}{i^2},\ssumB{i=1}n\frac{(-1)^i}{i}\Big\}\newline \hspace*{1.5cm}SimpleSumRepresentation\to False]\\
		\Out -\frac{3}{2} (-1)^n \Big(
		\sum_{j=1}^n \frac{(-1)^j}{j^2}
		\Big)
		\Big(
		\sum_{j=1}^n \frac{(-1)^j}{j}\Big)^2
		+(1+n) \Big(
		\sum_{j=1}^n \frac{(-1)^j}{j^2}
		\Big)
		\Big(
		\sum_{j=1}^n \frac{(-1)^j}{j}\Big)^3
		-\frac{1}{4} \Big(
		\sum_{j=1}^n \frac{(-1)^j}{j}\Big)^4\newline
		+\frac{1}{4} 
		\sum_{i=1}^n \Big(
		-\frac{3}{i^4}
		+\frac{2 (-1)^i}{i^2}\ssumB{j=1}i \frac{(-1)^j}{j^2}
		+\frac{4 (-1)^i}{i^3}\ssumB{j=1}i \frac{(-1)^j}{j}
		\Big)\\
	\end{mma}
	\noindent where in the summand of the found sum the degree w.r.t.\ $T$ is $1$. With the option \MText{SimpleSumRepresentation$\to$True} (which is the standard option) this sum is simplified further (by splitting it into atomics by partial fraction decomposition) and we get:
	\begin{mma}
		\In SigmaReduce[mySum2,n,SimplifyByExt\to DepthNumberDegree,Tower\to\Big\{\ssumB{i=1}n\frac{(-1)^i}{i^2},\ssumB{i=1}n\frac{(-1)^i}{i}\Big\}\\
		\Out -\frac{1}{2} 
		\sum_{i=1}^n \frac{1}{i^4}
		+\frac{1}{4} \Big(
		\sum_{j=1}^n \frac{(-1)^j}{j^2}\Big)^2
		-\frac{3}{2} (-1)^n \Big(
		\sum_{j=1}^n \frac{(-1)^j}{j^2}
		\Big)
		\Big(
		\sum_{j=1}^n \frac{(-1)^j}{j}\Big)^2\newline
		+(1+n) \Big(
		\sum_{j=1}^n \frac{(-1)^j}{j^2}
		\Big)
		\Big(
		\sum_{j=1}^n \frac{(-1)^j}{j}\Big)^3
		-\frac{1}{4} \Big(
		\sum_{j=1}^n \frac{(-1)^j}{j}\Big)^4
		+
		\sum_{i=1}^n \frac{(-1)^i}{i^3}\sum_{j=1}^i \frac{(-1)^j}{j}\\
	\end{mma}
\end{example}

\subsection{Parameterized telescoping (including creative telescoping)}\label{Sec:ParaTele}

The summation paradigm of telescoping can be generalized as follows.

\begin{programcode}{Problem PT: Parameterized Telescoping}
	\hspace*{-0.1cm}\noindent\begin{minipage}[t]{1.1cm}Given:\end{minipage}\begin{minipage}[t]{10.4cm}
	$F_1,\dots,F_d\in\Sum\Prod(\GG)$ with $\GG\in\{\rGG,\bGG,\mGG\}$. \end{minipage}\\
	\noindent\begin{minipage}[t]{1.1cm}Find:\end{minipage}\begin{minipage}[t]{10.4cm}
Find, if possible, a suitable $\sigma$-reduced finite set $W\subset\Sigma\Pi(\GG')$ and $\delta\in\NN$ with the following properties;
as in Problem~\textsf{SigmaReduce}, one might have to extend the constant field $\KK$ of $\GG$ to $\KK'$ yielding $\GG'$.

	\end{minipage}

\vspace*{-0.2cm}

\begin{itemize}
	\item One can take $F'_1,\dots,F'_d\in\Sum\Prod(W,\GG')$ such that for all $1\leq i\leq d$ and all $k\geq\delta$ we have $F_i(k)=F'_i(k)$;
	\item one can take $c_1,\dots,c_d\in\KK'$ with $c_1\neq0$ and $G\in\Sum\Prod(W,G')$ such that for all $k\geq\delta$ we have
	\begin{equation}\label{Equ:ParaTele}
	G(k+1)-G(k)=c_1\,F'_1(k)+\dots+c_d\,F'_d(k).
	\end{equation}

\vspace*{-0.7cm}

\end{itemize}
\end{programcode}
\noindent Given such $c_1,\dots,c_d\in\KK$, $G$ and $\delta\in\NN$ for $F_1,\dots,F_d$, we obtain
\begin{equation}\label{Equ:ParaSumRel}
c_1\,\sum_{k=\delta}^n F_1(k)+\dots+c_d\,\sum_{k=\delta}^n F_d(k)=G(n+1)-G(\delta)
\end{equation}
for all $n\geq\delta$. In particular, if one is given a bivariate sequence $F(n,k)$ with $F_i(k)=F(n+i-1,k)\in\Sum\Prod(\GG)$ for $i=1,\dots,d$, equation~\eqref{Equ:ParaTele} turns into~\eqref{Equ:Crea}.
In particular, the sum relation~\eqref{Equ:ParaSumRel} can can be transformed to the recurrence~\eqref{Equ:RecGeneral} for the sum $S(n)=\sum_{k=\delta}^nF(n,k)$. Summarizing, parameterized telescoping contains creative telescoping~\cite{Zeilberger:91} as a special case.

A straightforward solution to the above problem can be obtained by the application of Theorem~\ref{Thm:SolveProblemV}.  In the context of $\sigma$-reduced sets this can be rephrased as follows.

\begin{proposition}\label{Prop:SolvePLDEForW}
	Let $W=\{T_1,\dots,T_e\}\subseteq\Sigma\Pi(\GG)$ be $\sigma$-reduced where for each $1\leq i\leq e$ the arising sums and products within $T_i$ are contained in $\{T_1,\dots,T_{i-1}\}$ and are in sum-product reduced form. Let $F'_1,\dots,F'_d\in\Sum\Prod(W,\GG)$. Then one can compute, in case of existence, $(c_1,\dots,c_d)\in\KK^d$ with $c_1\neq0$ together with $G\in\Sum\Prod(W,\GG)$ and $\delta\in\NN$ such that~\eqref{Equ:ParaTele}
	holds for all $k\geq\delta$.
\end{proposition}
\begin{proof}
By Corollary~\ref{Cor:RPS=Canonical} we get an \rpisiE-extension $\dfield{\EE}{\sigma}$ of $\dfield{\GG}{\sigma}$ with $\EE=\GG\lr{t_1}\dots,\lr{t_e}$ together with an evaluation function $\ev$ and $o$-function $L$ with $\expr(t_i)=T_i$ for all $1\leq i\leq e$. In particular, we get $\vect{f}=(f_1,\dots,f_d)\in\EE^d$ with $\ev(f_i,k)=F'_i(k)$ for all $1\leq i\leq u$ and all $n\geq L(f_i)$. Note that $(c_1,\dots,c_d,G)\in\KK^d\times\Sum\Prod(W,\GG)$ with~\eqref{Equ:ParaTele} for all $k\geq\delta$ for some $\delta\in\NN$ iff $(c_1,\dots,c_d,g)\in\KK^d\times\EE$.
By Theorem~\ref{Thm:SolveProblemV} we can compute a basis $V=V_1(1,\vect{f},\EE)$ and can check if there is $(c_1,\dots,c_d,g)\in V$ with $c_1\neq0$. If this is not the case, then there is no $(c_1,\dots,c_d,G)\in\KK^d\times\Sum\Prod(W,\GG)$ with $c_1\neq0$. Otherwise, we 
rephrase the result as $(c_1,\dots,c_d,G)\in\KK^d\times\Sum\Prod(W,\GG)$ such that~\eqref{Equ:ParaTele} holds for all $k\geq\delta$ with $\delta=\max(L(F'_1),\dots,L(F'_u),L(G))$.
\end{proof}

\begin{remark}\label{Remark:ParaSigmaRemark}
In Proposition~\ref{Prop:SolvePLDEForW} we assume that the input expressions from $\Sum\Prod(\GG)$ can be rephrased directly in an \rpisiE-extension. If this is not the case, the representation machinery has to be applied in a preprocessing step. To support this construction, the user can control the $\sigma$-reduced set $W$ as outlined in the Remark~\ref{Remark:ControlTower}.(2) above. But this should be done with care in order to avoid useless results.
If $W$ contains, e.g., $T_j\in\Sum(l,F'_1)$, one gets trivially $G=T_j$ and $(c_1,c_2,\dots,c_d)=(1,0,\dots,0)$.
\end{remark}

\begin{example}\label{Exp:DefiniteSummation1} 
	We activate Proposition~\ref{Prop:SolvePLDEForW} to apply Zeilberger's creative telescoping paradigm. Take the summand $F(n,k)$ defined in
\begin{mma}
\In F=\frac{(-1)^k}{k}\binom{n}{k} 
	\sum_{i=1}^k \frac{1}{i}
		\sum_{j=1}^i \frac{1}{j
			+n
	};\\
\end{mma}
\noindent and define the definite sum
\begin{mma}\MLabel{MMA:definiteSum}
	\In definiteSum=SigmaSum[F,\{k,1,n\}]\\
	\Out \sum_{k=1}^n\frac{(-1)^k}{k}\binom{n}{k} 
	\sum_{i=1}^k \frac{1}{i}
	\sum_{j=1}^i \frac{1}{j
		+n
	}\\	
\end{mma}	
\noindent Then we can compute a linear recurrence for \MText{SUM[n]=definiteSum} with the call
\begin{mma}\MLabel{MMA:RecByCrea}
	\In rec=GenerateRecurrence[definiteSum,n,SimplifyByExt\to None]\\
	\Out \Big\{(1+n)^3 (8+3 n)^2 SUM[n]
	+\Big(
	-1692-4306 n-4369 n^2-2202 n^3-549 n^4-54 n^5\Big) SUM[n+1]\newline
	+(7+3 n) \Big(
	554+1072 n+764 n^2+237 n^3+27 n^4\Big) SUM[n+2]\newline
	-2 (3+n)^2 (5+2 n) (5+3 n)^2 SUM[n+3]
	==\frac{808+2008 n+2007 n^2+1017 n^3+261 n^4+27 n^5}{(2+n)^2 (3+n)}\Big\}\\
\end{mma}
\noindent Here \texttt{Sigma} searches for a solution of~\eqref{Equ:Crea} with $d=0,1,2,\dots$ and finally computes a solution for $d=3$. 
Internally, it takes the shifted versions $F(n+i,k)$ with $i=0,1,2,3$
\begin{mma}
\In FList=\{F,(F/.n\to n+1),(F/.n\to n+2),(F/.n\to n+3)\};\\
\end{mma}
\noindent and rewrites the expressions in a $\sigma$-reduced representation:
\begin{mma}
\In FListRed=SigmaReduce[FList,k]\\
\Out \Bigg\{\frac{(-1)^k}{k}\binom{n}{k} 
\sum_{i=1}^k \frac{1}{i}
\sum_{j=1}^i \frac{1}{j
	+n
},\newline
(-1)^k \binom{n}{k}\Bigg(
\tfrac{1}{(1+n) (1
-k
+n
) (1
+k
+n
)}
+\frac{1}{k (-1
+k
-n
)}\sum_{i=1}^k \frac{1}{i
+n
}
+\frac{(-1-n)}{k (-1
+k
-n
)}\sum_{i=1}^k \frac{1}{i}\sum_{j=1}^i \frac{1}{n
+j
}
\Bigg),\dots\Bigg\}\\
\end{mma}
\noindent Here we have printed only the first two entries of the output list. Afterwards it activates Proposition~\ref{Prop:SolvePLDEForW} by executing the command  
\begin{mma}
\In ParameterizedTelescoping[FListRed,n]\\
\Out \{\{0,0,0,0,1\},\{c_1,c_2,c_3,c_4,G\}\}\\
\end{mma}
\noindent The expressions $c_1,c_2,c_3,c_4$ and $G$ equal
\begin{align*}
c_1&=-(1 + n)^3 (8 + 3 n)^2,\\
c_2&=1692 + 4306 n + 4369 n^2 + 2202 n^3 + 
549 n^4 + 
54 n^5,\\
c_3&=-(7 + 3 n) (554 + 1072 n + 764 n^2 + 237 n^3 + 
27 n^4),\\
c_4&= 2 (3 + n)^2 (5 + 2 n) (5 + 3 n)^2,\\
G(n,k)=&(-1)^k \binom{n}{k}\Bigg(Q_1 \sum_{i=1}^k\frac{1}{i}\sum_{j=1}^i\frac{1}{n
		+j
}+Q_2\,\sum_{i=1}^k\frac{1}{i
		+n
}+Q_3\Bigg)
\end{align*}	
for some $Q_1,Q_2,Q_3\in\QQ(n,k)$. Alternatively, \MText{ParameterizedTelescoping[FList,k]} (without \MText{SigmaReduce} as a preprocessing step) could be used. The same result could be produced with \MText{CreativeTelescoping[definiteSum,n,SimplifyByExt$\to$ None]}.

Finally, summing~\eqref{Equ:Crea} with $d=3$ over $k$ from $0$ to $n$ yields the recurrence given in~\myOut{\ref{MMA:RecByCrea}}.
Note that the correctness of the solution $(c_1,c_2,c_3,c_4,G)$ of~\eqref{Equ:Crea} with $d=4$ can be verified straightforwardly: Since $W$ is $\sigma$-reduced, one simply has to plug in the solutions and checks that the left-hand and right-hand sides agree. Thus we have shown rigorously that the definite sum given in~\myIn{\ref{MMA:definiteSum}} is a solution of~\myOut{\ref{MMA:RecByCrea}}. 
\end{example}

In order to introduce refined methods, we need the following definition. 

\begin{definition}
Let $W\subset\Sigma\Pi(\GG)$ be $\sigma$-reduced depth-optimal and $\vect{F}'=(F'_1,\dots,F'_d)\in\Sum\Prod_1(W,\GG)^d$. $W$ is called \emph{$\vect{F}'$-one complete} if the following holds: If there is $(c_1,\dots,c_d,G)\in\KK^{d}\times\Sum\Prod_1(\GG)$ with $c_1\neq0$, $\depth(G)\leq \min(\depth(F'_1),\dots,\depth(F'_d))$ such that~\eqref{Equ:ParaTele} holds for all $n$ sufficiently large, then there is $G'\in\Sum\Prod_1(W,\GG)$ with the same $c_i$ such that\footnote{Since $W$ is depth-optimal, it follows in particular that $\depth(G')\leq\depth(G)$.}~\eqref{Equ:ParaTele} holds ($G$ replaced by $G'$) for all $n$ sufficiently large.
\end{definition}

Using the techniques from~\cite{Schneider:04a,Schneider:05f,Schneider:07d,Schneider:08c,Schneider:10a,Schneider:10b,Schneider:10c,Schneider:15} the following refined parameterized telescoping techniques are available for the class $\Sum\Prod_1(\mGG)$ over a $\sigma$-computable field $\KK$; for simplicity we skip more general cases, like $\Sum\Prod(\rGG)$.

\begin{programcode}{Problem RPT: Refined Parameterized Telescoping}
	\hspace*{-0.1cm}\noindent\begin{minipage}[t]{1.1cm}Given:\end{minipage}\begin{minipage}[t]{10.4cm}
		$F_1,\dots,F_d\in\Sum\Prod_1(\mGG)$. \end{minipage}\\
	\noindent\begin{minipage}[t]{1.1cm}Find:\end{minipage}\begin{minipage}[t]{10.4cm}
		$\delta\in\NN$ and 
		a depth-optimal $\sigma$-reduced set $W=\{T_1,\dots,T_e\}\subset\Sigma\Pi_1(\mGG)$ with $\depth(T_1)\leq\depth(T_2)\leq\dots\leq\depth(T_e)$  
		with the following properties:
	\end{minipage}
	\begin{itemize}
		\item One gets $\vect{F'}=(F'_1,\dots,F'_d)\in\Sum\Prod(W,\mGG)^d$ such that 	
		$F_i(k)=F'_i(k)$ holds for all $1\leq i\leq d$ and $k\geq\delta$.
\end{itemize}
		In addition, based on the refinements given below, one obtains $(c_1,\dots,c_d,G)\in\KK^d\times\Sum\Prod(W,\mGG)$ with $c_1\neq1$
		such that~\eqref{Equ:ParaTele} holds for all $k\geq\delta$.
	\begin{itemize}
		\item\textsf{Refinement 1:} $W$ is $\vect{F'}$-one complete. Further, one can compute (it it exists) such a solution with  $\depth(G)\leq \depth(F'_1)$ (by using \MText{\small SimplifyByExt$\to$MinDepth}). 
		\item\textsf{Refinement 2:} If this is not possible, one gets $\depth(G)=\depth(F'_1)+1$ with the following extra property: $\depth(T_{e-1})<\depth(T_e)=\depth(G)$ and $T_e=\Sum(\delta,H)$ with $H\in\Sum\Prod( \{T_1,\dots,T_{i}\},\mGG)$ where $i$ with $1\leq i<e$ is minimal (by using the option \MText{\small SimplifyByExt$\to$DepthNumber}).
		\item\textsf{Refinement 3:} One can compute, among all possible choices with $i$ minimal, $H$ such that also $\deg_{T_i}$ is minimal (by using \MText{\small SimplifyByExt$\to$DepthNumberDegree}).
	\end{itemize}

\vspace*{-0.3cm}

\end{programcode}

\noindent For technical details concerning \texttt{Sigma} we refer to Remarks~\ref{Remark:ControlTower} and~\ref{Remark:ParaSigmaRemark} above.

\begin{example}
While the standard approach finds for the definite sum given in~\myIn{\ref{MMA:definiteSum}} only a recurrence of order $3$, the refined parameterized telescoping toolbox (refinement 1) computes a recurrence of order $1$:
\begin{mma}
\In GenerateRecurrence[definiteSum,n,SimplifyByExt\to MinDepth]\\
\Out \Big\{SUM[n]
-SUM[n+1]
==\frac{1}{(1+n)^3}
-\frac{1}{2 (1+n)^2}\sum_{i=0}^n \frac{(-1)^i \binom{n}{i}}{1
	+i
	+n
}
+\frac{1}{1+n}\sum_{i=1}^n \frac{(-1)^i \binom{n}{i} 
	\ssumB{j=1}i \frac{1}{n
		+j
}}{i}
\Big\}\\
\end{mma}
\noindent by  introducing in addition the sum $\sum_{i=0}^n \frac{(-1)^i \binom{n}{i}}{1+i+n}$. The right-hand side is given by definite sums which are simpler than the input sum. In this situation, they can be simplified further to 
$$\frac{1}{(1+n)^3}
-\frac{1}{2 (1+n)^2 (1+2 n)}\frac{1}{\binom{2 n}{n}}
+\frac{
	1}{1+n}\sum_{i=1}^n \frac{1}{i^2}
-\frac{3}{1+n}\sum_{i=1}^n \frac{1}{i^2 \binom{2 i}{i}}
$$
in $\Sum\Prod_1(\QQ(x))$
by applying again the creative telescoping paradigm plus recurrence solving (which we will introduce in the next subsection).
\end{example}

This refined version turns out to be highly valuable in concrete applications. First, one can discover in many problems the minimal recurrence relation. Sometimes this enables one even to read off hypergeometric series solutions, like, e.g., in~\cite{PS:03}. In addition, the calculation of such recurrences of lower order is more efficient, and the extra time to simplify the more complicated right hand sides is often negligible. 
In applications from particle physics, like in~\cite{BigSums1,BigSums2}, the standard approach is even out of scope and only our improved methods produced the desired results.

\begin{remark}
	(1) Structural theorems (together with algorithmic versions) that are strongly related to Liouville's theorem of integration~\cite{Liouville:1835,Rosenlicht:68} can be found in~\cite{Schneider:10a}.\\
	(2) Based on~Theorems~\ref{Thm:InjectiveProp} and~\ref{Thm:SigmaTest}  
	  additional aspects of the algebraic independence of indefinite nested sums
	(related to~\cite{Singer:08}) are worked out in~\cite{Schneider:10c} and~\cite[Section~7.2]{DR3}. Namely, if there is no solution of a parameterized telescoping problem (in particular of a creative telescoping problem), then the indefinite sums defined over these parameters are algebraically independent.
\end{remark}

\subsection{Recurrence solving}

Finally, we turn to difference ring algorithms that solve parameterized higher-order linear difference equations.
Let $\dfield{\AR}{\sigma}$ be a difference ring with constant field $\KK$, $\vect{a}=(a_0,\dots,a_m)\in\AR^{m+1}$ and
$\vect{f}=(f_1,\dots,f_d)\in\AR^d$. Then we define~\cite{Karr:81}
\begin{multline*}
V(\vect{a},\vect{f},\AR)=\{(c_1,\dots,c_d,g)\in\KK^d\times
\AR\mid\\
a_{m}\,\sigma^{m}(g)+\dots+a_1\,\sigma(g)+a_0\,g=c_1\,f_1+\dots+c_d\,f_d\};
\end{multline*}
note  that we have $V((-u,1),\vect{f},\AR)=V_1(u,\vect{f},\AR)$. 

In \texttt{Sigma} algorithms are available to solve parameterized linear difference equations that are based on the following theorem.

\begin{theorem}\label{Thm:SolveProblemV}
	Let $\dfield{\EE}{\sigma}$ be a basic \rpisiE-extension of
	a \pisiE-field $\dfield{\FF}{\sigma}$ over $\KK$, $\vect{0}\neq\vect{a}=(a_0,\dots,a_m)\in\FF^{m+1}$ and $\vect{f}\in\EE^d$. Then the following holds: 
	\begin{enumerate}
		\item $V(\vect{a},\vect{f},\EE)$ is a $\KK$-vector space of dimension $\leq m+d$.
		\item If $\KK$ is $\sigma$-computable, then one can compute a basis of $V(\vect{a},\vect{f},\EE)$.
	\end{enumerate}
\end{theorem}
\begin{proof}
	(1) follows by a slight variant of \cite[Prop~3.1.1]{Schneider:01} and~\cite[Thm.~XII (page 272)]{Cohn:65}. By~\cite[Theorem~9]{ABPS:20} (based on~\cite{Abramov:89a,Petkov:92,Bron:00,Schneider:04b,Schneider:05a,Schneider:05b}) the statement~(2) holds for $\vect{f}\in\FF^n$. Thus with~\cite{AS:21} statement~(2) holds also for $\vect{f}\in\EE^n$.
\end{proof}

In addition, \texttt{Sigma} contains a solver that finds all hypergeometric solutions in the setting of \pisiE-fields. This result follows by Theorems~9 and~10 of~\cite{ABPS:20}, which can be considered as the differential version of Singer's celebrated algorithm~\cite{Singer:91} that finds Liouvillian solutions of linear differential equations with Liouvillian coefficients.

\begin{theorem}\label{Thm:SolveHyper}
	Let $\dfield{\FF}{\sigma}$ be a \pisiE-field over a $\sigma$-computable $\KK$. Let $a_0,\dots,a_m\in\FF$ with $a_0\,a_m\neq 0$. Then one can compute a $P$-extension $\dfield{\EE}{\sigma}$ of $\dfield{\FF}{\sigma}$ with $\EE=\FF\lr{t_1}\dots\lr{t_e}$ and $\frac{\sigma(t_i)}{t_i}\in\FF^*$ and finite sets $\emptyset\neq S_i\subset\FF^*$ for $1\leq i\leq e$ as follows. 
	\begin{enumerate}
		\item 
		For any $1\leq i\leq e$ and any $h\in S_i$ it follows that $g=h t_i$ is a solution of 
		\begin{equation}\label{Equ:HomSol}
		a_{m}\,\sigma^{m}(g)+\dots+a_1\,\sigma(g)+a_0\,g=0.
		\end{equation}
		\item For any difference ring extension $\dfield{\HH}{\sigma}$ of $\dfield{\FF}{\sigma}$ with $\HH=\FF\lr{p_1}\dots\lr{p_u}$ and $\frac{\sigma(p_i)}{p_i}\in\FF^*$	
		and any solution $g\in\HH$ of~\eqref{Equ:HomSol} with  $\alpha=\frac{\sigma(g)}{g}\in\FF^*$ one can take $i\in\{1,\dots,e\}$ with $f_1,\dots,f_l\in S_i$ and $c_1,\dots,c_l\in\KK^*$ such that $\frac{\sigma(g')}{g'}=\alpha$ holds for $g'=(c_1\,h_1+\dots+c_l\,h_l)t_i$.
	\end{enumerate}
\end{theorem}

We note that the obtained solver of hypergeometric solutions covers the special cases $\rGG$ (see~\cite{Petkov:92,vanHoeij:99}), $\bGG$ with $v=1$ (see~\cite{APP:98}) and $\mGG$ (see~\cite{Bauer:99}).

\begin{remark}
	Theorems~\ref{Thm:SolveProblemV} and~\ref{Thm:SolveHyper} hold in the more general setting where $\dfield{\FF}{\sigma}$ is a \pisiE-field extension of a difference field $\dfield{\FF_0}{\sigma}$ where certain properties are satisfied (see \cite[Def.~7]{ABPS:20}). In addition, there is a generalization of Theorem~\ref{Thm:SolveProblemV} given in~\cite{AS:21} (based on~\cite{DR3,ABPS:20}) where the $a_i$  (with some extra properties) can be taken from the ring $\EE$; the implementation can be found in the Mathematica package \texttt{PLDESolver}.  
\end{remark}

Based on~\cite{Abramov:94,Abramov:96} we obtain the following result to find all d'Alembertian solutions, a subclass of Liouvillian solutions~\cite{Singer:99}. The solver relies on~\cite[Cor~2.1]{BRS:16} and~\cite[Alg.~4.5.3]{Schneider:01} and the algorithmic machinery of Theorems~\ref{Thm:SolveHyper} and~\ref{Thm:SolveProblemV}.

\begin{programcode}{Problem PLDE: Solving Parameterized Linear Difference Equations}
	\hspace*{-0.1cm}\noindent\begin{minipage}[t]{1.1cm}Given:\end{minipage}\begin{minipage}[t]{10.4cm}
$a_0,\dots,a_m\in\GG$ with $a_m\neq0$ and $F_1,\dots,F_d\in\Sum\Prod_1(\GG)$ with $\GG\in\{\rGG,\bGG,\mGG\}$, i.e., $\GG=\KK(x,x_1,\dots,x_v)$ (or $\GG=\KK(x_1,\dots,x_v)$)
where $\KK=\calA(y_1,\dots,y_o)(q_1,\dots,q_v)$ is a rational function field over an algebraic number field $\calA$.
\end{minipage}\\[0.2cm]
	\noindent\begin{minipage}[t]{1.1cm}Find:\end{minipage}\begin{minipage}[t]{10.4cm}
$\delta\in\NN$, a finite $\sigma$-reduced set $W\subset\Sigma\Pi_1(\GG')$ and\\
$B=\{(c_{i,1},\dots,c_{i,d},G_i)\}_{1\leq i\leq \nu}\subseteq\KK^d\times\Sum\Prod(W,\GG')$  
such that 
$$a_{m}(n)\,G_i(n+m)+\dots+a_{0}(n)G_i(n)=c_{i,1}F_1(n)+\dots+c_{i,d}F_d(n)$$
holds for all  $n\geq\delta$ with $1\leq i\leq\nu$;
here $\GG'=\KK'(x,x_1,\dots,x_v)$ (or $\GG'=\KK'(x_1,\dots,x_v)$) with $\KK'=\calA'(y_1,\dots,y_o)(q_1,\dots,q_v)$ where $\calA'$ is an algebraic field extension of $\calA$.
\end{minipage}\\[0.2cm]
In addition, the following properties hold:
\begin{enumerate}
	\item Completeness:
	For any $\GG''=\KK''(x,x_1,\dots,x_v)$ (or $\GG''=\KK''(x_1,\dots,x_v)$) with $\KK''=\calA''(y_1,\dots,y_o)(q_1,\dots,q_v)$ where $\calA''$ is an algebraic extension of $\calA$
	 and
	$(c_1,\dots,c_d,G)\in\KK^d\times\Sum\Prod_1(\GG'')$ with 

\vspace*{-0.6cm}

	$$(a_{m}(n)\,G(n+m)+\dots+a_{0}(n)G(n))_{n\geq0}=(c_{1}F_1(n)+\dots+c_{d}F_d(n))_{n\geq0}$$

\vspace*{-0.1cm}

\noindent there is a $(\kappa_1,\dots,\kappa_{\nu})\in(\KK'')^{\nu}$ with

\vspace*{-0.6cm}

	\begin{align*}
	(c_1,\dots,c_d)&=\kappa_1(c_{1,1},\dots,c_{1,d})+\dots+\kappa_{\nu}(c_{\nu,1},\dots,c_{\nu,d}),\\
	(G(n))_{n\geq0}&=(\kappa_1 G_1(n)+\dots+\kappa_{\nu} G_{\nu}(n))_{n\geq0}.
	\end{align*}

\vspace*{-0.1cm}

	\item Linear independence:
	If there is a $(\kappa_1,\dots,\kappa_{\nu})\in(\KK')^{\nu}$ with

\vspace*{-0.6cm}

	\begin{align*}
	\kappa_1(c_{1,1},\dots,c_{1,d})+\dots+\kappa_{\nu}(c_{\nu,1},\dots,c_{\nu,d})&=0,\\
	(\kappa_1 G_1(n)+\dots+\kappa_{\nu} G_{\nu}(n))_{n\geq0}&=\vect{0},
	\end{align*}

\vspace*{-0.3cm}

	\noindent then $(\kappa_1,\dots,\kappa_{\nu})=\vect{0}$.

\vspace*{-0.1cm}

\end{enumerate}

\vspace*{-0.3cm}

\end{programcode}

\begin{remark}
Right from the start the case $d=1$ was available (and fully solved with~\cite{ABPS:20}) in \texttt{Sigma} with the function call \MText{SolveRecurrence[$a_0G[n]+\dots+a_m G[n+m]==F_1,G[n]$]}.
The case $d>1$ has been incorporated in \texttt{Sigma} only recently. It can be carried out with \MText{SolveRecurrence[$a_0G[n]+\dots+a_m G[n+m]==\{F_1,\dots,F_d\},G[n]$]} or \MText{SolveRecurrenceList[$\{a_0,\dots,a_m\},\{F_1,\dots,F_d\},n$]}. It works also for nested products, i.e., $F_1,\dots,F_d\in\Sum\Prod(\GG)$, if the $F_i$ can be expressed straightforwardly in an \rpisiE-extension. Using in addition the package \texttt{NestedProducts} this toolbox works also fully algorithmically for the case $\Sum\Prod(\rGG)$.
\end{remark}

\begin{example}\label{Exp:DefiniteSummationStep2}[Cont.~of Ex.~\ref{Exp:DefiniteSummation1}] We proceed with the calculations given in Example~\ref{Exp:DefiniteSummation1}. We apply our solver in \texttt{Sigma} to the already computed recurrence~\myOut{\ref{MMA:RecByCrea}} and get

\begin{mma}
\In recSol = SolveRecurrence[rec[[1]], SUM[n]]\\
\Out \Big\{\{0, 1\},\{0 , 
	\sum_{i=1}^n \frac{1}{i}\},
	\{0 , \frac{4}{9} 
	\sum_{i=1}^n \frac{i!^2}{i^3 (2i)!}
	+\frac{4}{3} \Big(
	\sum_{i=1}^n \frac{1}{i}\Big) 
	\sum_{i=1}^n \frac{i!^2}{i^2 (2i)!}
	-\frac{4}{3} 
	\sum_{i=1}^n \frac{\displaystyle
		i!^2\ssumB{j=1}i \frac{1}{j}}{i^2 (2i)!}\},
	\{1 , -
	\sum_{i=1}^n \frac{\displaystyle
		\ssumB{j=1}i \frac{1}{j^2}}{i}\} \Big\}\\
\end{mma}
\noindent The first three entries provide three linearly independent solutions of the homogeneous version of the recurrence and the last entry gives a particular solution of the recurrence itself.
\end{example}

\begin{remark}
(1) By default the found solutions are represented in a depth-optimal $\sigma$-reduced set $W$ to keep the nesting depth of the solutions as small as possible.\\ 
(2) Since all components (i.e., $a_i,F_i,G_i)$ can be represented in the given $\sigma$-reduced set $W$, the correctness of the solutions $G_i$ can be verified by plugging them into the recurrence and checking if the left-hand and right-hand sides are equal.\\
(3) If one finds $m$ linearly independent solutions of the homogeneous version together with a particular solution, the solution space is fully determined. In particular, any sequence, which is a solution of the recurrence, can be represented by $\Sum\Prod(\GG)$: simply combine the found solutions accordingly (which is always possible from a certain point on) such that the evaluation of the expression agrees with the first $m$ initial values.
\end{remark}

\begin{example}[Cont.~of Ex.~\ref{Exp:DefiniteSummationStep2}]\label{Exp:FindId}
In Example~\ref{Exp:DefiniteSummationStep2} we found all solutions of the recurrence. Since also the definite sum given in~\myIn{\ref{MMA:definiteSum}} is a solution of the recurrence, we can combine the solutions accordingly and get an alternative solution of the input sum:

\vspace*{-0.2cm}

\begin{mma}
\In sol = FindLinearCombination[recSol, definiteSum, n, 3]\\
\vspace*{-0.01cm}
\Out 3 \Big(
\sum_{i=1}^n \frac{1}{i}\Big) 
\sum_{i=1}^n \frac{i!^2}{i^2 (2i)!}
-
\sum_{i=1}^n \frac{
	\ssumB{j=1}i \frac{1}{j^2}}{i}
-3 
\sum_{i=1}^n \frac{
	i!^2\ssumB{j=1}i \frac{1}{j}}{i^2 (2i)!}
+
\sum_{i=1}^n \frac{i!^2}{i^3 (2i)!}\\
\end{mma}
\noindent Finally, we can rewrite the result in terms of the central binomial coefficient with
\begin{mma}
\In sol = SigmaReduce[sol, n, Tower \to \{SigmaBinomial[2 n, n]\}]\\
\Out 3 \Big(
\sum_{i=1}^n \frac{1}{i}\Big) 
\sum_{i=1}^n \frac{1}{i^2 \binom{2 i}{i}}+\sum_{i=1}^n \frac{1}{i^3 \binom{2 i}{i}}
-
\sum_{i=1}^n \frac{1}{i}\sum_{j=1}^i \frac{1}{j^2}
-3 
\sum_{i=1}^n \frac{1}{i^2 \binom{2 i}{i}}
\sum_{j=1}^i \frac{1}{j}
\\
\end{mma}

\noindent Summarizing we have discovered and proved the identity

\vspace*{-0.6cm}

\begin{equation*}
\sum_{k=1}^n\frac{(-1)^k\binom{n}{k}}{k} 
\sum_{i=1}^k \frac{\ssumB{j=1}i \frac{1}{j
		+n
}}{i}
\\
=3 \Big(
\sum_{i=1}^n \frac{1}{i}\Big) 
\sum_{i=1}^n \frac{1}{i^2 \binom{2 i}{i}}
-
\sum_{i=1}^n \frac{
	\ssumB{j=1}i \frac{1}{j^2}}{i}
-3 
\sum_{i=1}^n \frac{
	\ssumB{j=1}i \frac{1}{j}}{i^2 \binom{2 i}{i}}
+
\sum_{i=1}^n \frac{1}{i^3 \binom{2 i}{i}}.
\end{equation*}
\end{example}

\begin{example}
More generally, using the algorithms from~\cite{ABPS:20} we can solve recurrences where the coefficients are represented within a \pisiE-field. E.g., for the recurrence
\begin{mma}
\In recFactorial=-F[n+2]
+(1+n) \big(
8+9 n+2 n^2\big)n!F[n+1]
-2 (1+n)^3 (3+n)n!^2 F[n] 
=0;\\
\end{mma}
\noindent where the coefficients are taken from $\Sum\Prod_1(\QQ(x))$, we can find all its solutions (in this instance, they are again from $\Sum\Prod_1(\QQ(x))$) by executing the \texttt{Sigma}-call
\begin{mma}
\In SolveRecurrence[recFactorial,F[n]]\\
\Out \Big\{\big\{0,\prod_{i=1}^{n}i!\big\},\big\{0,-2^n n! \prod_{i=1}^{n} i!
+\tfrac{3}{2}
\prod_{i=1}^{n}i! 
\sum_{i=1}^{n} 2^i i!\big\}\Big\}\\
\end{mma}
\end{example}	

\vspace*{-0.4cm}

\section{Application: Evaluation of Feynman integrals}\label{Sec:Applications}

The elaborated summation tools from above contributed to highly nontrivial applications, e.g., in the research areas of combinatorics, number theory and particle physics. Here we emphasize the following striking aspects that are most relevant for the treatment of Feynman integrals.

\medskip

\noindent\textbf{Multi-summation.} In order to support the user for the evaluation of definite multi-sums to expressions in $\Sum\Prod(\GG)$, the package \texttt{EvaluateMultiSums}~\cite{Schneider:13a,Schneider:13b} 

\begin{mma}
\In << EvaluateMultiSums.m \\
\vspace*{-0.1cm}
\Print \LoadP{EvaluateMultiSums by Carsten Schneider
	\copyright\ RISC-JKU}\\
\end{mma}

\noindent has been developed to tackle definite sums in one stroke. It uses as backbone \texttt{Sigma} with all the available tools  introduced above. E.g., by executing
\begin{mma}
\In EvaluateMultiSum[\sum_{k=1}^n\frac{(-1)^k}{k}\binom{n}{k} 
\sum_{i=1}^k \frac{1}{i}
\sum_{j=1}^i \frac{1}{j
	+n
}, \{\}, \{n\}, \{1\}, \{\infty\}]\\
\Out 3 \Big(
\sum_{i=1}^n \frac{1}{i}\Big) 
\sum_{i=1}^n \frac{i!^2}{i^2 (2i)!}
-
\sum_{i=1}^n \frac{
	\ssumB{j=1}i \frac{1}{j^2}}{i}
-3 
\sum_{i=1}^n \frac{
	i!^2\ssumB{j=1}i \frac{1}{j}}{i^2 (2i)!}
+
\sum_{i=1}^n \frac{i!^2}{i^3 (2i)!}\\
\end{mma}
\noindent we reproduce the identity given in Example~\ref{Exp:FindId}. In particular, it can tackle definite multi-sums by zooming from inside to outside and, in case that this is possible, transforming stepwise the sums to expressions in $\Sum\Prod_1(\GG)$. 
In this way we could treat highly complicated massive 3-loop Feynman integrals. More precisely, using techniques described in~\cite{BKSF:12} these integrals can be transformed to several thousands of multiple sums with summands from $\Prod_1(\rGG)$. Afterwards, the package \MText{SumProduction}~\cite{Schneider:12a,Schneider:13b} is applied. It combines these sums to few (but large) sums tailored for our summation toolbox. Afterwards the command \MText{EvaluateMultiSum} can be applied (without any further interaction) to treat the obtained sums. In the course of these calculations, we treated
up to sevenfold multi-sums~\cite{Schneider:16b} or fourfold sums with up to 1GB of size~\cite{BigSums1,BigSums2}. In addition, this package helped significantly to solve problems from combinatorics~\cite{SS:17,KS:20,Absent:20}. 

In addition, the difference field/ring approach described in this article has been united with important parts of the holonomic approach~\cite{Zeilberger:90a,Chyzak:00} in~\cite{Schneider:05d}. While its first main application arose in combinatorics~\cite{APS:05}, this combined toolbox has been improved further in~\cite{BRS:16} and enabled us to tackle various multi-sums coming from particle physics~\cite{Physics2,CoupledSys1,Physics4}. In addition, these improved tools have been applied in~\cite{SZ:21} to complicated multi-sums that arose in the context of irrationality proofs of $\zeta(4)$. We remark further, that also other multi-sum and integral techniques from~\cite{BKSF:12,Schneider:16b,Ablinger:21} have been explored; for further technologies see also~\cite{TechnologyQCD} and the references therein.

\medskip

\noindent\textbf{Solving coupled systems.}
Using integration-by-parts methods~\cite{Chetyrkin:1981qh,Laporta:2001dd} one can represent physical expressions in terms of master integrals which can be calculated by solving recursively defined coupled systems of linear differential equations. Most of these master integrals can be represented in terms of power series. Utilizing the techniques from above, this gives rise to two general tactics to compute the physical expressions in terms of known special functions (in case that this is possible).

\textit{Uncoupling and solving the underlying recurrences.} In the first approach we uncouple iteratively the systems of linear differential equation using Gerhold's package \texttt{OreSys}~\cite{ORESYS} and reduce the problem to solving scalar linear differential equations of each master integral $I(x)=\sum_{n=0}^{\infty}F(n)x^n$. In a first step, each linear differential equation can be transformed to a linear recurrence. Applying \texttt{Sigma}'s recurrence solver in a second step enables one to decide constructively if the coefficient $F(n)$ can be expressed in terms of $\Sum\Prod_1(\GG)$. If this is possible for each master integral, one can express also the physical expressions in $\Sum\Prod_1(\GG)$. Using these technologies implemented in the package \texttt{SolveCoupledSystem}~\cite{SolveCoupledM1:16,SolveCoupledM2:16} (using \texttt{Sigma}) 
we could treat highly nontrivial problems of particle physics as given in~\cite{CoupledSys1,CoupledSys2,CoupledSys3,CoupledSys4,CoupledSys5}. Note that there are also other methods available~\cite{Henn:2013pwa,Lee:2014ioa} that can solve certain classes of systems. Furthermore, in ongoing investigations nontrivial methods are developed to solve the coupled systems directly without recourse to uncoupling methods; see~\cite{Barka99,MS:2018,vanhoeij2020family} and the literature therein.

\textit{The large moment method.} The second highly successful approach is based on the technology~\cite{LM:17,LMR:19} implemented within the package \texttt{SolveCoupledSystem}. It enables one to produce for the master integrals the first coefficients $F(n)$ with $n=0,\dots,\mu$; so far we encountered cases where $\mu=10.000$ was necessary. Here one does not solve the arising recurrences as proposed above, but uses them to produce a large number of sequence values; as starting point one needs in addition a few initial values that can be produced by our summation tools or procedures like {\tt Mincer} \cite{Mincer:91} or {\tt MATAD} \cite{Steinhauser:2000ry}. A significant feature of the large moment method is that one can avoid complicated function spaces (either nested sums with high weight or new classes, like nested sums over, e.g., elliptic functions~\cite{Elliptic1,Elliptic2,Elliptic3,EllipticBook}) during the calculation. Only in the last step, one combines all the calculations and gets large moments of the physical expressions. Then one can use, e.g., the package \texttt{ore\_algebra}~\cite{GSAGE} in Sage to guess
recurrences (so far up to order 40) that specify precisely the different components of the physical problem. Finally,  
one can decide algorithmically if the physical problem (or individual subexpressions) can be represented within the class $\Sum\Prod_1(\rGG)$. 
In this way we could compute, e.g., the 3-loop splitting functions~\cite{CoupledSys5}, the polarized 3-loop anomalous dimensions~\cite{PhysicsNew1} and the massive 2- and 3-loop form factor~\cite{FORMF3,PhysicsNew2}; for another case study see, e.g.,~\cite{BKKS:09}.

\section{Conclusion}\label{Sec:Conclusion}

We presented two different layers to treat the class of indefinite nested sums defined over nested products in the context of symbolic summation. First, the term algebra layer $\Sum\Prod(\GG)$ (covering the rational case $\GG=\rGG$, the multibasic case $\GG=\bGG$ and the mixed multibasic case $\GG=\mGG$) equipped with an evaluation function $\ev:\Sum\Prod(\GG)\times\NN\to\KK$ has been introduced. There the user can define, evaluate and manipulate the class of nested sums and products conveniently. In particular, we illustrated how this user interface is implemented within the summation package \texttt{Sigma}.

Second, the formal difference ring/field layer has been elaborated. Here the elements of $\Sum\Prod(\GG)$ are rephrased in a ring $\EE$ that is built by (Laurent) polynomial ring extensions. More precisely, the adjoined variables (in some instances factored out by ideals) represent the summation objects with two extra ingredients: a ring automorphism $\sigma:\EE\to\EE$ that describes the action of the shift operator on the ring elements and an evaluation function $\ev:\EE\times\NN\to\KK$ that allows one to evaluate the formal ring elements to sequences. In this formal setting one can not only develop and implement complicated summation algorithms, but also set up a summation theory that enables one to embed the formal ring extensions into the ring of sequences (see Theorem~\ref{Thm:InjectiveProp}).

One of the secrets of \texttt{Sigma}'s success within, e.g., particle physics, combinatorics and number theory is the smooth interaction between these two different layers: as illustrated in Figure~\ref{Fig:UserInterface} on page~\pageref{Fig:UserInterface} one can represent the objects from the two worlds so that their interpretation with the corresponding evaluation function agrees. 
In this article, we worked out in detail this algorithmic translation back and forth between the user-friendly term algebra and the complicated difference ring setting. To gain a better understanding of \texttt{Sigma}'s capabilities we established a precise input-output specification of the available summation tools using the introduced term algebra language. Special emphasis has been put on the canonical form representation (and its relation to the difference ring theory) for the class $\Sum\Prod(\GG)$.

\bigskip

\noindent\textbf{Acknowledgments.} This article got inspired during the workshop \textit{Antidifferentiation and the Calculation of Feynman Amplitudes} at Zeuthen, Germany. In this regard, I would like to thank DESY for its hospitality 
and the Wolfgang Pauli Centre (WPC) for the financial support. In particular, I would like to thank the referee for the very thorough reading and valuable suggestions which helped me to improve the presentation.

\small

\end{document}